\documentclass[final, 10pt, draftcls]{IEEEtran}
\usepackage[T1]{fontenc}
\usepackage{url}
\usepackage{doi}
\usepackage{bm}
\usepackage{multirow}
\ifCLASSINFOpdf
\else
\fi

\usepackage{amsmath}
\usepackage{algorithmicx}
\usepackage{algorithm}
\usepackage{graphicx}
\usepackage{mathtools}
\usepackage{amsthm}
\usepackage{amssymb}
\usepackage[nocompress]{cite}
\usepackage{algpseudocode}
\usepackage{nccmath,textcomp}
\usepackage{eso-pic}
\usepackage{bbm}
\usepackage{xcolor}
\theoremstyle{remark}
\newtheorem{theorem}{Theorem}
\newtheorem{lemma}{Lemma}
\newtheorem{remark}{Remark}
\newtheorem{property}{Property}
\newtheorem{comment}{Comment}
\newtheorem{definition}{Definition}
\newtheorem{proposition}{Proposition}
\newtheorem{assumption}{Assumption}
\newtheorem{corollary}{Corollary}
\newtheorem*{proof-mine}{Proof}
\usepackage{verbatim}

\makeatletter
\let\MYcaption\@makecaption
\makeatother
\usepackage[font=footnotesize]{subcaption}

\makeatletter
\let\@makecaption\MYcaption
\makeatother

\begin{document}
	\title{Iterative RNDOP-Optimal Anchor Placement for Beyond Convex Hull ToA-based Localization: Performance Bounds and Heuristic Algorithms}
	\author{Raghunandan M. Rao,~\IEEEmembership{Member,~IEEE}, Don-Roberts Emenonye,~\IEEEmembership{Student Member,~IEEE} \\ 
		\thanks{Manuscript received xxxx; revised xxxx; accepted xxxx. Date of publication xxxx; date of current version xxxx. The associate editor
			coordinating the review of this article and approving it for publication was xxxx. \textit{(Corresponding author: Raghunandan M. Rao.)}}
		\thanks{R. M. Rao is currently at Amazon Lab 126, Sunnyvale, CA, USA, 94086. This work was undertaken before joining Amazon in his capacity as an independent scholar and does not relate to his current role at Amazon. All findings and opinions in the paper are the author's own and do not reflect the opinions of his employer (e-mail: raghumr@vt.edu).}
		\thanks{D.- R. Emenonye is affiliated with the Wireless@VT research group, Bradley Department of ECE, Virginia Tech, Blacksburg, VA, USA 24061 (email: donroberts@vt.edu).}
	}
	\maketitle
	\thispagestyle{empty}
	\begin{abstract}
		Localizing targets outside the anchors' convex hull is an understudied but prevalent scenario in vehicle-centric, UAV-based, and self-localization applications. Considering such scenarios, this paper studies the optimal anchor placement problem for Time-of-Arrival (ToA)-based localization schemes such that the \emph{worst-case Dilution of Precision (DOP)} is minimized. Building on prior results on DOP scaling laws for beyond convex hull ToA-based localization, we propose a novel metric termed the \textit{Range-Normalized DOP (RNDOP)}. We show that the worst-case DOP-optimal anchor placement problem simplifies to a min-max RNDOP-optimal anchor placement problem. Unfortunately, this formulation results in a non-convex and intractable problem under realistic constraints. To overcome this, we propose iterative anchor addition schemes, which result in a tractable albeit non-convex problem. By exploiting the structure arising from the resultant rank-1 update, we devise three heuristic schemes with varying performance-complexity tradeoffs. In addition, we also derive the upper and lower bounds for scenarios where we are placing anchors to optimize the worst-case (a) 3D positioning error and (b) 2D positioning error. We build on these results to design a cohesive iterative algorithmic framework for robust anchor placement, characterize the impact of anchor position uncertainty, and then discuss the computational complexity of the proposed schemes. Using numerical results, we validate the accuracy of our theoretical results. We also present comprehensive Monte-Carlo simulation results to compare the positioning error and execution time performance of each iterative scheme, discuss the tradeoffs, and provide valuable system design insights for beyond convex hull localization scenarios.
	\end{abstract}
	
	\IEEEpeerreviewmaketitle
	\section{Introduction}\label{sec:intro}
	Localization applications are pervasive in the 5G and beyond-5G era, especially in consumer electronics, industrial, and defense sectors. In part due to the success of geolocation GPS and GNSS, consumers of wireless technology have increasing expectations regarding the reliability and accuracy of the location estimates in unprecedented deployments. Vehicle-centric localization is one such scenario that has been actively pursued by auto manufacturers and the wireless industry \cite{FiRa_Consortium}, wherein the target locations are estimated relative to the vehicle. This is in contrast with localizing targets on a local/global map, which is the typical approach in currently deployed positioning systems. 
	
	A vehicle-centric approach is crucial in applications such as autonomous driving \cite{Bresson_TIV_SLAMSurvey_2017}, vehicle safety, keyless locking/unlocking \cite{FiRa_Consortium}, smart parking \cite{Bus_Docking_VTM_2021}, etc. It is also useful for unmanned aerial vehicle (UAV)-based applications such as distributed UAV swarm coordination \cite{Liu_Distr_3D_Rel_Loc_TVT_2020}, autonomous takeoff and landing, and disaster assistance \cite{First_Respond_UAV_ComMag_2021}. To enable autonomous vehicular operation, multiple anchors are deployed on the vehicle to track targets that are typically located outside the vehicle \cite{Sadhegi_Optim_Anch_Plc_2D3D_TSP_2020, Rao_SRA_GDOP_Scale_Dimn_Adapt_2021}. In other words, the targets are outside the convex hull formed by the anchor locations (henceforth referred to as a \textit{`beyond convex hull' localization scenario}) \cite{Rao_SRA_GDOP_Scale_Dimn_Adapt_2021}. Such conditions also arise in UAV-based localization applications for swarm coordination (Fig. \ref{Fig1_Anchors_FarAway_Regime}), and self-localization where-in a new anchor needs to be localized relative to existing anchors in a wireless sensor network.
	
	The target's location can be estimated using measurements such as Time-of-Arrival (ToA) \cite{Rao_SRA_GDOP_Scale_Dimn_Adapt_2021, Xu_Optim_TOA_COMML_2021}, Time-Difference-of-Arrival (TDoA) \cite{Wang_TDOA_Approx_TSP_2016}, received signal strength (RSS) \cite{Tomic_RSS_Loc_TVT_2015}, and Angle-of-Arrival (AoA) typically fused with ToA or TDoA. Among these methods, ToA-based and TDoA-based localization provides robust range error performance in high multipath scenarios that accompany typical Non-Line of Sight (NLoS) wireless propagation environments. The state-of-the-art ranging methods using low-cost technologies such as ultrawideband (UWB) \cite{Dardari_Win_UWB_Ranging_ProcIEEE_2009}, which use high bandwidth, low power pulses that is capable of precise location estimates for moderate distances between the anchor and the target. If all the anchors are adequately time-synchronized, ToA-based algorithms have superior positioning performance when compared to TDoA-based schemes when the target is outside the anchors' convex hull\footnote{Intuitively, this is because locus of the target in ToA is a circle that has a bounded region, whereas that in the case of TDoA is a hyperbola, \textit{an unbounded region} \cite{buehrer2019handbook}}. In this paper, we seek to optimize anchor locations to minimize the worst-case positioning error in vehicle-centric ToA-based localization scenarios, where the target is located beyond the anchors' convex hull.
	
	\vspace{-10pt}
	
		\begin{table*}[t]
		\footnotesize
		\caption{Comparison with Relevant Prior Works That Study Localization of Targets Outside the Convex Hull}		
		\label{tab:relevant_works}
		\centering 
		
			\begin{tabular}{|c | c | c | c | c | c |}
			\hline
			\textbf{Reference} & \textbf{Localization}  & \textbf{Targets outside} & \textbf{Notable Constraints}& \textbf{Metric} & \textbf{Relevant Details} \\
			& \textbf{Method} &  \textbf{Convex Hull} &  & & \\
			\hline
			\cite{RSS_UAV_Loc_TVT_2022} & RSS & Yes & Terrestrial target, UAV anchors & Average MSE & Studies MSE lower bound \\
			\hline
				\cite{Nasir_Alouini_TCOM_2019} & Optical & Yes & Anchors on water surface & D-optimality of Fisher & Incorporates outlier\\
			& Ranging &  & and underwater targets & Information Matrix (FIM) & removal methods \\
			\hline 
			\cite{Monica_Ferrari_TAES_2015} & TDOA & Yes & Target moves along & CRLB & Derives optimal anchor \\
			& & & corridor in & minimization & separation, as a function\\
			& & & straight line & & of corridor dimensions. \\
			\hline 
			\cite{Lazzari_UWB_UAV_Sens_2017} & ToA & Yes & Anchors mounted on roof & Mean Absolute & Empirically evaluates 3\\
			& & & roof of vehicle and UAV target & Positioning Error & anchor configurations\\
			\hline
			\cite{Sadhegi_Optim_Anch_Plc_2D3D_TSP_2020} & ToA & Yes & Anchors lie on & A-optimality of & Closed-form expressions\\
			& & & circle & FIM & for 2D positioning in 2D space\\
			\hline
		\end{tabular}\\
	[-3ex]
	\end{table*}
	
	\subsection{Related Work}
	Several prior works have studied anchor placement optimization techniques to minimize positioning error-related metrics. The localization error is a function of the measurement error, and the anchor locations relative to the target. The latter is often quantified by a dimensionless metric termed as the \textit{Dilution of Precision (DOP)} \cite{Torrieri_Stat_Lclze_TAES_1984}. Therefore, localization accuracy can be improved using (a) robust protocols and signal processing algorithms to reduce measurement error \cite{Domuta_Palade_TVT_2021}, and (b) by geometrically placing anchors such that the overall impact of measurement error on localization error is reduced.
	
	Various works study the optimal anchor placement problem relative to a single anchor \cite{Van_de_Velde_FramTheor_VTCS_2014, Sheng_Optim_Anch_Plcmnt_TSP_2019,sahu2021optimal,Fang_Optim_AoA_Geom_COMML_2022,Panwar_Optim_TOA_RSS_AOA_TAES_2022,Nafiseh_Frame_Theor_Anch_Plcmnt_TAES_2022,Chepuri_Sparse_Anch_Plcmnt_EUSIPCO_2013}. The optimal anchor placement for TOA-based single static target localization is derived in \cite{Van_de_Velde_FramTheor_VTCS_2014} for the 2D case, and \cite{Sheng_Optim_Anch_Plcmnt_TSP_2019} for the 3D case. Sahu at al. \cite{sahu2021optimal} propose a unified mathematical framework based on alternating direction method of multipliers (ADMM) to design the optimal anchor geometry for ToA, AOA, RSS-based localization of a single target. Sheng et al. \cite{Xu_Optim_TOA_COMML_2021} analyze the optimal anchor placement problem for localizing two targets on the 2D plane using a common set of anchors. Fang et al. \cite{Fang_Optim_AoA_Geom_COMML_2022} study the optimal sensor placement problem for AOA-based localization where a single target and multiple anchors are constrained to lie within a circular region, while ensuring a minimum separation between all nodes. Using the notion of D-optimality, they formulate a constrained optimization problem where the constraints are related to the maximum feasible
    angle and the optimal separation angle. Panwar et al. \cite{Panwar_Optim_TOA_RSS_AOA_TAES_2022} study the optimal sensor placement for hybrid TOA-RSS-AOA localization scenarios, and propose a solution based on the majorization-minimization (MM) technique applied on the hybrid TOA-RSS-AoA CRLB.
	
	However, for localizing multiple mobile targets, minimizing the positioning error for all targets at all possible locations in the coverage region is impossible. Therefore, state-of-the art works study, design, and characterize the performance of algorithms to minimize the \textit{weighted average of the localization error of targets in a region of interest}. However, most of the prior works \cite{Sheng_Optim_Anch_Plcmnt_TSP_2019, Xu_Optim_TOA_COMML_2021, Nafiseh_Frame_Theor_Anch_Plcmnt_TAES_2022, Chepuri_Sparse_Anch_Plcmnt_EUSIPCO_2013} focus on studying localization scenarios where the target lies within the convex hull of anchors. Such geometries are favorable because the resulting anchor geometry has been rigorously proven to be favorable for placing anchors in general \cite{OLone_Conv_Hull_PLANS_2016,Schloemann_Dhillon_WCL_2016}. 
	
	A few works have studied the optimal placement problem for scenarios where the target lies outside the convex hull, which are summarized in Table \ref{tab:relevant_works}. Zhou et al. in \cite{RSS_UAV_Loc_TVT_2022} consider RSS-based positioning of a single terrestrial target located uniformly at random in a circular region using multiple UAVs, and propose an average MSE-optimal anchor placement scheme. Their results show that conditioned on an anchor configuration, the MSE scales linearly with the radius of the coverage region. Saeed et al. \cite{Nasir_Alouini_TCOM_2019} propose a robust 3D localization method for a partially connected underwater optical sensor network, where there is a significant possibility of missing measurements. The anchors are located near the water surface, and localize the underwater sensor locations using range measurements. Using the D-optimality criterion and by mitigating missing measurements using outlier removal schemes, they design the optimal depth for the anchors. Monica et al. \cite{Monica_Ferrari_TAES_2015} study the optimal anchor placement problem for indoor localization of ground targets moving in a corridor, wherein the anchors are placed on the ceiling. They use TDOA measurements, and derive the average MSE as a function of the anchor separation and the corridor dimensions. Finally, the find the optimal separation as a function of the range error statistics and the corridor dimensions. Lazzari et al. \cite{Lazzari_UWB_UAV_Sens_2017} develop a ToA-based technique to localize a far-away mobile UAV target using terrestrial vehicle-mounted anchors. They empirically investigate the mean absolute positioning error of three anchor configurations, in order to choose the configuration with the lowest mean absolute error. Sadeghi et al. \cite{Sadhegi_Optim_Anch_Plc_2D3D_TSP_2020} study the optimal anchor placement for localizing a target outside the anchor region in \textit{2D space}. The anchors lie on a circular area, while the target lies outside the circular region. Using the A-optimality criteria as the objective, they compute the optimal anchor locations as a \textit{function of the target location} and the range error statistics at each anchor.
	
	\vspace{-12pt}
	\subsection{Contributions}
	In various applications, UAV-based systems are often autonomous, tracking mobile aerial vehicles \textit{in unconstrained 3D space}. Hence, they need to have the capability to localize multiple targets outside the convex hull with low positioning error. In contrast to the prior works, our problem formulation imposes volumetric constraints on the anchors, imposes no constraints on the far-away target, and is specifically designed such that (a) we can study 3D-optimal as well as 2D-optimal anchor placement schemes (for localization in 3D space), (b) the \textit{optimal anchor configuration is independent of the target's distance from the anchors, once it is sufficiently far away}, and (c) the anchor placement is robust to the target's relative bearing w.r.t. the anchors. We make the following contributions in this work.
	
	\subsubsection*{Range-Normalized DOP (RNDOP)} When the target lies outside the anchors' convex hull, we showed in our prior work \cite{Rao_SRA_GDOP_Scale_Dimn_Adapt_2021} that the DOP ($\mathsf{D}$) scales linearly with the distance from the anchors' centroid ($r_t$). Firstly, we extend these results by deriving closed-form expressions of the 2D-DOP (for XY localization in 3D space), as a function of the azimuth angle, and its upper and lower bounds. Furthermore, we propose a novel metric termed \textit{Range-Normalized DOP} ($\mathsf{R}$). Using the linear DOP scaling results and the definition of RNDOP, we demonstrate that a weighted DOP optimization problem can be simplified to a weighted RNDOP optimization problem, where the weighting function is a function of the azimuth and elevation angles.  
	
	\subsubsection*{Min-Max DOP-optimal Iterative Anchor Placement Schemes for 2D and 3D Positioning} 
	Using the asymptotic equivalence between the weighted DOP and weighted RNDOP optimization problems, we show that a robust DOP-optimal anchor configuration is almost invariant to the targets' distance as long as it is far away outside the anchors' convex hull. We then formulate a min-max RNDOP-optimal anchor placement problem under practical conditions such as (a) minimum anchor separation constraints, that are often imposed to ensure adequate device isolation, UAV safety, etc., and (b) box constraints, that approximately model the shape of terrestrial vehicles (such as trucks/buses), and the airspace above tall buildings. Unfortunately, the resultant problem is non-convex and intractable. To overcome this, we propose an iterative scheme, where a single anchor is added to the existing set of anchors at each iteration to minimize the maximum RNDOP. Observing the structure of the recurrence relation, we propose two additional heuristic schemes: one based on the (a) trace, and (b) eigenvector of a \emph{matrix-valued function solely composed of the anchor coordinates}. In the process, we derive the upper and lower bounds on (a) the achievable RNDOP at each iteration, and (b) the universal RNDOP bounds. Then, we combine these schemes into a cohesive iterative algorithmic framework to obtain robust 2D and 3D RNDOP-optimal anchor configurations. Finally, we derive theoretical results to understand the impact of anchor position uncertainty on the maximum RNDOP.
	
	\subsubsection*{Complexity Analysis and Numerical Results}
	We analyze the worst-case computational complexity of each proposed variant, and show that for the Min-Max RNDOP and Trace-based heuristic schemes, the worst-case complexity is linear in the number of additional anchors to be added ($N_{\rm a}$). We then use numerical results to validate the derived theoretical results, and develop a Monte-Carlo simulator to compare the positioning error and execution time performance of the proposed schemes across different computer configurations. We also provide useful insights by discussing the \emph{localization error-execution time tradeoff} for all the proposed schemes, which can aid the design of practical \emph{beyond convex hull} localization systems. 
	
	\subsubsection*{Notation}
	$\bm{X}$ is a matrix, and $\bm{x}$ a column vector. $\| \bm{x} \|_2$ and $[\bm{x}]_i$ represent the 2-norm and $i^{\text{th}}$ element of $\bm{x}$. For a square matrix $\bm{X}$, ${\rm tr}(\bm{X})$, $\bm{X}^{-1}$ and $\bm{X}^{\frac{1}{2}}$ represent the trace, inverse,  and matrix square-root. The ordered eigenvalues of $\bm{X}$ are denoted by $\lambda_1 (\bm{X}) \leq \lambda_2 (\bm{X}) \leq \cdots \leq \lambda_n (\bm{X})$, where $\lambda_i (\bm{X})$ is the $i^{th}$ lowest eigenvalue, and its corresponding eigenvector is $\bm{v}_i(\bm{X})$.  $\lambda_+(\bm{X})$ and $\lambda_{-}(\bm{X})$ represent the maximum and minimum eigenvalue, and $\bm{v}_+(\bm{X})$ and $\bm{v}_{-}(\bm{X})$ represent their corresponding eigenvectors. Similarly, for symmetric matrices $\bm{X}$ and $\bm{Y}$, $\lambda_+(\bm{X}, \bm{Y})$ and $\lambda_{-}(\bm{X}, \bm{Y})$ denotes the maximum and minimum generalized eigenvalues of $(\bm{X}, \bm{Y})$.  $[\bm{X}]_{i,j}$, $[\bm{X}]_{i,:}$, $[\bm{X}]_{:,j}$ and $[\bm{X}]_{i:m,j:n}$ represent $(i,j)^{\text{th}}$ element, $i^{th}$ row, $j^\text{th}$ column, and the sub-matrix formed by choosing the row indices $i:m$  and column indices $j:n$. $\bm{I}$ represents the identity matrix. $\bm{I}_{m \times n}$ has the rows and columns of an $m \times m$ identity matrix, concatenated with $n - m$ columns of all zeros. The set of real-valued $n \times n$ matrices and $n \times 1$ vectors are represented by $\mathbb{R}^{n \times n}$ and $\mathbb{R}^{n \times 1}$ respectively, and that of $n \times n$ symmetric matrices are represented by $\mathbb{S}^{n \times n}$. The set of positive definite and positive-semidefinite $n \times n$ matrices are represented by $\mathbb{S}^{n \times n}_{+}$ and $\mathbb{S}^{n \times n}_{++}$ respectively. $X \sim \mathcal{U}[\mathcal{X}]$ represents a random variable $X$ whose value is drawn uniformly at random from the set $\mathcal{X}$.
	\vspace{-5pt}
	\section{Mathematical Preliminaries}\label{sec:math_prelims}
	In this section, we mention a few useful results from matrix theory, which are crucial for designing and characterizing the performance of algorithms presented in this paper. 
	\begin{proposition}\label{prop:Sherman_Morrison_formula}
		(Sherman-Morrison formula \cite{golub2012matrix}) Let $\bm{X} \in \mathbb{R}^{n \times n}$ be a non-singular matrix, and $\bm{y}, \bm{z} \in \mathbb{R}^{n}$ such that $\bm{z}^T \bm{X}^{-1} \bm{y} \neq -1$. Then, the inverse of the rank-1 update $\bm{Y} = \bm{X} + \bm{yz}^T$ is given by $\bm{Y}^{-1} = \bm{X}^{-1} - \tfrac{\bm{X}^{-1} \bm{yz}^T \bm{X}^{-1}}{1 + \bm{z}^T \bm{X}^{-1} \bm{y}}$.
		
	\end{proposition}
	
	\begin{proposition}\label{prop:general_eigvals_vecs}
		(Generalized eigenvalues/eigenvectors \cite{boyd2004convex}) Suppose $\bm{X} \in \mathbb{S}^{n \times n}$, and $\bm{Y} \in \mathbb{S}^{n \times n}_+$. Then, for a non-zero vector $\bm{w} \in \mathbb{R}^n$, we have 
		\begin{align*}
			\lambda_{-}(\bm{Z}) &= \lambda_{-}(\bm{X}, \bm{Y}) \leq \tfrac{\bm{w}^T \bm{X w}}{\bm{w}^T \bm{Y w}} \leq \lambda_{+}(\bm{X}, \bm{Y}) = \lambda_{+}(\bm{Z}),
		\end{align*}
		where $\bm{Z} = \bm{Y}^{-\frac{1}{2}} \bm{X Y}^{-\frac{1}{2}}$. Furthermore, the upper bound is achieved when $\bm{w} = \alpha \bm{v}_{+}(\bm{Z})$ and the lower bound when $\bm{w} = \alpha \bm{v}_{-}(\bm{Z})$, where $\alpha \in \mathbb{R} \setminus \{0\}$.
	\end{proposition}
	
	\begin{proposition}\label{prop:eigvalue_interlace_thm}
		(Eigenvalue bounds due to rank-1 perturbation \cite{golub2012matrix}, \cite{cheng2012bounds}) Let $\bm{Y} = \bm{X} + \epsilon \bm{ww}^T$, where $\bm{X},\bm{Y} \in \mathbb{R}^{n\times n}$, $\bm{w} \in \mathbb{R}^n$ such that $\| \bm{w} \|_2 = 1$, and $\epsilon \in \mathbb{R}$. If $\epsilon < 0$, we have
		\begin{align}
			\label{eq:eig_interlace_epsilon_neg}
			\lambda_1(\bm{X}) -  \epsilon \leq & \lambda_1 (\bm{Y}) \leq \lambda_1 (\bm{X}) \leq \lambda_2 (\bm{Y}) \leq \lambda_2 (\bm{X}) \leq \cdots \nonumber \\
			& \leq \lambda_n (\bm{Y}) \leq \lambda_n (\bm{X}).
		\end{align}
		On the other hand, if $\epsilon > 0$ then we have
		\begin{align}
			\label{eq:eig_interlace_epsilon_pos}
			\lambda_1(\bm{X}) \leq & \lambda_1 (\bm{Y}) \leq \lambda_2 (\bm{X}) \leq \lambda_2 (\bm{Y}) \leq \lambda_3 (\bm{X}) \leq \cdots \nonumber \\
			& \leq \lambda_{n} (\bm{X}) \leq \lambda_n (\bm{Y}) \leq \lambda_n (\bm{X}) + \epsilon.
		\end{align}    
		In both cases, there exists $d_i$ for $ i=1,2,\cdots,n$ such that
		\begin{align}
			\label{eq:eigY_eq_eigX_plus_eps_frac}
			\lambda_i(\bm{Y}) = \lambda_i (\bm{X}) + d_i \epsilon, i = 1,2,\cdots, n,		
		\end{align}
		where $0 \leq d_i \leq 1$ and $\sum_{i=1}^n d_i = 1$.
	\end{proposition}
	\begin{proposition}\label{proposition:sum_of_top_2_eigvals_optim}
		Let $\bm{X_{3}} \in \mathbb{S}^{3 \times 3}_{++}$ and $\bm{X_{2}} \in \mathbb{S}^{2 \times 2}_{++}$. The solution to the optimization problems 
		\begin{align}
			\label{eq:sum_of_top_2_eigvals_optim_X3}
			\Lambda^*_3 = \underset{\bm{X_{3}}}{\min} & \ {\rm tr}(\bm{X}^{-1}_{\bm{3}}) - \lambda_{-}(\bm{X}^{-1}_{\bm{3}})\ {\rm s.t.\ } {\rm tr}(\bm{X_{3}}) = K,  \\ 
			\label{eq:sum_of_top_2_eigvals_optim_X2}
			\Lambda^*_2 = \underset{\bm{X_{2}}}{\min} & \ {\rm tr}(\bm{X}^{-1}_{\bm{2}}) - \lambda_{-}(\bm{X}^{-1}_{\bm{2}})\ {\rm s.t.\ } {\rm tr}(\bm{X_{2}}) = K,
		\end{align}
		are given by $\Lambda^*_3 = \frac{6}{{\rm tr}(\bm{X_{3}})} = \frac{6}{K}$ and $\Lambda^*_3 = \frac{2}{{\rm tr}(\bm{X_{2}})} = \frac{2}{K}$.
	\end{proposition}
	
	\begin{comment} 
	\begin{corollary}\label{proposition:sum_of_top_1_eigvals_optim}
		Let $\bm{X_{2}} \in \mathbb{S}^{2 \times 2}_{++}$. The solution to the optimization problem 
		\begin{align}
			\label{eq:sum_of_top_1_eigvals_optim}
			\Lambda^* = \underset{\bm{\bm{X_{2}}}}{\min} & \ {\rm tr}(\bm{X}^{-1}_{\bm{2}}) - \lambda_{-}(\bm{X}^{-1}_{\bm{2}}) \\
			{\rm s.t.\ } & {\rm tr}(\bm{X_{2}}) = K \nonumber
		\end{align}
		is given by $\Lambda^* = \frac{4}{K}$, which is achieved when its eigenvalues $\lambda_i(\bm{X_{2}}) = \frac{K}{2}$ for $i=1,2$.
	\end{corollary}
	\begin{proof}
		The proof is similar to Proposition \ref{proposition:sum_of_top_2_eigvals_optim}, and is omitted due to space constraints. 
	\end{proof}
	\end{comment}

	\section{System Model} \label{sec:sys_model}
	\subsection{Ranging Error Model}
	We consider a network of $N$ anchors which are used to determine the unknown position of a target in 3D space. All the anchor locations are assumed to be known, and the $i^{th}$ anchor is located at $\bm{r_i} = [x_i\ y_i\ z_i]^T, i=1,2,\cdots,N$. The target's unknown location is at $\bm{r_t}=[x_t\ y_t\ z_t]^T$. The $i^{th}$ anchor measures the time-of-arrival (ToA) $\tau_i$ from the target (also referred to as `tag') using two-way ranging mechanisms \cite{Domuta_Palade_TVT_2021}, which is given by \cite{Mazuelas_Win_SRI_TSP_2018}
	\begin{align}
		\label{eq:range_error_model}
		\tau_i = \frac{1}{c}(\| \bm{r_i} - \bm{r_t}  \|_2 + w_i), \text{ for } i = 1,2,\cdots,N,
	\end{align}
	where $c$ is the velocity of light, $w_i \sim \mathcal{N}(b_i,\sigma^2_{w_i})$ is the ranging error, $\sigma^2_{w_i}$ the range error variance, and $b_i > 0$ is the range bias due to multipath \cite{Dardari_Win_UWB_Ranging_ProcIEEE_2009, Mazuelas_Win_SRI_TSP_2018}.
	\vspace{-10pt}
	\subsection{Localization Error and the Dilution of Precision (DOP)}
	For range-based localization schemes, the overall positioning error depends on the statistics of the ranging errors ($w_i, i=1,2,\cdots,N$), as well as the anchor coordinates relative to the target. The Cram{\'e}r-Rao Lower Bound (CRLB) for the localization error\footnote{This quantity is often also referred to as the Positioning Error Bound (PEB) in the literature.} under (\ref{eq:range_error_model}) is given by \cite{Torrieri_Stat_Lclze_TAES_1984}
	\begin{align}
		\label{eq:CRLB_for_TOA_under_Gaussian_error}
		{\rm Var}([\bm{r_t} - \bm{\hat{r}_t}]_i) & \geq \big[ \big(\bm{H}^T \bm{W}^{-1} \bm{H} \big)^{-1} \big]_{i,i}, \text{where } \\
		\label{eq:matrix_H}
		\bm{H}^T & = \big[\tfrac{\bm{r_t} - \bm{r_1}}{\| \bm{r_t} - \bm{r_1} \|_2}\ \  \tfrac{\bm{r_t} - \bm{r_2}}{\| \bm{r_t} - \bm{r_2} \|_2} \cdots \tfrac{\bm{r_t} - \bm{r_N}}{\| \bm{r_t} - \bm{r_N} \|_2} \big], \text{and }\\ 
		\bm{W} & = \mathbb{E}[\bm{ww}^T], \text{ where } \bm{w} = [w_1\ w_2\ \cdots w_N]^T.  
	\end{align}
	In real-world scenarios, multipath scattering in NLoS conditions is the primary source of range bias $b_i$ \cite{Dardari_Win_UWB_Ranging_ProcIEEE_2009}, and the received signal SNR is the primary source of the unbiased range error ($w_i - b_i$) \cite{Gezici_Poor_UWB_Posn_ProcIEEE_2009}. In this paper, we focus on ToA-based localization of far-away targets outside the convex hull of the anchor locations (henceforth referred to as `beyond convex hull' localization scenarios) \cite{Rao_SRA_GDOP_Scale_Dimn_Adapt_2021}. 
	When the target is very far from the anchors, the received signal's SNR at all anchors tends to be the same. Furthermore, the multipath scattering geometry also tends to be similar, resulting in the same range bias at all anchors\footnote{Indeed, this model is exceedingly used in the optimal anchor placement literature \cite{Nguyen_Optim_ToA_Multistatic_TSP_2016, Schloemann_Dhillon_WCL_2016, Sheng_Optim_Anch_Plcmnt_TSP_2019}. In addition, this model can be simplified further by ignoring the range bias in LoS conditions \cite{Dardari_Win_UWB_Ranging_ProcIEEE_2009, Gezici_Poor_UWB_Posn_ProcIEEE_2009, Alsindi_RangeErr_Meas_UWB_Ind_TVT_2009}.}.
	Hence, for far-away targets, we assume that 
	\begin{itemize}
		\item the bias of the range error on each anchor is the same, i.e. $b_i = b\ \forall i=1,2,\cdots,N$, and
		\item the unbiased range error $(w_i - b_i)$ on each anchor are i.i.d. normal random variables.
	\end{itemize}
	Under these assumptions, we get $\bm{W} = (b^2 + \sigma^2_w) \bm{I}_{\bm{N}}$ and the CRLB expression in (\ref{eq:CRLB_for_TOA_under_Gaussian_error}) simplifies to 
	\begin{align}
		\label{eq:CRLB_under_common_error_stats}
		{\rm Var}([\bm{r_t} - \bm{\hat{r}_t}]_i) & \geq (b^2 + \sigma^2_w)\big[ \big(\bm{H}^T \bm{H} \big)^{-1} \big]_{i,i}.
	\end{align}
	Thus, the positioning error bound (PEB) is given by
	\begin{align}
		\label{eq:PEB_same_error_stats_3D_loc}
		\sqrt{{\rm Var}(\| \bm{r_t} - \bm{\hat{r}_t}\|_2)} & \geq \sqrt{(b^2 + \sigma^2_w) \times \text{tr}([\bm{H}^T \bm{H}]^{-1}) }.
	\end{align}
	Thus, the PEB is dependent on (a) the statistics of range error $(b, \sigma_w)$, and the (b) geometrical placement of anchors relative to the target $(\bm{H})$. The impact of anchor placement on the target's positioning error is quantified by the metric termed as \textit{Dilution of Precision (DOP)} \cite{Torrieri_Stat_Lclze_TAES_1984}. It is defined in terms of $\bm{Q} = [\bm{H}^T \bm{H}]^{-1}$ below for 3D ($\mathsf{D_{xyz}}$), and 2D ($\mathsf{D_{xy}}$) localization.
	\begin{align}
		\label{eq:DOP_xyz_defn}
		\mathsf{D_{xyz}} & = \sqrt{{\rm tr}(\bm{Q})}, \mathsf{D_{xy}} = \sqrt{{\rm tr}([\bm{Q}]_{1:2,1:2})}.
	\end{align}
	As seen in (\ref{eq:matrix_H}), the columns of $\bm{H}^T$ are composed of dimensionless unit vectors. Hence, the DOP is a \textit{dimensionless metric}. In additon, it satisfies the following properties.
	\begin{property} \label{prop:trans_Inv}
		(Translational Invariance) The DOP does not change when all the anchors and the target are translated by a vector $\bm{\Delta r}$. 
	\end{property}
	
	\begin{property} \label{prop:scale_inv}
		(Scale Invariance) If the anchor and target coordinates are scaled by a factor of $k$ ($k \neq 0$), then the DOP remains unchanged.
	\end{property}
	Using the definition of $\bm{H}$ in (\ref{eq:matrix_H}), Property \ref{prop:trans_Inv} follows from the fact that the entries of $\bm{H}$ depend on the \textit{vector difference between} the target and the anchor coordinates. On the other hand, Property \ref{prop:scale_inv} follows since the entries in $\bm{H}$ are normalized ratios of the anchor positions relative to the target. 
	\vspace{-12pt}
	\subsection{DOP and Range-Normalized DOP (RNDOP) in ToA-based Positioning of Far-Away Targets}\label{subsec:DOP_and_RNDOP_laws}
	Beyond convex hull localization is motivated by several practical use cases, such as (a) using vehicular sensors to track outdoor targets \cite{Bresson_TIV_SLAMSurvey_2017, Rao_SRA_GDOP_Scale_Dimn_Adapt_2021}, (b) localization of UAVs using terrestrial vehicle-mounted anchors \cite{Lazzari_UWB_UAV_Sens_2017}, and (c) self-localization scenarios. Formally, the target is considered to be far-away if the following conditions are met, for $i,j=1,2,\cdots,N$ and $i \neq j$:
	\begin{enumerate}
		\item $\| \bm{r_i} - \bm{r_t} \|_2 \approx \| \bm{r_j} - \bm{r_t} \|_2 \approx \| \bm{r_t}\|_2$, and
		\item $\| \bm{r_i} - \bm{r_j} \|_2 \ll \| \bm{r_i} - \bm{r_t} \|_2$.
	\end{enumerate}	
	While the notion of \textit{`approximate equality'} in the above description is rather arbitrary, we will provide a formal working definition of the `far-away target' scenario later in this section. Due to the translational invariance of the DOP, we have the following modeling assumption regarding the origin of the coordinate system.
	\begin{assumption}\label{assump:anchor_centroid_at_origin}
		The origin of the local coordinate system (LCS) is located at the centroid of the anchor locations, i.e. $\sum_{i=1}^N \bm{r_i} = \bm{0}$. 
	\end{assumption}
	
	\begin{figure}[t]
		\centering
		\includegraphics[width=2.9in]{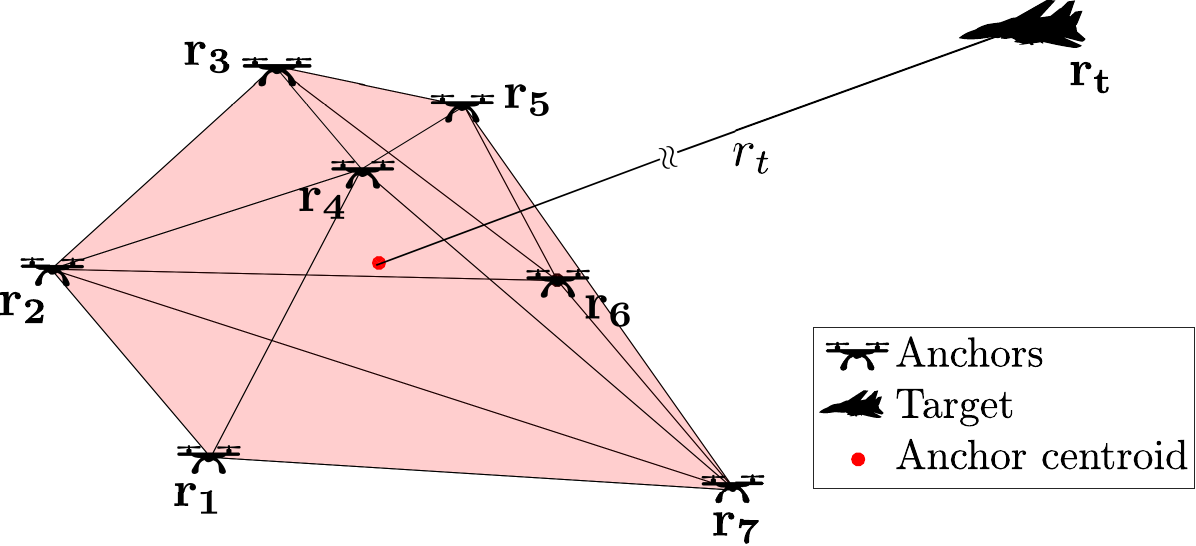}\\
		[-2ex]
		\caption{3D ToA-based localization of far-away targets. The anchors are mounted on UAVs, and the $i^{\rm th}$ anchor is at $\bm{r_i}=(x_i,y_i,z_i), i=1,2,\cdots,N$ such that all anchors are confined in a relatively small volume enclosed by their convex hull, such that $\sum_{i=1}^N \bm{r_i}= \bm{0}$. The target is far-away from the origin such that $r_t \approx \| \bm{r_t} - \bm{r_i} \|_2$.} 
		\label{Fig1_Anchors_FarAway_Regime}
	\end{figure}
	
	For far-away targets outside the anchors' convex hull, we have the following proposition. 
	\begin{proposition}\label{prop:approx_matrix_rep_GDOP_matrix}
		In the far-away target regime, we have 
		\begin{align} 
			\bm{H}^T \bm{H} & \approx \frac{1}{r^2_t}(\bm{C} + N r^2_t \bm{a_t}(\theta, \phi)\bm{a}^T_{\bm{t}}(\theta, \phi)) \succcurlyeq \bm{0}, 
		\end{align}
		where $\bm{C} = \sum_{i=1}^N \bm{r_i r}^T_{\bm{i}}$, $\bm{r_t} = r_t \bm{a_t}(\theta, \phi)$, and $\bm{a_{t}}(\theta, \phi) = [\cos \theta \sin \phi\ \sin \theta \sin \phi\ \cos \phi]^T$.
	\end{proposition}
	\begin{proof} 
	    Please refer \cite[Proposition 1]{Rao_SRA_GDOP_Scale_Dimn_Adapt_2021}.	
	\end{proof}
	The above result can then be used to obtain a distinct \textit{`DOP scaling law'} for range-based localization in all dimensions, as shown in the following theorem.
	\begin{theorem} \label{thm:GDOP_closed_form}
		If $r_t \gg \frac{1}{\sqrt{N \lambda_{-}(\bm{C}^{-1})}}$, then $\bm{Q}$ scales linearly with $r_t$ for a fixed target direction $\bm{a_{t}} = \bm{r_{t}}/\| \bm{r_{t}} \|_2$, and is given by
		\begin{align}
			\label{eq:Q_simplify}
			\bm{Q} = r^2_t \Big(\bm{C}^{-1} - \tfrac{\bm{C}^{-1} \bm{a_{t}} \bm{a_{t}}^T \bm{C}^{-1}}{\bm{a_{t}}^T \bm{C}^{-1} \bm{a_{t}}} \Big). 
		\end{align}
	\end{theorem}
	\begin{proof}
	Please refer \cite[Theorem 1]{Rao_SRA_GDOP_Scale_Dimn_Adapt_2021}.
	\end{proof}
	Using the above result, we have the following working definition of the \textit{`far-away target'} for the purpose of this paper.
	\begin{definition}\label{defn:far_away_target_workingdefn}
		A target is far-away if $r_t \gg [N \lambda_{-} (\bm{C}^{-1})]^{-\frac{1}{2}}$.
	\end{definition}
	We are interested in (a) 3D localization, and (b) 2D localization of targets moving in the XY plane. In the following result, we derive the corresponding DOPs in closed form. 
	\begin{lemma}\label{lemma:equi_DOP_xyz_and_xy}
		In the far-away target regime, we have 
		\begin{align} 
			\label{eq:DOP_xyz_rt_scaling_concise}
			\mathsf{D_{xyz}} (\mathcal{R}_a, r_t, \theta, \phi) & = r_t \sqrt{{\rm tr}(\bm{D}) - \tfrac{\bm{a_{t}}^T(\theta, \phi) \bm{D}^2 \bm{a_{t}}(\theta, \phi)}{\bm{a_{t}}^T(\theta, \phi) \bm{D} \bm{a_{t}}(\theta, \phi)} },\\
			\label{eq:DOP_xy_rt_scaling_concise}
			\mathsf{D_{xy}}(\mathcal{R}_a, r_t, \theta) & = r_t \sqrt{{\rm tr}(\bm{E}) - \tfrac{\bm{b_{t}}^T(\theta) \bm{E}^2 \bm{b_{t}}(\theta) }{\bm{b_{t}}^T(\theta) \bm{E} \bm{b_t} (\theta)} },
		\end{align}
		where $\bm{D} \triangleq \bm{C}^{-1}, \bm{E} \triangleq \big[ \bm{C}^{-1} \big]_{1:2,1:2}$, and $\bm{b_{t}}(\theta) = [\cos \theta \ \sin \theta]^T$.
	\end{lemma}
	\begin{proof}
		Applying the definition of the DOP for 3D positioning in (\ref{eq:DOP_xyz_defn}) on (\ref{eq:Q_simplify}), we get 
		\begin{align}
			\label{eq:3D_DOP_scaling_law_expr}
			\mathsf{D_{xyz}} = r_t \sqrt{{\rm tr}\Big(\bm{C}^{-1} - \tfrac{\bm{C}^{-1} \bm{a_{t}} \bm{a_{t}}^T \bm{C}^{-1}}{\bm{a_{t}}^T \bm{C}^{-1} \bm{a_{t}}} \Big) }.
		\end{align}
		We get (\ref{eq:DOP_xyz_rt_scaling_concise}) by using the linearity and cyclic property of the trace, representing $\bm{a_{t}}(\theta, \phi)\triangleq \bm{a_{t}}$ as a unit vector in spherical coordinates, and using the definition of $\bm{D}$.
		
		On the other hand, for 2D positioning of moving targets on the XY plane, the elevation angle $\phi = \pi/2$. Furthermore, using the definition of $\mathsf{D_{xy}}$ in (\ref{eq:DOP_xyz_defn}) and representing the unit vector $\bm{a_t}$ in spherical coordinates, we get 
		\begin{align}
			& {\rm tr} \big([\bm{Q}]_{1:2,1:2}\big) = r^2_t  \cdot {\rm tr} \Big( [\bm{C}^{-1}]_{1:2,1:2} - \nonumber \\
			& \Big[ \tfrac{\bm{C}^{-1} \bm{a_{t}}(\theta, \tfrac{\pi}{2}) \bm{a_{t}}^T (\theta, \tfrac{\pi}{2}) \bm{C}^{-1}}{\bm{a_{t}}^T(\theta, \tfrac{\pi}{2}) \bm{C}^{-1} \bm{a_{t}}(\theta, \tfrac{\pi}{2})} \Big]_{1:2,1:2} \Big) \nonumber \\
			 \stackrel{(a)}{=} & r^2_t  \cdot \Bigg[{\rm tr} (\bm{E}) - \tfrac{{\rm tr} \big( \big[\bm{C}^{-1} \bm{a_{t}}(\theta, \tfrac{\pi}{2}) \bm{a_{t}}^T (\theta, \tfrac{\pi}{2}) \bm{C}^{-1}\big]_{1:2,1:2} \big)}{\bm{a_{t}}^T(\theta, \tfrac{\pi}{2}) \bm{C}^{-1} \bm{a_{t}}(\theta, \tfrac{\pi}{2})}  \Bigg] \nonumber \\
			\stackrel{(b)}{=} & r^2_t  \cdot \Big({\rm tr} (\bm{E}) -  \tfrac{{\rm tr} \big(\bm{E} \bm{b_{t}}(\theta) \bm{b_{t}}^T (\theta) \bm{E} \big)}{\bm{b_{t}}^T(\theta) \bm{E} \bm{b_{t}}(\theta)}  \Big),
		\end{align}
		where (a) is obtained using the definition of $\bm{E}$. Furthermore, (b) is obtained after simplification using the fact that $[\bm{a_{t}}(\theta, \tfrac{\pi}{2})]_3 = 0$ and defining $\bm{b_{t}}(\theta) \triangleq [\bm{a_{t}}(\theta, \tfrac{\pi}{2})]_{1:2}$. The final result is obtained by using the cyclic permutation property of the trace.
	\end{proof}
	Since the distance $(r_t)$ and direction ($\bm{a_t}$) terms can be decoupled to characterize the 2D/3D DOP of a far-away target, we defined an asymptotic quantity termed as the \textit{range-normalized DOP (RNDOP)}, defined below. 
	\begin{definition}\label{definition:RNDOP}
		For very far-away targets outside the convex hull of anchor locations, the range-normalized DOP ($\mathsf{R}$) is defined as
		\begin{align}
			\label{eq:RNDOP_defn}
			\mathsf{R}(\mathcal{R}_a, \theta, \phi) \triangleq \underset{r_t \rightarrow \infty}{\lim} \tfrac{\mathsf{D}(\mathcal{R}_a, \bm{r_t})}{r_t}, 
		\end{align}
		where $r_t$ is measured from the anchors' centroid.
	\end{definition}%
	It is worthwhile to note that the RNDOP is dependent only on the anchor coordinates and the \emph{target's direction} relative to the anchor centroid. Using this definition in Lemma \ref{lemma:equi_DOP_xyz_and_xy}, we get 
	\begin{align}
		\label{eq:RNDOP_def_3D_2D}
		\mathsf{R_{xyz}}(\mathcal{R}_a, \theta, \phi) & = \sqrt{{\rm tr}(\bm{D}) - \tfrac{\bm{a}^T_{\bm{t}}(\theta, \phi) \bm{D}^2 \bm{a_t}(\theta, \phi)}{\bm{a}^T_{\bm{t}}(\theta, \phi) \bm{D} \bm{a_t}(\theta, \phi)} }, \nonumber \\
		\mathsf{R_{xy}}(\mathcal{R}_a, \theta) & = \sqrt{{\rm tr}(\bm{E}) - \tfrac{\bm{b}^T_{\bm{t}}(\theta) \bm{E}^2 \bm{b_t}(\theta) }{\bm{b}^T_{\bm{t}}(\theta) \bm{E} \bm{b_t} (\theta)} }.
	\end{align}
	
	Using Theorem \ref{thm:GDOP_closed_form}, the upper and lower bounds of the DOP for a target lying on a sphere relative to the anchors' centroid (henceforth referred to as the \textit{`asymptotic bound'} or \textit{`DOP bound'}) can be represented as a \textit{separable function} of the range and the anchor coordinates, as shown in the following lemma.
	\begin{lemma}\label{lemma:asympt_bounds_DOP_xyz_xy}
		If $\bm{D} \succ \bm{0}$, the asymptotic lower and upper bounds of $\mathsf{D_{xyz}}$ for a target at a distance $r_t$ are given by
		\begin{align}
			\label{eq:DOP_xyz_bounds}
			\mathsf{D^{(-)}_{xyz}}(\mathcal{R}_a, r_t) & = r_t \sqrt{{\rm tr}(\bm{D}) - \lambda_+(\bm{D})} \leq \mathsf{D_{xyz}} \leq \nonumber \\
			&  r_t \sqrt{{\rm tr}(\bm{D}) - \lambda_-(\bm{D})} = \mathsf{D^{(+)}_{xyz}}(\mathcal{R}_a, r_t). 
		\end{align}
		If $\bm{E} \succ \bm{0}$, the asymptotic bounds $\mathsf{D_{xy}}$ on the XY-plane for a target at a distance $r_t$ are given by 
		\begin{align}
			\label{eq:DOP_xy_bounds}	    
			\mathsf{D^{(-)}_{xy}}(\mathcal{R}_a, r_t) & = r_t \sqrt{\lambda_{-}(\bm{E})} \leq \mathsf{D_{xy}} \leq  r_t \sqrt{\lambda_{+}(\bm{E})} = \nonumber \\
			& \mathsf{D^{(+)}_{xy}}(\mathcal{R}_a, r_t).
		\end{align}
	\end{lemma}
	\begin{proof}
		Using (\ref{eq:DOP_xyz_rt_scaling_concise}), we can solve the $\mathsf{D_{xyz}}$ minimization problem in the following manner. 
		\begin{align}
			\label{eq:optmization_probs_lower_bnd}
			\mathsf{D^{(-)}_{xyz}}(\mathcal{R}_a, r_t) & = \min_{-\pi \leq \theta < \pi, -\tfrac{\pi}{2} \leq \phi < \tfrac{\pi}{2}}\ \mathsf{D_{xyz}}(\mathcal{R}_a, r_t, \theta, \phi), \nonumber \\
			& \stackrel{(a)}{=} r_t \times \min_{\| \bm{a} \|_2 = 1}\ \sqrt{{\rm tr}(\bm{D}) - \tfrac{\bm{a}^T \bm{D}^2 \bm{a}}{\bm{a}^T \bm{D} \bm{a}}}, \nonumber \\
			& \stackrel{(b)}{=} r_t \sqrt{{\rm tr}(\bm{D}) - \max_{\| \bm{a} \|_2 = 1} \tfrac{\bm{a}^T \bm{D}^2 \bm{a}}{\bm{a}^T \bm{D} \bm{a}}}, \nonumber \\
			& \stackrel{(c)}{=} r_t  \sqrt{{\rm tr}(\bm{D}) - \lambda_{+}(\bm{D}^2, \bm{D})}.
		\end{align}
		Here, (a) is possible due to the decoupling of the target's distance ($r_t$) and direction ($\bm{a_{t}}$) terms for far-away targets. Furthermore, (b) is true since DOP is a positive metric and $\bm{D}$ is independent of $\bm{a}$. Step (c) is obtained by using Proposition \ref{prop:general_eigvals_vecs}, and simplifying the generalized eigenvalue as $\lambda_{-}(\bm{D}^2, \bm{D}) = \lambda_{-}(\bm{D}^{-\frac{1}{2}} \bm{D}^2 \bm{D}^{-\frac{1}{2}}) = \lambda_{-}(\bm{D})$. The upper bound of $\mathsf{D_{xyz}}(\mathcal{R}_a, r_t, \theta, \phi)$ can be obtained in a similar manner, and is omitted for brevity. 
		
		Furthermore, noting that the expression for $\mathsf{D_{xy}}(r_t, \theta)$ is similar to that of $\mathsf {D_{xyz}}(r_t, \theta, \phi)$, we can solve for the bounds in a similar manner as (\ref{eq:optmization_probs_lower_bnd}) to get
		\begin{align}
			\label{eq:DOP_xy_unsimplified_bounds}
			r_t \sqrt{{\rm tr}(\bm{E}) - \lambda_{+}(\bm{E})} \leq \mathsf{D_{xy}} \leq  r_t \sqrt{{\rm tr}(\bm{E}) - \lambda_{-}(\bm{E})}.
		\end{align} 
		Using the fact that $\bm{E} \in \mathbb{S}^{2 \times 2}_{++}$ and ${\rm tr}(\bm{E}) = \lambda_{+}(\bm{E}) + \lambda_{-}(\bm{E})$, we simplify the above to get (\ref{eq:DOP_xy_bounds}).
	\end{proof}
	\begin{remark}\label{remark:anchor_centroid_at_origin}
		By Assumption \ref{assump:anchor_centroid_at_origin}, the distance $r_t$ is measured from the anchor centroid for the aforementioned scaling laws and asymptotic bounds to hold. It is important to note that this assumption does not impact the mathematical characterization of the DOP by itself, since it is a translation-invariant metric. However, if the current anchor centroid is at $\bm{0}$, the addition or removal of anchors will shift the centroid away from the origin. Hence, anchor addition/removal steps must be accompanied by a coordinate translation for these results to be applicable in the new anchor configuration. 
	\end{remark}
	
	\begin{remark}
		For 2D positioning of far-away targets on the XY plane, we do not explicitly consider the target's height ($h_{\rm targ}$) since $\phi = \cos^{-1}(h_{\rm targ}/r_t) \rightarrow \tfrac{\pi}{2}$ for typical far-away targets in vehicular scenarios ($h_{\rm targ} \ll r_t$). While this is an accurate approximation for far-away targets, these approximations are relaxed when discussing the numerical results in Section \ref{sec:num_results}.
	\end{remark}
	As discussed in the next two sections, this analytical formulation simplifies the design of DOP-optimal anchor placement schemes in the context of beyond convex hull localization scenarios. 
	\vspace{-5pt}
	\section{Weighted DOP-Optimal Anchor Placement for 3D Positioning outside the convex hull}
	\label{sec:DOP_3D}
	This section studies weighted DOP-optimal anchor placement schemes for 3D beyond convex hull localization scenarios. We first leverage the RNDOP formulation to simplify the weighted DOP cost function. Then, we propose a robust minimax RNDOP-based anchor placement optimization problem. We then derive the lower bound on the minimax DOP, propose various iterative anchor addition schemes, and derive iterative upper and lower bounds on the RNDOP.
	 
	Consider a region of interest $\mathcal{S}$. Let us assume that $N$ anchors are available, and can be realistically deployed in the localization network taking into account the cost, as well as other system design considerations. A weighted DOP-based cost function can be written as $c(\mathcal{R}_a, \mathcal{S}) = h(w_r(\bm{r}) \cdot \mathsf{D}(\mathcal{R}_a, \bm{r}))$, where $w_r(\bm{r})$ is a location-dependent weight for $\bm{r} \in \mathcal{S}$. To ensure good network-wide localization performance in the region, a reasonable choice for $h(\cdot)$ is the spatial average, denoted by 
	\begin{align} 
		\label{eq:avg_DOP_fun}
		c(\mathcal{R}_a, \mathcal{S}) = \mathbb{E}_{\bm{r}}[w_r(\bm{r}) \mathsf{D}(\mathcal{R}_a, \bm{r})], \text{for } \bm{r} \in \mathcal{S}. 
	\end{align} 	
	In the following result, the insights obtained from Section \ref{subsec:DOP_and_RNDOP_laws} are used to simplify the average weighted DOP-based optimization when the \textit{weighting function is dependent only on} $(\theta, \phi)$.
	\begin{theorem}\label{lemma:optim_simplification}
		Let $\mathcal{A} = \{(\theta, \phi)|\theta \in \bm{\Theta}_\mathcal{A} \subseteq [-\pi, \pi],\phi \in \bm{\Phi}_\mathcal{A} \subseteq [-\tfrac{\pi}{2}, \tfrac{\pi}{2}] \}$ be an arbitrary \textit{angular region of interest}, and $f:\mathcal{A} \rightarrow \mathbb{R}$ be the scalar-valued angle-dependent weighting function. For localizing far-away targets, the optimization problem 
		\begin{align}
			\label{eq:optim_problem_original}
			\mathcal{R}^{\dagger}_a =\ & \underset{\mathcal{R}_a = \{\bm{r_1},\cdots, \bm{r_N} \}}{{\rm arg} \min}\ \mathbb{E}_\mathcal{A}[f(\theta, \phi) \mathsf{D}(\mathcal{R}_a, \bm{r_t})], 
		\end{align}
		is equivalent to
		\begin{align} 
			\label{eq:optim_problem_simplified}
			\mathcal{R}^{\dagger}_a =\ & \underset{\mathcal{R}_a = \{\bm{r_1},\cdots, \bm{r_N} \}}{{\rm arg} \min}\ \mathbb{E}_\mathcal{A}[f(\theta, \phi)\mathsf{R}(\mathcal{R}_a, \theta, \phi)].
		\end{align}
	\end{theorem}
	\begin{proof}
		In the far-away regime, $\mathsf{D}(\mathcal{R}_a, \bm{r_t}) = r_t \mathsf{R}(\mathcal{R}_a, \theta, \phi)$ using Definition \ref{definition:RNDOP}. Therefore, for a scalar function $f(\theta, \phi)$ where $(\theta, \phi) \in \mathcal{A}$, we get 
		\begin{align*} 
			\mathbb{E}_\mathcal{A}[f(\theta, \phi)\mathsf{D}(\mathcal{R}_a, r_t, \theta, \phi)] = r_t \mathbb{E}_\mathcal{A}[f(\theta, \phi) \mathsf{R}(\mathcal{R}_a, \theta, \phi)], 
		\end{align*} 
		since $(\theta, \phi)$ is independent of $r_t$ and therefore, can be factored out of the expectation. Hence, proved. 
	\end{proof}
	\vspace{-15pt}
	\subsection{Minimax DOP-Optimal Anchor Placement and Lower Bound for 3D Positioning}\label{subsec:minimax_3d_rnddop_lb}
	For robust localization performance, minimizing the worst-case positioning error is the end goal \cite{Xiufang_RobustLoc_TWC_2017}. Therefore in this work, we are interested in minimizing the maximum DOP of a target in an \textit{omnidirectional angular region of interest}. In the following result, we simplify the unconstrained minimax problem that facilitates the development of efficient algorithms.  
	\begin{corollary}\label{corollary:minimax_RNDOP}
	The minimax DOP optimization problem for 3D positioning can be simplified as 
	\begin{align}
		\label{eq:minimax_DOP_simplify}
		\underset{\mathcal{R}_a}{{\rm arg} \min}\ \max_{\bm{\Theta}, \bm{\Phi}}\ & \mathsf{D_{xyz}}(\mathcal{R}_a, \bm{r_t}) \Leftrightarrow \underset{\mathcal{R}_a}{{\rm arg}\min}\ & \mathsf{R}^{(+)}_\mathsf{xyz}(\mathcal{R}_a),
	\end{align} 
        where $\bm{\Theta} = \{\theta | \theta \in [-\pi, \pi] \}, \bm{\Phi} = \{\phi | \phi \in [-\frac{\pi}{2}, \frac{\pi}{2}] \}$
	\begin{proof}
		The minimax cost function in (\ref{eq:minimax_DOP_simplify}) is obtained by using the Dirac delta $\delta(\theta_{\max}, \phi_{\rm max})$  as the weight function in (\ref{eq:optim_problem_original}), where $(\theta_{\max}, \phi_{\max})$ corresponds to the maximum eigen-direction such that $\mathbb{E}_{\bm{\Theta}, \bm{\Phi}}[\delta(\theta_{\max}, \phi_{\rm max}) \mathsf{D_{xyz}}(\mathcal{R}_a, r_t, \theta, \phi)] = \mathsf{D}^{(+)}_\mathsf{xyz}(\mathcal{R}_a, r_t)$. Using the equivalence result from Lemma \ref{lemma:optim_simplification} and the fact that $\mathbb{E}_{\bm{\Theta}, \bm{\Phi}} [\delta(\theta_{\max}, \phi_{\rm max}) \mathsf{R_{xyz}}(\mathcal{R}_a, \theta, \phi)] = \mathsf{R}^{(+)}_\mathsf{xyz}(\mathcal{R}_a)$, we obtain the desired result. 
	\end{proof}
	\end{corollary}  
It is worthwhile to observe that since $\mathsf{R} \in (0, \infty]$, the cost function $\mathsf{R}^{(+)}$ can be replaced by $\big(\mathsf{R}^{(+)} \big)^2$ without impacting the optimal solution. In the rest of this paper, we focus on a minimax DOP-optimal anchor placement problem of the form
	\begin{subequations}
		\begin{align}
			\label{eq:minimax_optim_prob_simpilify}
			\mathcal{R}^{\dagger}_a =\ & \underset{\mathcal{R}_a = \{\bm{r_{1}},\cdots, \bm{r_{N}} \}}{{\rm arg} \min}\ {\rm tr}(\bm{D}) - \lambda_-(\bm{D}) \\
			\label{eq:box_constraints}
			{\rm s.t.\ } & \bm{r_l} \preccurlyeq \bm{r_i} \preccurlyeq \bm{r_u},\ i=1,\cdots,N  \\
			\label{eq:min_dist_constraints}
			& \|\bm{r_i} - \bm{r_j} \|_2 \geq d_{\rm th},\ i,j=1,\cdots,N, i \neq j  \\
			\label{eq:Da_defn}
			& \bm{D} = \Big( \sum\nolimits_{i=1}^N \bm{r_{i}} \bm{r_{i}}^T \Big)^{-1} \succ \bm{0} \\
			\label{eq:anch_centroid_constraints}
			& \sum\nolimits_{i=1}^N \bm{r_{i}} = \bm{0}, 
		\end{align}
	\end{subequations}
	where the cost function (\ref{eq:minimax_optim_prob_simpilify}) is due to (\ref{eq:DOP_xyz_bounds}) and Definition \ref{definition:RNDOP}. Equation (\ref{eq:box_constraints}) represents box constraints, which model real-world shapes such as (1) UAV swarms trying to locate a far-away object, and (2) the surface of a terrestrial vehicle (as a first approximation). Constraint (\ref{eq:min_dist_constraints}) denotes the minimum anchor separation constraint using a threshold distance $d_{\rm th}$, to model the practical need for minimizing communication interference, ensuring electromagnetic isolation between all anchors, and ensuring UAV safety. The anchor centroid constraint in \ref{eq:anch_centroid_constraints} appears due to Assumption \ref{assump:anchor_centroid_at_origin}, which is the basis of Lemma \ref{lemma:optim_simplification} and Corollary \ref{corollary:minimax_RNDOP}. It is important to note that the aforementioned optimization problem is non-convex due to (\ref{eq:min_dist_constraints}) and (\ref{eq:Da_defn}). The following result derives a lower bound on the cost function in (\ref{eq:minimax_optim_prob_simpilify}).
	\begin{lemma}\label{lemma:dop_universal_bounds_3d}
		If $\mathcal{R}_a = \{ \bm{r_1}, \bm{r_2}, \cdots, \bm{r_N}\}$ satisfy (\ref{eq:box_constraints})-(\ref{eq:anch_centroid_constraints}), then the lower bounds on the minimax DOP $\mathsf{D}^{(+)}_\mathsf{xyz}(r_t, \mathcal{R}^{\dagger}_a)$ in (\ref{eq:minimax_optim_prob_simpilify}) are given by
		\begin{subequations}
			\begin{align}
				\label{eq:RNDOP_bound_univ_trace}
				\mathsf{D}^{(+)}_\mathsf{xyz}(r_t, \mathcal{R}^{\dagger}_a) \overset{(a)}{\geq} & \ r_t \sqrt{\tfrac{6}{\sum_{i=1}^N \| \bm{r_i} \|^2_2}} \\
				\label{eq:RNDOP_bound_universal}
				\overset{(b)}{\geq} & \tfrac{r_t}{r_{\max}}\sqrt{\tfrac{6}{N}}, 
			\end{align}  
		\end{subequations}
		 
		\begin{align} 
			\text{where } r_{\max} \triangleq & \underset{i=1,2,\cdots, N}{\max} \ \| \bm{r_i} \|_2 \text{ s.t. } \sum\nolimits_{i=1}^N \bm{r_i} = \bm{0}. \nonumber 
		\end{align}
	\end{lemma}
	\begin{proof}
		We begin from the definition of $\bm{C} \triangleq \sum\nolimits_{i=1}^N \bm{r_{i}} \bm{r_{i}}^T$, whose ordered eigenvalues are $\lambda_1(\bm{C}) \leq \lambda_2(\bm{C}) \leq \lambda_3(\bm{C})$. Using properties of the trace, we can write ${\rm tr}(\bm{C}) = \sum_{i=1}^3 \lambda_i(\bm{C}) = \sum_{j=1}^N \| \bm{r_{j}} \|^2$. The first inequality (a) is obtained by the direct application of Proposition \ref{proposition:sum_of_top_2_eigvals_optim} on (\ref{eq:DOP_xyz_bounds}). Furthermore, using the definition of $r_{\max}$, we have $\sum_{i=1}^3 \lambda_i(\bm{C}) = \sum_{j=1}^N \| \bm{r_j} \|^2 \leq N r^2_{\max}$. Using this result, we get inequality (b). 
	\end{proof}
	
		\begin{remark}\label{rem:3d_r_max_univ_bound}
			It is worthwhile to note that $r_{\rm max}$ is the \textit{``size of the 3D physical space''} in which anchors are placed. Therefore, (\ref{eq:RNDOP_bound_universal}) is a universal lower bound on the solution to (\ref{eq:minimax_optim_prob_simpilify}), thus providing theoretical guarantees on the optimal minimax RNDOP given volumetric constraints such as (\ref{eq:box_constraints}). A few important insights from this result are that the lower bound is inversely proportional to (a) $r_{\rm max}$, and (b) the square-root of the number of anchors. While (b) is in agreement with prior works \cite{Sharp_TVT_GDOP_AnchDesgn_2009, Levanon_Lowest2D_GDOP_IEE_2000}, (a) is a novel insight.
		\end{remark}

		\begin{remark}\label{rem:minimize_anchs_prblm}
		This work studies the minimax RNDOP-optimal anchor placement problem under the assumption that it is feasible to deploy $N$ anchors. On the other hand, if the end goal is to minimize the cost of the deployment under robust localization performance constraints, the optimization problem can be reformulated as $N^\dagger = \min_{N \in \mathbb{N}}\ N$ under the constraints (\ref{eq:box_constraints})-(\ref{eq:anch_centroid_constraints}) and ${\rm tr}(\bm{D}) - \lambda_-(\bm{D}) \leq \mathsf{R}^{(+)}_\mathsf{xyz,th}$, where $\mathsf{R}^{(+)}_\mathsf{xyz,th}$ is a threshold on the maximum RNDOP. 
	\end{remark}

	\vspace{-10pt}
	\subsection{Minimax RNDOP-Optimal Single Anchor Addition and Lower Bound for 3D Positioning}\label{subsec:rndop_3d_single_anch_plcmnt}
	The minimax optimization problem in (\ref{eq:minimax_optim_prob_simpilify})-(\ref{eq:anch_centroid_constraints}) is intractable due to the non-convex constraints (\ref{eq:min_dist_constraints})-(\ref{eq:Da_defn}). In addition, as $N$ increases, the curse of dimensionality makes it impractical to solve for $\mathcal{R}^{\dagger}_a$ in real-world scenarios. In the quest for efficient algorithms, we consider the \emph{single} minimax 3D DOP-optimal anchor addition problem to compute the $(k+1)^{\text{th}}$ anchor coordinate, given that $k$ anchors already exist in the system. However, as we already noted in Remark \ref{remark:anchor_centroid_at_origin}, the simplified RNDOP formulation in Section \ref{subsec:DOP_and_RNDOP_laws} is derived based on the assumption that the centroid of the anchors is at the origin. Since adding the $(k+1)^\text{th}$ anchor shifts the anchor centroid away from the origin, a coordinate system translation should be performed, using the following proposition. 
	\begin{proposition}\label{prop:anch_matrix_update_3D}
	Let the existing $k$ anchor coordinates be $\bm{r_i}$ for $i=1,2,\cdots,k$, such that it satisfies (\ref{eq:box_constraints})-(\ref{eq:anch_centroid_constraints}) for $N=k$, and let $\bm{C_k} \triangleq \sum_{i=1}^k \bm{r_i r_i}^T$. After adding the $(k+1)^\text{th}$ anchor, the updated anchor matrix $\bm{C_{k+1}}$ is given by 
	\begin{align}
		\label{eq:update_3D_C_k+1}
		\bm{C_{k+1}} = \bm{C_k} + \big( 1 - \tfrac{1}{k+1} \big)\bm{r_{k+1} r_{k+1}}^T.
	\end{align}
	\end{proposition} 
	\begin{proof}
	We note that Assumption \ref{assump:anchor_centroid_at_origin} is key to the minimax 3D RNDOP-based formulation in (\ref{eq:minimax_optim_prob_simpilify}). When the $(k+1)^\text{th}$ anchor is added, the anchor centroid is at $\frac{\bm{r_{k+1}}}{k+1}$, since $\sum_{i=1}^k \bm{r_i}/k = \bm{0}$ by design. Hence, translating the coordinate system so that the new centroid is at $\bm{0}$, we get the new anchor matrix $\bm{C_{k+1}} = \sum_{i=1}^{k+1} \big(\bm{r_i} - \tfrac{\bm{r_{k+1}}}{k+1} \big) \big(\bm{r_i} - \tfrac{\bm{r_{k+1}}}{k+1} \big)^T$. The desired result is obtained using the definition of $\bm{C_k}$ and simplifying it. 
	\end{proof} 
	Thus, (\ref{eq:update_3D_C_k+1}) enables us to iteratively update the cost function in \ref{eq:minimax_optim_prob_simpilify} while (a) representing the $(k+1)^\text{th}$ anchor location in the \emph{local coordinate system (LCS)}, and (b) enforcing Assumption \ref{assump:anchor_centroid_at_origin}. The following result bounds the achievable RNDOP when placing the $(k+1)^\text{th}$ anchor, given $k$ existing anchors. 
	\begin{proposition}\label{prop:up_low_iterative_3d_bounds}
		Following Proposition \ref{prop:anch_matrix_update_3D}, let $\bm{D_{k+1}} \triangleq \bm{C_{k + 1}}^{-1}$. The minimax RNDOP ($\mathsf{R}_\mathsf{xyz}^{(+)}(\bm{D_{k+1}})$) after adding the $(k+1)^\text{th}$ anchor can be bounded using
		\begin{align}
			\label{eq:up_low_bounds_minimax_dop}
			\tfrac{1}{\lambda_3(\bm{C_k})} + \tfrac{1}{\lambda_2(\bm{C_k})} & \leq \mathsf{R}_\mathsf{xyz}^{(+)}(\bm{D_{k+1}}) \leq \tfrac{1}{\lambda_2(\bm{C_k})} + \tfrac{1}{\lambda_1(\bm{C_k})}.
		\end{align}
	\end{proposition}
	\begin{proof}
		We begin by noting that for the iterative anchor placement problem, (\ref{eq:minimax_optim_prob_simpilify}) can be written as $\mathsf{R}_\mathsf{xyz}^{(+)}(\bm{D_{k+1}}) = \lambda_3 (\bm{D_{k+1}}) + \lambda_2 (\bm{D_{k+1}})$. Using this, the desired result follows from the direct application of Proposition \ref{prop:eigvalue_interlace_thm} on the eigenvalues of $\bm{C_{k+1}}$ using (\ref{eq:update_3D_C_k+1}), and noting that $\lambda_{4-i}(\bm{C_{k+1}}) = \lambda^{-1}_{i}(\bm{D_{k+1}})$ for $i=1,2,3$.
	\end{proof}	
	Using (\ref{eq:update_3D_C_k+1}), we transform (\ref{eq:minimax_optim_prob_simpilify})-(\ref{eq:anch_centroid_constraints}) in the following main result to design an efficient minimax DOP-optimal \emph{single} anchor placement scheme.
	\begin{lemma}
		\label{lemma:iterative_3D_RNDOP_based}
		Let the anchor locations $\bm{r_i},i=1,2,\cdots,k$ satisfy (\ref{eq:box_constraints})-(\ref{eq:anch_centroid_constraints}), where $\bm{C_k} = \sum_{i=1}^{k} \bm{r_{i}} \bm{r_{i}}^T$. The RNDOP-optimal anchor coordinate $\bm{r^{(\rm rnd)}_{k+1,3D}}$ in the current LCS is given by
		\begin{subequations}
			\begin{align}
				\label{eq:robust_seq_anch_plcmnt_no_coord_trnsfrm}
				\bm{r^{(\rm rnd)}_{k+1,3D}} =\ & \underset{\bm{r}}{{\rm arg} \min}\ \Big \{ {\rm tr}\big(\big[\bm{C_k} + \big(1 - \tfrac{1}{k+1}\big) \bm{r r}^T \big]^{-1} \big)    - \nonumber \\
				& \lambda_{-} \big(\big[\bm{C_k} + \big(1 - \tfrac{1}{k+1}\big) \bm{r r}^T \big]^{-1} \big) \Big\} \\
				\label{eq:box_cnstrnt_no_trnsfrm}
				{\rm s.t. } &\ \bm{r_l} \preceq \bm{r} \preceq \bm{r_u} \\
				\label{eq:dist_btwn_anch_cnstrnt_no_trnsfrm}
				& \| \bm{r} - \bm{r_i} \|_2 \geq d_{\rm th}, i=1,2,\cdots, k.
			\end{align}
		\end{subequations}
	\end{lemma}
	\begin{proof}
		The cost function is obtained by direct substitution of (\ref{eq:update_3D_C_k+1}) in (\ref{eq:minimax_optim_prob_simpilify}). Since $\bm{C_{k+1}}$ results from a rank-1 update applied to $\bm{C_k} \succ \bm{0}$, it follows that $\bm{C_{k+1}} \succ \bm{0}$. In addition, as discussed in Proposition \ref{prop:anch_matrix_update_3D}, the constraint (\ref{eq:anch_centroid_constraints}) is implicitly enforced by the update equation without translating the current LCS. Since the $k$ anchors already satisfy (\ref{eq:box_constraints})-(\ref{eq:min_dist_constraints}), they need to be enforced only for the $(k+1)^\text{th}$ anchor. Hence, proved.
	\end{proof}
	It is important to note that at the beginning of each iteration, the anchor centroid should be at $\bm{0}$. This needs to be facilitated by a coordinate translation of all anchors as well as the `volumetric' constraints at the end of each iteration, as discussed later in Section \ref{sec:minimax_3D_DOP_iter_algo}. 
	\vspace{-10pt}
	\subsection{Heuristic Cost Functions for Single Anchor Addition for 3D Positioning}\label{subsec:3d_anch_addn_heuristics}
	\subsubsection{Trace-based Heuristic}
	By observing the cost function (\ref{eq:robust_seq_anch_plcmnt_no_coord_trnsfrm}) and its bounds in (\ref{eq:up_low_bounds_minimax_dop}), it is clear that minimizing ${\rm tr} (\bm{D_{k+1}}) \triangleq {\rm tr} \big(\big[\bm{C_k} + \big(1 - \tfrac{1}{k+1}\big) \bm{r r}^T \big]^{-1} \big)$ \textit{tends to minimize} the maximum 3D RNDOP. Hence, we propose the trace-based minimization problem below.
	\begin{lemma}
	    \label{lemma:iterative_trace_based_3d}
		Let the anchor locations $\bm{r_i},i=1,2,\cdots,k$ satisfy (\ref{eq:box_constraints})-(\ref{eq:anch_centroid_constraints}), where $\bm{C_k} = \sum_{i=1}^{k} \bm{r_{i}} \bm{r_{i}}^T$. The trace-based heuristic RNDOP minimization problem is given by
		\begin{align}
			\label{eq:trace_based_heuristic_3D_RNDOP}
			\bm{r^{(\rm tr)}_{k+1,3D}} & = \underset{\bm{r}}{{\rm arg} \max}\ \tfrac{\bm{r}^T \bm{D_{k+1}}^2 \bm{r}}{1 + \frac{k}{k+1} \bm{r}^T \bm{D_{k+1}} \bm{r}} \nonumber \\
			{\rm s.t. } & \text{(\ref{eq:box_cnstrnt_no_trnsfrm})-(\ref{eq:dist_btwn_anch_cnstrnt_no_trnsfrm}).} 
		\end{align}  
	\end{lemma}
	\begin{proof}
		Applying Proposition \ref{prop:Sherman_Morrison_formula} to $\bm{D_{k+1}} \triangleq \bm{C_{k+1}}^{-1}$ using (\ref{eq:update_3D_C_k+1}), we get $\bm{D_{k+1}} = \bm{D_k} - \frac{k}{k+1}\cdot \frac{\bm{D_k r r}^T \bm{D_k}}{1 + \frac{k}{k+1} \bm{r}^T \bm{D_k r}}$. Using properties of the trace and ignoring the terms independent of $\bm{r}$, we get (\ref{eq:trace_based_heuristic_3D_RNDOP}). The constraints remain unchanged relative to (\ref{eq:robust_seq_anch_plcmnt_no_coord_trnsfrm})-(\ref{eq:dist_btwn_anch_cnstrnt_no_trnsfrm}). 
	\end{proof}
	\begin{remark}
        When sequentially adding anchors based on the minimax RNDOP cost function over multiple iterations, the trace-based heuristic (\ref{eq:trace_based_heuristic_3D_RNDOP}) tends to minimize \textit{all the eigenvalues} of $\bm{C_{k+1}}$ relative to those of $\bm{C_k}$. In other words, minimizing the trace in the $(k+1)^\text{th}$ iteration has the effect of tightening the upper and lower bounds in (\ref{eq:up_low_bounds_minimax_dop}), thereby minimizing the maximum 3D RNDOP \textit{over consecutive iterations}. 
	\end{remark}
	
	\subsubsection{Eigenvector-based Heuristic}
	From (\ref{eq:up_low_bounds_minimax_dop}), we observe that maximizing the minimum eigenvalue $\lambda_3(\bm{C_k})$ tends to minimize the upper bound. In the following, a computationally efficient heuristic scheme is proposed, which \textit{tends to minimize} the maximum 3D RNDOP. 
	\begin{lemma}
	    \label{lemma:min_eigvalue_anch_plcmnt_3D}
		Let the anchor locations $\bm{r_i},i=1,2,\cdots,k$ satisfy (\ref{eq:box_constraints})-(\ref{eq:anch_centroid_constraints}), where $\bm{C_k} = \sum_{i=1}^{k} \bm{r_{i}} \bm{r_{i}}^T$. The eigenvector-based heuristic-based optimization problem is given by 
		\begin{subequations} 
			\begin{align}
				\label{eq:min_eig_based_optim}
				\bm{r^{(\rm eig)}_{k+1, 3D}} & = \alpha_{\rm max} \bm{v_{-}}(\bm{C_k}), \text{ where } \\
				\label{eq:max_dist_along_min_eigvec}
				\alpha_{\rm max} & = \underset{\bm{r_l}  \preceq \alpha\bm{v_{-}}(\bm{C_k})  \preceq \bm{r_u}}{{\rm arg} \max} \ \alpha.
			\end{align}
		\end{subequations} 
	\end{lemma}
	\begin{proof}
		We have $\bm{C_k} = \sum_{i=1}^{3} \lambda_i(\bm{C_k}) \bm{v_i}(\bm{C_k})\bm{v}^T_{\bm{i}}(\bm{C_k})$. Substituting this in (\ref{eq:update_3D_C_k+1}), if $\bm{r} = \alpha \bm{v_1}(\bm{C_k}) = \alpha \bm{v_{-}}(\bm{C_k})$, then $\lambda_{i}(\bm{C_{k+1}})$ for $i=1,2,3,$ are given by $\lambda_1(\bm{C_{k+1}}) = \lambda_1(\bm{C_{k+1}}) + \tfrac{k \alpha^2}{k+1}$, and $\lambda_l(\bm{C_{k+1}}) = \lambda_l(\bm{C})$, for $l=2,3$. Therefore, maximizing $\alpha$ tends to minimize the upper bound in (\ref{eq:up_low_bounds_minimax_dop}). Accounting for the box constraints, the maximum possible $\alpha$ ($\alpha_{\rm max}$) is then obtained by solving (\ref{eq:max_dist_along_min_eigvec}). 
	\end{proof}
	It is worthwhile to note that (\ref{eq:min_eig_based_optim}) is a convex problem that can be solved efficiently for large $k$. However, the above formulation does not explicitly enforce the minimum distance constraint (\ref{eq:dist_btwn_anch_cnstrnt_no_trnsfrm}). In Section \ref{sec:minimax_3D_DOP_iter_algo}, we discuss random perturbation-based methods to enforce (\ref{eq:dist_btwn_anch_cnstrnt_no_trnsfrm}). 
		
	\section{RNDOP-Optimal Single Anchor Addition for 2D Positioning Outside the Convex Hull}
		\label{sec:DOP_2D}
	In this section, we extend the results from Section \ref{sec:DOP_3D} to design \emph{single} anchor addition schemes for beyond convex-hull 2D positioning. Intuitively, this extension is enabled by the fact that the key anchor-dependent quantities $\bm{D}$ (for 3D positioning) and $\bm{E}$ (for 2D positioning) are related by $\bm{E} = [\bm{D}]_{1:2,1:2}$ as shown in Lemma \ref{lemma:asympt_bounds_DOP_xyz_xy}. First, we first derive the update equation for the anchor matrix $\bm{E}$, and introduce the notation used throughout the rest of this section. 
	\begin{proposition}
		Let the anchor locations $\bm{r_i},i=1,2,\cdots,k$, such that $\sum_{i=1}^k \bm{r_i} = \bm{0}$, $\bm{E_k} \triangleq \big[\bm{C_k}^{-1}\big]_{1:2,1:2}$, and $\bm{\tilde{r}}\triangleq [\bm{r}]_{1:2}$. Upon adding the $(k+1)^\text{th}$ anchor at $\bm{r_{k+1}}$, the updated anchor matrix $\bm{E_{k+1}}$ is given by
		\begin{align}
			\label{eq:E_k+1_update}
			\bm{E_{k+1}} = \bm{E_k} - \tfrac{\big(1 - \frac{1}{k+1}\big) \bm{E_k} \bm{\tilde{r}_{k+1} \tilde{r}_{k+1}}^T \bm{E_k} }{1+ \big(1 - \frac{1}{k+1}\big)\bm{r_{k+1}}^T \bm{D_k} \bm{r_{k+1}} }.
		\end{align}
		
	\end{proposition}
	\begin{proof}
		By definition, we have $\bm{D_{k+1}} \triangleq \bm{C_{k+1}}^{-1}$. By applying Proposition \ref{prop:Sherman_Morrison_formula} on (\ref{eq:update_3D_C_k+1}), we get
		\begin{align}
			\label{eq:D_k+1_mat_inv_lemma}
			\bm{D_{k+1}} = \bm{D_k} - \tfrac{\big(1 - \frac{1}{k+1}\big) \bm{D_k} \bm{r_{k+1} r_{k+1}}^T \bm{D_k} }{1+ \big(1 - \frac{1}{k+1}\big)\bm{r_{k+1}}^T \bm{D_k} \bm{r_{k+1}} }. 
		\end{align} 
		Then, we use the definition $\bm{E_l} = [\bm{D_l}]_{1:2,1:2}$ for $l=k, k+1,$ and the fact that 
		\begin{align*}
			[\bm{D_k} \bm{r_{k+1} r_{k+1}}^T \bm{D_k}]_{1:2,1:2} = \bm{I_{2 \times 3}} \bm{D_k} \bm{r_{k+1} r_{k+1}}^T \bm{D_k}	\bm{I_{2 \times 3}}^T,	
		\end{align*}
		where $\bm{I_{2 \times 3}} = [\bm{I_{2}}\ \ \bm{0_{2 \times 1}}]$ and $\bm{0_{2 \times 1}} = [0\ 0]^T$. Simplifying the numerator of (\ref{eq:D_k+1_mat_inv_lemma}) using these, we get the desired result.
	\end{proof}

	\begin{corollary}\label{corollary:dop_universal_bounds_2d}
		Let $\mathcal{R}^\dagger_a = \{ \bm{r_1}, \bm{r_2}, \cdots, \bm{r_N}\}$ satisfy (\ref{eq:box_constraints})-(\ref{eq:anch_centroid_constraints}). Furthermore, let the anchors be positioned in a non-coplanar geometrical shape. Then the lower bounds on the minimax DOP $\mathsf{D}^{(+)}_\mathsf{xy}(r_t, \mathcal{R}^{\dagger}_a)$ is given by
			\begin{align}
				\label{eq:RNDOP_xy_bound_univ_and_trace}
				\mathsf{D}^{(+)}_\mathsf{xy}(r_t, \mathcal{R}^{\dagger}_a) \overset{(a)}{\geq} \ r_t \sqrt{\tfrac{2}{\sum_{i=1}^N \| \bm{s_i} \|^2_2}} \overset{(b)}{\geq}\tfrac{r_t}{s_{\max}}\sqrt{\tfrac{2}{N}}, 
			\end{align}  
		where $s_{\max} \triangleq \underset{i=1,2,\cdots, N}{\max} \ \| \bm{s_i} \|_2 \text{ s.t. } \sum\nolimits_{i=1}^N \bm{s_i} = \bm{0}$ and  $\bm{s_i} = [x_i\ y_i]^T$.
	\end{corollary}
	\begin{proof}
		Please refer Appendix \ref{appendix:a}.
	\end{proof}
	\begin{remark}
		Similar to Remark \ref{rem:3d_r_max_univ_bound}, $s_{\rm max}$ is the \textit{``size of the 2D physical space''} in which anchors are placed. Therefore, (\ref{eq:RNDOP_xy_bound_univ_and_trace}) is a universal lower bound on the solution to the 2D minimax RNDOP anchor placement problem, thus providing theoretical guarantees under volumetric constraints. In addition, the inverse proportional dependence on $s_{\rm max}$ and $\sqrt{N}$ holds in 2D localization scenarios as well. 
	\end{remark}
        The following result bounds the achievable 2D RNDOP when placing the $(k+1)^\text{th}$ anchor, given $k$ existing anchors. 	
        \begin{proposition}\label{prop:lb_ub_rndop_xy_iterative}
	       Let the current anchor matrix be $\bm{E_k}$. Then, the minimax 2D RNDOP for the $(k+1)^{\rm th}$ iteration can be bounded as $\lambda_{-}(\bm{E_k}) \leq \mathsf{R}^{(+)}_\mathsf{xy}(\bm{E_{k+1}}) \leq \lambda_{+}(\bm{E_k})$.
        \end{proposition}
        \begin{proof}
	       The proof follows from the direct application of Proposition \ref{prop:eigvalue_interlace_thm} on the eigenvalues of $\bm{E_k}$ and $\bm{E_{k+1}}$ in (\ref{eq:E_k+1_update}).
        \end{proof}
        Similar to Lemmas \ref{lemma:iterative_3D_RNDOP_based}-\ref{lemma:min_eigvalue_anch_plcmnt_3D}, we design single RNDOP-optimal anchor placement schemes for 2D positioning using (\ref{eq:E_k+1_update}) as shown below. 
	
	\subsubsection{Iterative Minimax RNDOP-Optimal Anchor Addition for 2D Positioning}\label{subsec:rndop_2d_single_anch_plcmnt}
	In the following result, we extend the iterative max-min RNDOP optimization problem discussed in Lemma \ref{lemma:iterative_3D_RNDOP_based} to the 2D positioning case. 
	
	\begin{lemma}
	    \label{lemma:iterative_rndop_2d}
		The RNDOP-optimal anchor coordinate $\bm{r^{(\rm rnd)}_{k+1,2D}}$ in the $(k+1)^\text{th}$ iteration is obtained by solving
			\begin{align}
				\label{eq:robust_seq_anch_plcmnt_no_coord_trnsfrm_2D}
				\bm{r^{(\rm rnd)}_{k+1,2D}}  = & \underset{\bm{r}}{{\rm arg} \min}\ \lambda_{+} \Big( \bm{E_k} - \tfrac{\big(1 - \frac{1}{k+1}\big) \bm{E_k} \bm{\tilde{r} \tilde{r}}^T \bm{E_k} }{1+ \big(1 - \frac{1}{k+1}\big)\bm{r }^T \bm{D_k} \bm{r } } \Big) \nonumber \\
				{\rm s.t. } & (\ref{eq:box_cnstrnt_no_trnsfrm})-(\ref{eq:dist_btwn_anch_cnstrnt_no_trnsfrm}), \bm{\tilde{r}}=[\bm{r}]_{1:2}.
			\end{align}
	\end{lemma}
	\begin{proof}
		The proof directly follows from Lemma \ref{lemma:iterative_3D_RNDOP_based}, by noting that for $\bm{Z} \in \mathbb{S}^{2\times 2}_{++}$, ${\rm tr}(\bm{Z}) - \lambda_{-}(\bm{Z}) = \lambda_{+}(\bm{Z})$. The constraints remain unchanged when compared to the optimal anchor placement schemes for 3D positioning. 
	\end{proof}
	
\subsubsection{Trace-based Heuristic}
Analyzing the cost function (\ref{eq:robust_seq_anch_plcmnt_no_coord_trnsfrm_2D}), it is evident that minimizing ${\rm tr}(\bm{E_{k+1}})$ \textit{tends to minimize} the max RNDOP. Hence, similar to the 3D case, we propose the trace-based RNDOP minimization problem below.
\begin{lemma}
    \label{lemma:iterative_trace_based_2d}
	The trace-based heuristic 2D RNDOP minimization problem is given by
	\begin{align}
		\label{eq:trace_based_heuristic_2D_RNDOP}
		\bm{r^{(\rm tr)}_{k+1,3D}} & = \underset{\bm{r}}{{\rm arg} \max}\ \tfrac{\bm{\tilde{r}}^T \bm{E_k}^2 \bm{\tilde{r}}}{1 + \big(1 - \frac{1}{k+1} \big)\bm{r}^T \bm{D_k} \bm{r}} \nonumber \\
		{\rm s.t. } & (\ref{eq:box_cnstrnt_no_trnsfrm})-(\ref{eq:dist_btwn_anch_cnstrnt_no_trnsfrm}), \bm{\tilde{r}}=[\bm{r}]_{1:2}.
	\end{align}  
\end{lemma}
\begin{proof}
	In the cost function ${{\rm arg} \min}_{\bm{r}}\ \ {\rm tr}(\bm{E_{k+1}})$, we observe from (\ref{eq:E_k+1_update}) that only the second term is dependent on $\bm{r}$. Therefore, ignoring the first term, noticing that $(1 - \frac{1}{k+1})$ in the numerator is independent of $\bm{r}$, and absorbing the negative sign, we get the desired result. 
\end{proof}
\subsubsection{Eigenvector-based heuristic}
This scheme finds the anchor location at the $(k+1)^{\rm th}$ iteration by solving two convex problems sequentially. The two problems are designed by observing in (\ref{eq:E_k+1_update}) that (a) aligning $\bm{\tilde{r}_{k+1}}$ along the maximum eigenvector of $\bm{E_k}$, and (b) minimizing the denominator of the second term, \emph{tends to minimize} $\lambda_{+}(\bm{E_{k+1}})$. The mathematical formulation is provided in the following proposition. 
\begin{lemma}
    \label{lemma:iterative_eigvec_based_2d}
	The solution to the eigenvector heuristic-based 2D RNDOP minimization problem is given by $\bm{r^{(\rm eig)}_{k+1,2D}} = [\bm{\tilde{r}_{\lambda +}}^T\ r_z]^T$ where $\bm{\tilde{r}_{\lambda +}} = \alpha_{\rm max} \bm{v}_{+}(\bm{E_k})$ such that
	\begin{align}
	\label{eq:eigvec_based_heuristic_2D_RNDOP}
	\text{P1:} \ \ \alpha_{\rm max} & = \underset{[\bm{r_l}]_{1:2}  \preceq \alpha\bm{v_{+}}(\bm{E_{k}})  \preceq [\bm{r_u}]_{1:2}} {{\rm arg} \max }\ \ \alpha \bm{v}_{+}(\bm{E_{k}}) \\
	\label{eq:eigvec_based_heuristic_z_part}
	\text{P2:}\ \ r_z & = \underset{\bm{r_l}]_{3}  \leq z  \leq [\bm{r_u}]_{3}} {{\rm arg} \min }\ \ pz^2 + 2z \bm{q}^T \bm{\tilde{r}_{\lambda +}} 
	\end{align}  
where $\bm{q} = [\bm{D_k}]_{1:2,3}$ and $p = [\bm{D_k}]_{3,3}$.
\end{lemma}
\begin{proof}
	The optimization problem (P1) seeks to minimize the numerator of the second term in (\ref{eq:E_k+1_update}) under the box constraints in (\ref{eq:box_cnstrnt_no_trnsfrm}). (P1) is obtained in a similar manner as in Lemma \ref{lemma:min_eigvalue_anch_plcmnt_3D} by noticing that (a) the minimum value of the numerator in the second term is achieved when $\bm{\tilde{r}}=\alpha \bm{v}_{+}(\bm{E_k})$, and (b) $\alpha$ is chosen as high as possible under the XY-box constraints. 
	
	Then, (P2) solves for the z-coordinate by minimizing the denominator under (a) the box constraints for the z-coordinate, and (b) $\bm{\tilde{r}_{\lambda +}}$ which are the x- and y-coordinates of the $(k+1)^{\rm th}$ anchor. Considering the terms dependent on $\bm{r}$, substituting $\bm{r} = [\bm{\tilde{r}_{\lambda +}}^T\ r_z]^T$, and retaining only the terms dependent on the optimization variable $r_z$, we get the desired solution. 
\end{proof}	
It is worthwhile to note that in the above scheme, the distance constraint (\ref{eq:dist_btwn_anch_cnstrnt_no_trnsfrm}) is not explicitly enforced, as in the case of the corresponding 3D case. To do so, we use a random perturbation-based scheme discussed in Algorithm \ref{algo:random_prtrb_min_anch_sep}. 
\begin{algorithm}[t]
	\small
	\begin{algorithmic}[1]
		\State \textbf{Input:} ${\rm method} \in \{\mathtt{rnd}, \mathtt{tr}\}, {\rm mode} \in \{\mathtt{2D}, \mathtt{3D} \}, N_{\rm a}, d_{\rm th}, \bm{r_{u}}, \bm{r_{l}}$,  $\mathcal{R}_{a}(N) = \{\bm{r_{1}}, \bm{r_{2}}, \cdots, \bm{r_{N}} \}$ such that constraints (\ref{eq:box_constraints})-(\ref{eq:anch_centroid_constraints}) are satisfied.
		\State $\triangleright$ Iteratively obtain coordinates of $N_{\rm a}$ anchors.
		\For{$k \gets (N+1)$ to $(N + N_{\rm a})$} 
		\If{$k = (N+1)$}\Comment{Initialization}
		\State $\bm{r^{'}_{u}} \gets \bm{r_{u}}, \bm{r^{'}_{l}} \gets \bm{r_{l}}$ \Comment{Save GCS box constraints}
		\EndIf  
		\State $\bm{C_{k-1}} \gets \sum_{i=1}^{k-1} \bm{r_{i}} \bm{r_{i}}^T$
		\State $\bm{D_{k-1}} \gets \bm{C_{k-1}}^{-1}, \bm{E_{k-1}} \gets [\bm{D_{k-1}}]_{1:2,1:2}$
		\If{${\rm method} = \mathtt{rnd}$}
		\If{${\rm mode}= \mathtt{3D}$}
		\State Solve (\ref{eq:robust_seq_anch_plcmnt_no_coord_trnsfrm})-(\ref{eq:dist_btwn_anch_cnstrnt_no_trnsfrm}) to obtain $\bm{r_k} \gets \bm{r^{(\rm rnd)}_{k,3D}}$.
		\ElsIf{${\rm mode}= \mathtt{2D}$} 
		\State Solve (\ref{eq:robust_seq_anch_plcmnt_no_coord_trnsfrm_2D}) to obtain $\bm{r_k} \gets \bm{r^{(\rm rnd)}_{k,2D}}$.
		\EndIf
		\ElsIf{${\rm method} = \mathtt{tr}$}
		\If{${\rm mode}= \mathtt{3D}$}
		\State Solve (\ref{eq:trace_based_heuristic_3D_RNDOP}) to obtain $\bm{r_k} \gets \bm{r^{(\rm tr)}_{k,3D}}$.
		\ElsIf{${\rm mode}= \mathtt{2D}$}
		\State Solve (\ref{eq:trace_based_heuristic_2D_RNDOP}) to obtain $\bm{r_k} \gets \bm{r^{(\rm tr)}_{k,2D}}$.
		\EndIf
		\begin{comment} 
		\ElsIf{${\rm method} = \mathtt{tr}$}
		\State Solve the optimization problem (\ref{eq:min_eig_based_optim}). 
		\State $\bm{r'_N} \gets \bm{r^{(\lambda_{-})}_{N,3D}}$.
		\State $\bm{r_N} \gets \text{Algorithm2}(\bm{r_1}, \cdots, \bm{r'_N}, d_{\rm th})$.
		\end{comment} 
		\EndIf 
		\State $\triangleright$ Compute new centroid and translate anchors/box constraints to new LCS.
		\State $\bm{r_c} \gets \frac{1}{k} \sum_{i = 1}^{k} \bm{r_i}$. 
		\State $\bm{r_u} \gets (\bm{r_u} - \bm{r_c})$ and $\bm{r_l} \gets (\bm{r_l} - \bm{r_c})$. 
		\For{$i \gets 1$ to $k$} 
		\State $\bm{r_i} \gets (\bm{r_i} - \bm{r_c})$ 
		\EndFor
		\EndFor
		\State $\triangleright$ Transform all coordinates back to GCS.
		\State $\bm{r'_c} \gets (\bm{r'_u} - \bm{r_u})$ \Comment{Equivalent to $\bm{r'_c} \gets (\bm{r'_l} - \bm{r_l})$} 
		\For{$i \gets 1$ to $(N + N_{\rm a})$}
		\State $\bm{r_i} \gets (\bm{r_i} + \bm{r'_c})$ 
		\EndFor
		\State \Return $\mathcal{R}_{a}(N + N_{\rm a}) = \mathcal{R}_{a}(N) \cup \{ \bm{r_{N+1}}, \cdots, \bm{r_{N + N_{\rm a}}} \}$.
	\end{algorithmic}
	\caption{Iterative Minimax RNDOP and Trace Heuristic-based Anchor Placement for 2D and 3D Positioning}
	\label{algo:sequential_minimax_gdop_anch_plcement}
\end{algorithm}
\vspace{-10pt}
\section{Iterative Anchor Addition Algorithms for 2D and 3D Positioning Outside the Convex Hull}\label{sec:minimax_3D_DOP_iter_algo}
In this section, we present an iterative anchor addition-based algorithmic framework for beyond convex hull positioning, that leverages the optimal single anchor addition schemes proposed in Sections \ref{sec:DOP_3D} and \ref{sec:DOP_2D}. The goal is to adding $N_{\rm a}$ anchors to an existing set of $N$ anchors such that the maximum RNDOP is minimized, thereby resulting in robust positioning performance. 
\vspace{-10pt}
\subsection{Algorithm Description}
Algorithm \ref{algo:sequential_minimax_gdop_anch_plcement} shows the iterative anchor addition algorithm for the minimax RNDOP and the trace-based heuristic schemes. Algorithm \ref{algo:random_prtrb_min_anch_sep} shows the iterative anchor addition algorithm based on the eigenvector heuristic, coupled with a random perturbation step to ensure that the minimum anchor distance constraint is satisfied. Both the algorithms follow the sequence of steps outlined below. However for ease of exposition, we use Algorithm (\ref{algo:sequential_minimax_gdop_anch_plcement}) to illustrate the key steps.
\begin{enumerate}
	\item \textbf{Initialization}: The algorithm is initialized with $N$ anchor coordinates, constraint-related system parameters, and optimization method-related parameters. This corresponds to lines 1-2 in Algorithm \ref{algo:sequential_minimax_gdop_anch_plcement}.
	\item \textbf{Iterative Anchor Placement}: Depending on the method, the appropriate \emph{single anchor addition} problem is solved (lines 3-28 in Algorithm (\ref{algo:sequential_minimax_gdop_anch_plcement})). At each step, a coordinate translation is performed \emph{post anchor addition}, to ensure that the anchor centroid is at the origin for the next iteration (lines 23-27 in Algorithm (\ref{algo:sequential_minimax_gdop_anch_plcement})); a necessary condition as highlighted in Remark \ref{remark:anchor_centroid_at_origin}.
	\item \textbf{Coordinate System Reversion}: After all the $N_{\rm a}$ optimal anchor coordinates are computed in the LCS, they are transformed back to the global coordinate system for deployment purposes (lines 29-34 in Algorithm \ref{algo:sequential_minimax_gdop_anch_plcement}).
\end{enumerate}
\vspace{-13pt}
\subsection{Random Perturbation-based Extension}\label{subsec:rand_perturb_eigvec}
The eigenvector heuristic-based schemes do not explicitly enforce the box and minimum anchor distance constraints. To do so, we propose a random perturbation-based extension shown in Algorithm \ref{algo:random_prtrb_min_anch_sep}, where we search for potential anchor locations near $\bm{r^{(\rm eig)}_{k,3D}}/\bm{r^{(\rm eig)}_{k,2D}}$ (for 3D/2D positioning respectively) by randomly perturbing them until (\ref{eq:box_cnstrnt_no_trnsfrm})-(\ref{eq:dist_btwn_anch_cnstrnt_no_trnsfrm}) is satisfied. This operation is repeated up to $N_{\rm max}$ times in each \emph{anchor addition iteration until the perturbed location} satisfies all the constraints (lines 22-42 in Algorithm \ref{algo:random_prtrb_min_anch_sep}. However, the random perturbations are not guaranteed to yield a solution that satisfies all the constraints. To improve robustness, we propose adding redundant anchors. In this mechanism, if $N_{\rm f}$ anchor addition iterations fail to yield a feasible anchor location, the algorithm effectively executes $(N_{\rm a} + N_{\rm f})$ iterations until a subset of $N_{\rm a}$ anchor locations satisfy (\ref{eq:box_cnstrnt_no_trnsfrm})-(\ref{eq:dist_btwn_anch_cnstrnt_no_trnsfrm}).

	\begin{algorithm}[t]
		\small
		\begin{algorithmic}[1]
			\State \textbf{Input:} ${\rm mode} \in \{\mathtt{2D}, \mathtt{3D} \}, d_{\rm th}, \bm{r_{u}}, \bm{r_{l}}, N_{\rm a}, \eta, N_{\rm max}, \mathcal{R}_{a}(N) = \{\bm{r_{1}}, \bm{r_{2}}, \cdots, \bm{r_{N}} \}$ such that constraints (\ref{eq:box_constraints})-(\ref{eq:anch_centroid_constraints}) are satisfied. 
			\State $\bm{V} = \bm{1_N} \in \{0,1\}^{N}$, $q \gets 0, n \gets 0$.
			\While{$n \leq N_{\rm a}$}
				\State $q \gets q + 1, k \gets (N + q)$.
				\If{$q = 1$}
					\State $\bm{r^{'}_{u}} \gets \bm{r_{u}}, \bm{r^{'}_{l}} \gets \bm{r_{l}}$ \Comment{Save GCS box constraints}
				\EndIf  
				\State $\bm{C_{k-1}} \gets \sum_{i=1}^{k-1} \bm{r_{i}} \bm{r_{i}}^T$
				\State $\bm{D_{k-1}} \gets \bm{C_{k-1}}^{-1}, \bm{E_{k-1}} \gets [\bm{D_{k-1}}]_{1:2,1:2}$
				\If{${\rm mode} = \mathtt{3D}$}
				\State Solve (\ref{eq:min_eig_based_optim})-(\ref{eq:max_dist_along_min_eigvec}) to get $\bm{r^\ddagger} \gets \bm{r^{\rm (eig)}_{k,3D}}$.
				\ElsIf{${\rm mode} = \mathtt{2D}$}
				\State Solve (\ref{eq:eigvec_based_heuristic_2D_RNDOP})-(\ref{eq:eigvec_based_heuristic_z_part}) to get $\bm{r^\ddagger} \gets \bm{r^{\rm (eig)}_{k,2D}}$.
				\EndIf
				\For{$l \gets 1$ to $(k-1)$}
					\State $d_l \gets \| \bm{r^\ddagger} - \bm{r_{l}} \|_2$
				\EndFor	
				\State $d_{\rm min} \gets \min([d_1, d_2, \cdots, d_{k-1}])$
				\State $\triangleright$ If (\ref{eq:dist_btwn_anch_cnstrnt_no_trnsfrm}) is satisfied, we have a new valid anchor location
				\If{$d_{\rm min} \geq d_{\rm th}$} 
					\State $\bm{r_k} \gets \bm{r^\ddagger}, n \gets n + 1, \bm{V}(k) \gets 1$.
				\Else \Comment{perform random perturbation until (\ref{eq:dist_btwn_anch_cnstrnt_no_trnsfrm}) is satisfied}
					\State $\bm{r_k} \gets \bm{r^\ddagger}, d_{\rm min,c} \gets d_{\rm min} $ \Comment{Initialize $\bm{r_k}$ and minimum anchor separation}
					\For{$p \gets 0$ to $N_{\rm max}$}
						\State $\theta \sim \mathcal{U}[-\pi, \pi], \phi \sim \mathcal{U}\big[-\tfrac{\pi}{2}, \tfrac{\pi}{2} \big]$
						\State $\bm{r^\dagger} = \bm{r^\ddagger} + \eta d_{\rm th} \cdot \bm{a}(\theta, \phi)$ \Comment{perturb $\bm{r^\ddagger}$ by a distance $\eta d_{\rm th}\ (\eta > 1)$ along a random direction $(\theta, \phi)$.}
						\If{$\bm{r_l} \preccurlyeq \bm{r^\dagger} \preccurlyeq \bm{r_u}$}
							\For{$j \gets 1$ to $(k-1)$}
								\State $d^{\dagger}_j \gets \| \bm{r^{\dagger}} - \bm{r_j} \|_2$
							\EndFor
							\State $d_{\rm min} = \min([d^{\dagger}_1, d^{\dagger}_2, \cdots, d^{\dagger}_{k-1}])$
							\If {$d_{\rm min} \geq d_{\rm th}$}
								\State $\bm{r_k} \gets \bm{r^\dagger}, n \gets n + 1, \bm{V}(k) \gets 1$.
								\State \textbf{break}
							\ElsIf {$d_{\rm min} > d_{\rm min,c}$}
								\State $\bm{r_k} \gets \bm{r^\dagger},d_{\rm min,c} \gets d_{\rm min}, \bm{V}(k) \gets 0$
							\EndIf 
						\Else
							\State \textbf{continue}
						\EndIf 
					\EndFor
				\EndIf
				\State Lines 14-18 in Algorithm \ref{algo:sequential_minimax_gdop_anch_plcement}.
				\begin{comment} 
				\State $\bm{r_c} \gets \frac{1}{k} \sum_{i = 1}^{k} \bm{r_i}$. 
				\State $\bm{r_u} \gets (\bm{r_u} - \bm{r_c})$ and $\bm{r_l} \gets (\bm{r_l} - \bm{r_c})$. 
				\For{$i \gets 1$ to $N$} 
					\State $\bm{r_i} \gets (\bm{r_i} - \bm{r_c})$ 
				\EndFor
				\end{comment} sss
			\EndWhile
			\State $\bm{r'_c} \gets (\bm{r'_u} - \bm{r_u}) = (\bm{r_l} - \bm{r'_l})$ 
			\State $\mathcal{R}_{\rm val} = \emptyset$.
			\For{$i \gets 1$ to $k$}
				\State $\bm{r_i} \gets (\bm{r_i} + \bm{r'_c})$ \Comment{Transform all anchor coordinates to GCS.}
				\If{$\bm{V}(i)=1$}
					\State $\mathcal{R}_{\rm val} \gets \mathcal{R}_{\rm val} \cup \{ \bm{r_i}\}$.
				\EndIf 
			\EndFor 
			\State \Return $\mathcal{R}_{a}(N+N_{\rm a}) = \mathcal{R}_{a}(N) \cup \mathcal{R}_{\rm val}$.
		\end{algorithmic}
		\caption{Iterative Eigenvector Heursitic-based Iterative Anchor Placement for 2D and 3D Positioning}
		\label{algo:random_prtrb_min_anch_sep}
	\end{algorithm}
	\vspace{-10pt}
	\subsection{Computational Complexity and Convergence}\label{subsec:compute_complexity}
	\subsubsection{Algorithm 1}
	Since the considered optimization problems are non-convex, convergence to the global minima cannot be guaranteed. A rigorous proof of convergence to a local minima is beyond the scope of this work. However, we would like to emphasize that we use the Primal-Dual Interior-Point (PDIP) class of algorithms such as \cite{Waltz_Nocedal_Int_Pnt_LS_TR_MathPro_2006} based on the trust-region and line search-based approaches, which have robust convergence properties that ensure that each step makes progress towards a local minima \cite{Waltz_Nocedal_Int_Pnt_LS_TR_MathPro_2006}. 
	
	Let $d$ be the dimension in which the target is tracked. Hence, $d=2$ and $3$ for 2D and 3D positioning respectively. Since the optimization problems composed of Algorithm \ref{algo:sequential_minimax_gdop_anch_plcement} are nonlinear and non-convex, we solve them using PDIP-based methods. Here, the Newon step is the limiting operation requiring the most number of computations, and the corresponding worst-case complexity is $O(d^3)$. In addition, the complexity of the cost function computation for each single anchor addition problem is no greater than $O(d^3)$, since they are based on operations such as the eigen-decomposition, multiplication, and inversion of $\mathbb{S}^{d\times d}$ matrices \cite{Banks_eigdecomp_complexity_FOCS_2020}. For adding the $k^{\rm th}$ anchor, an $\epsilon$-accurate iteration can be computed in $O\big(\sqrt{k} \log\big(\frac{1}{\epsilon}\big) \big)$ steps \cite{Xiaoge_QNE_noncnovex_TWC_2013}, \cite{Xiaona_glob_conv_int_pnt_JCAM_2009}, \cite{boyd2004convex}. Therefore, a single sequential anchor addition step in Algorithm \ref{algo:sequential_minimax_gdop_anch_plcement} has a complexity of $O \big(d^3 \sqrt{k} \log\big(\frac{1}{\epsilon} \big) \big)$. Therefore, to add $N_a$ anchors sequentially to $N$ existing anchors, the computational complexity is $\sum_{n=N}^{N + N_a - 1} O\big(d^3\sqrt{n} \log \big(\frac{1}{\epsilon}\big) \big)$. Using the fact that $\sum_{i=1} x^{1/2}_i \leq \big(\sum_{i=1} x_i \big)^{1/2}$, the computational complexity can be simplified as $O \Big(d^3 \sqrt{N-1 + \frac{N_a(N_a + 1)}{2}} \log\big(\frac{1}{\epsilon} \big) \Big)$. If $N_a \gg N$, then the worst-case complexity can be accurately approximated using $O \big(d^3 N_a \log \big(\frac{1}{\epsilon} \big) \big)$.  
	
	\subsubsection{Algorithm 2} 	
	Due to the random perturbation step, convergence is not guaranteed for each iteration of the anchor addition scheme. As discussed in Section \ref{subsec:rand_perturb_eigvec}, this is alleviated by retaining the constraint-violating anchor locations and proceeding forward to add redundant additional anchors until a total of $N_{\rm a}$ constraint-satisfying anchor locations are found, and then removing the anchors that violate the deployment constraints. The number of removed anchors ($N_{\rm f}$) depends on the optimization problem parameters $\bm{r_u}, \bm{r_l},d_{\rm th}$, and $N_{\rm a}$. A rigorous analysis of the parameters' impact on the convergence is out of the scope of this paper, and will be addressed in a future work. Nevertheless, we have comprehensively analyzed the execution time performance of Algorithm \ref{algo:random_prtrb_min_anch_sep} using Monte-Carlo simulations in Section \ref{sec:num_results}.
	
	The computational complexity is determined by two key steps: (a) Eigen-decomposition of $\bm{D}$, which is given by $O(d^3)$, and (b) the random perturbation and minimum anchor distance comparison step, which has a complexity of $O(kN_{\rm max})$ for adding the $(k+1)^{\rm th}$ anchor. Therefore, the worst-case complexity for Algorithm \ref{algo:random_prtrb_min_anch_sep} is given by $O\big( (N_{\rm a} + N_{\rm f})d^3 + (N + \sum_{i=1}^{N_{\rm a} + N_{\rm f}} i) N_{\rm max} \big)$. If ($N_{\rm a} + N_{\rm f}) \gg N$, this simplifies to $O\big( (N_{\rm a} + N_{\rm f})d^3 + (N_{\rm a} + N_{\rm f})^2 N_{\rm max} \big)$.

	\begin{figure*}[!t]
		\centering
		\begin{subfigure}[t]{0.32\textwidth}
			\raggedleft
			\includegraphics[width=2.4in]{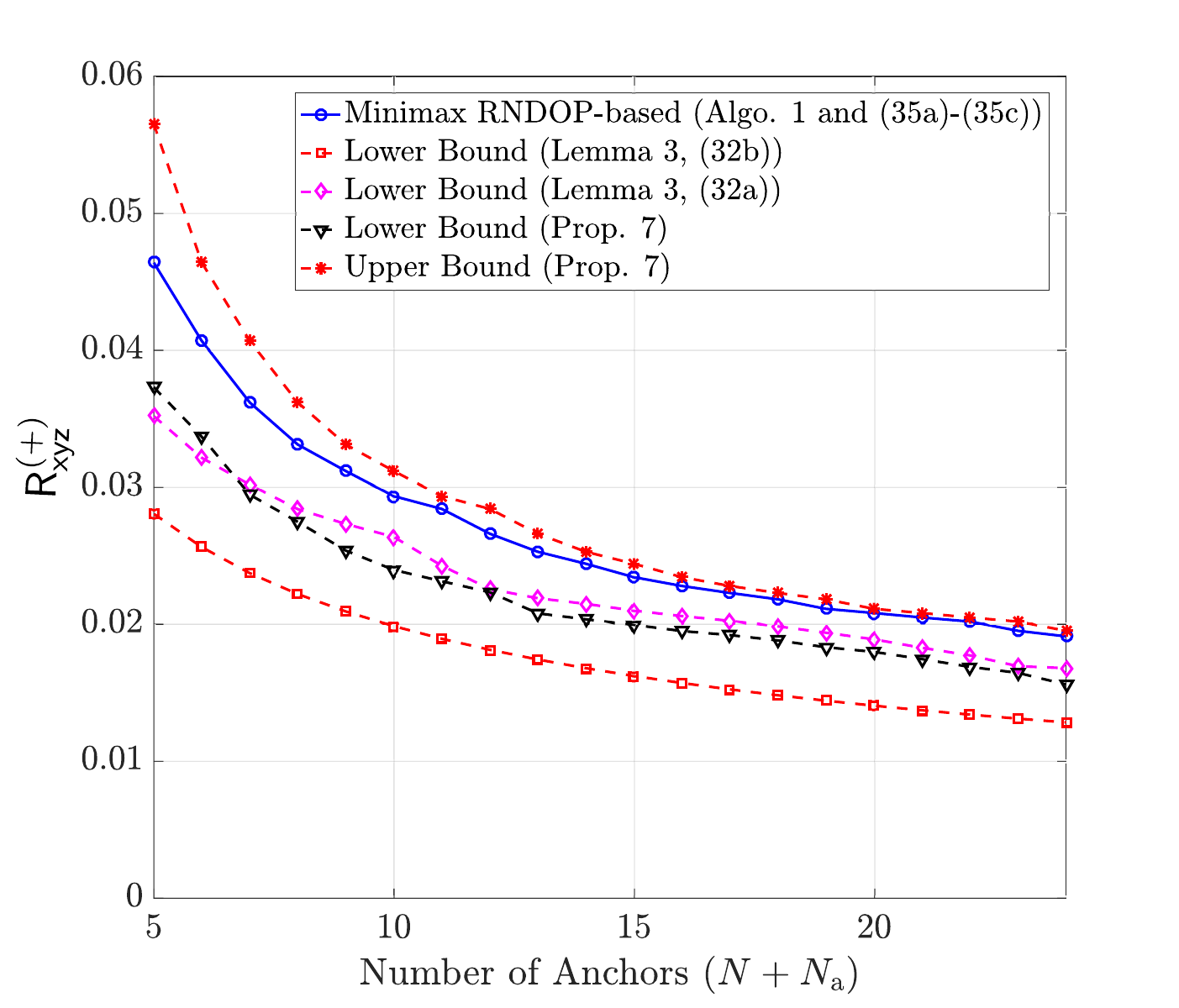}\\
			[-1ex]
			\caption{}
			\label{Fig3a_RNDOP_based_single_3D}
		\end{subfigure}
		~
		\begin{subfigure}[t]{0.32\textwidth}
			\centering
			\includegraphics[width=2.4in]{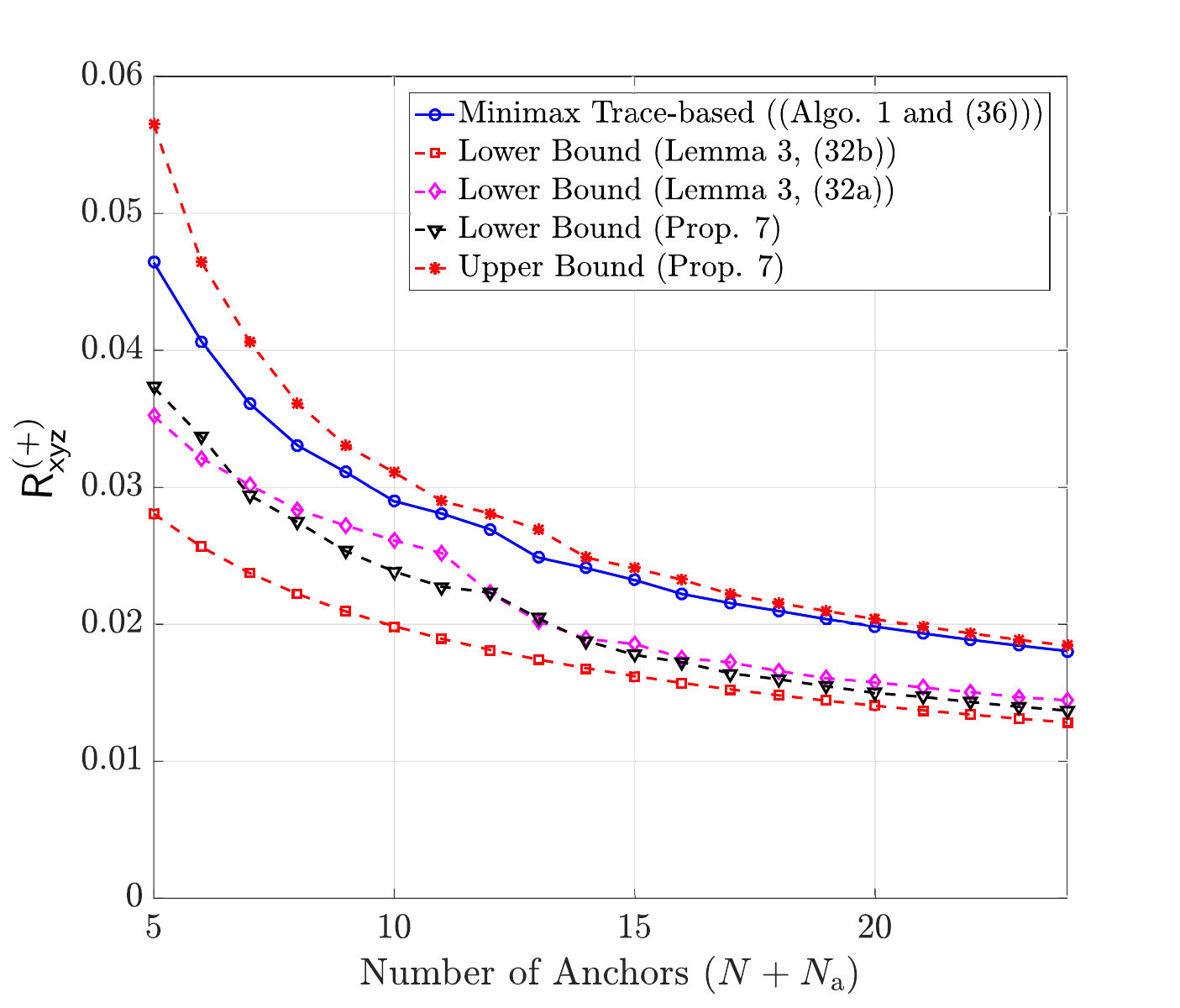}\\
			[-1ex]
			\caption{}
			\label{Fig3b_Trace_based_single_3D}		
		\end{subfigure}
		~
		\begin{subfigure}[t]{0.32\textwidth}
			\centering
			\includegraphics[width=2.4in]{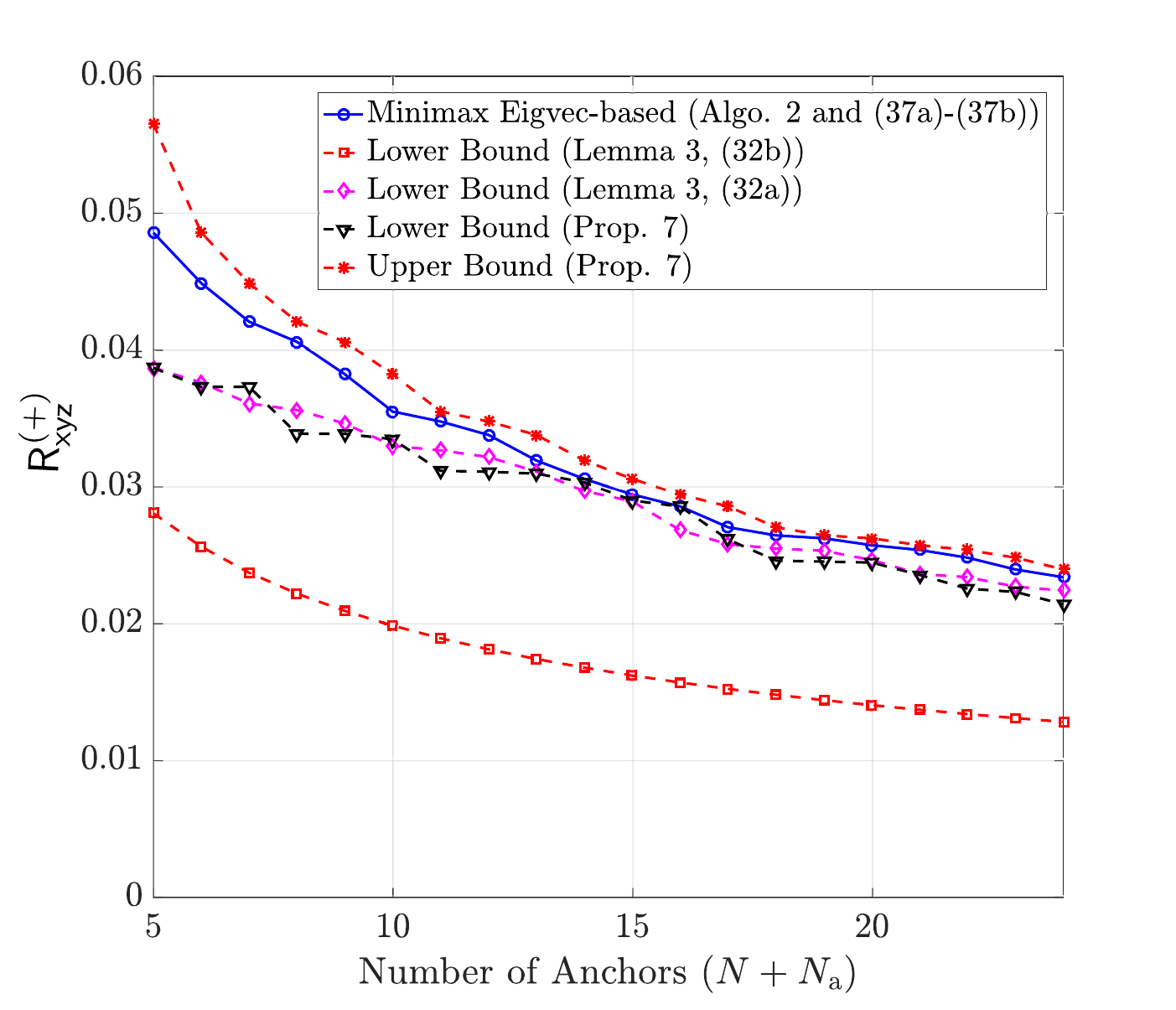}\\
			[-1ex]
			\caption{}
			\label{Fig3c_Trace_based_single_3D}		
		\end{subfigure}
		\caption{Variation of ${\mathsf{R}^{(+)}_\mathsf{xyz}}$ and its bounds as a function of the number of total anchors ($M=N + N_{\rm a}$) under the (a) Minimax RNDOP-optimal scheme (Lemma \ref{lemma:iterative_3D_RNDOP_based}), (b) Minimax Trace-based scheme (Lemma \ref{lemma:iterative_trace_based_3d}), and (c) Minimax Eigenvector-based scheme (Lemma \ref{lemma:min_eigvalue_anch_plcmnt_3D}), for beyond convex hull 3D localization.}
		\label{Fig3_3d_schemes_RNDOP_single}
	\end{figure*}

	\begin{figure*}[!t]
		\centering
		\begin{subfigure}[t]{0.32\textwidth}
			\raggedleft
			\includegraphics[width=2.4in]{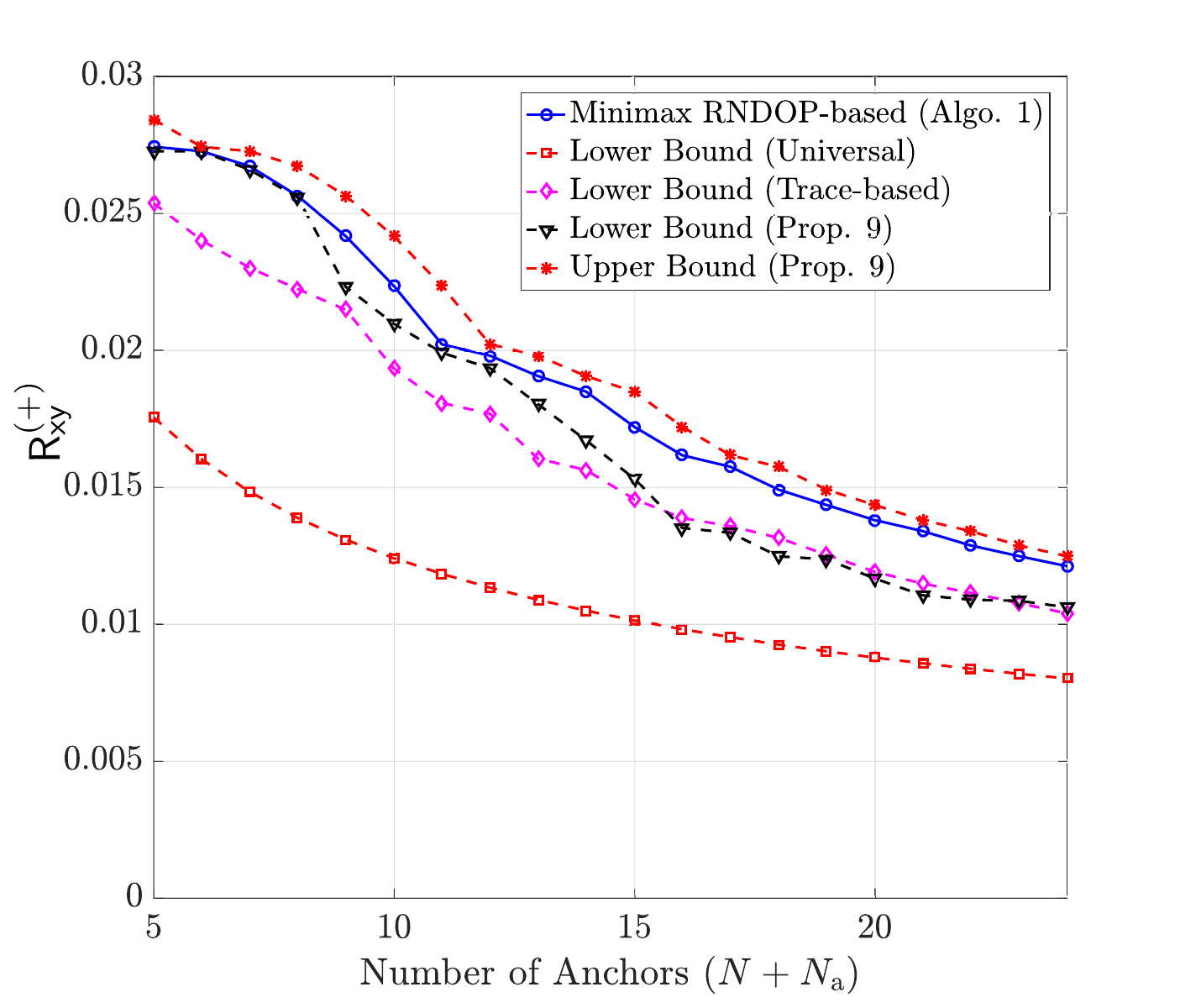}\\
			[-1ex]
			\caption{}
			\label{Fig4a_RNDOP_based_single_2D}
		\end{subfigure}
		~
		\begin{subfigure}[t]{0.32\textwidth}
			\centering
			\includegraphics[width=2.4in]{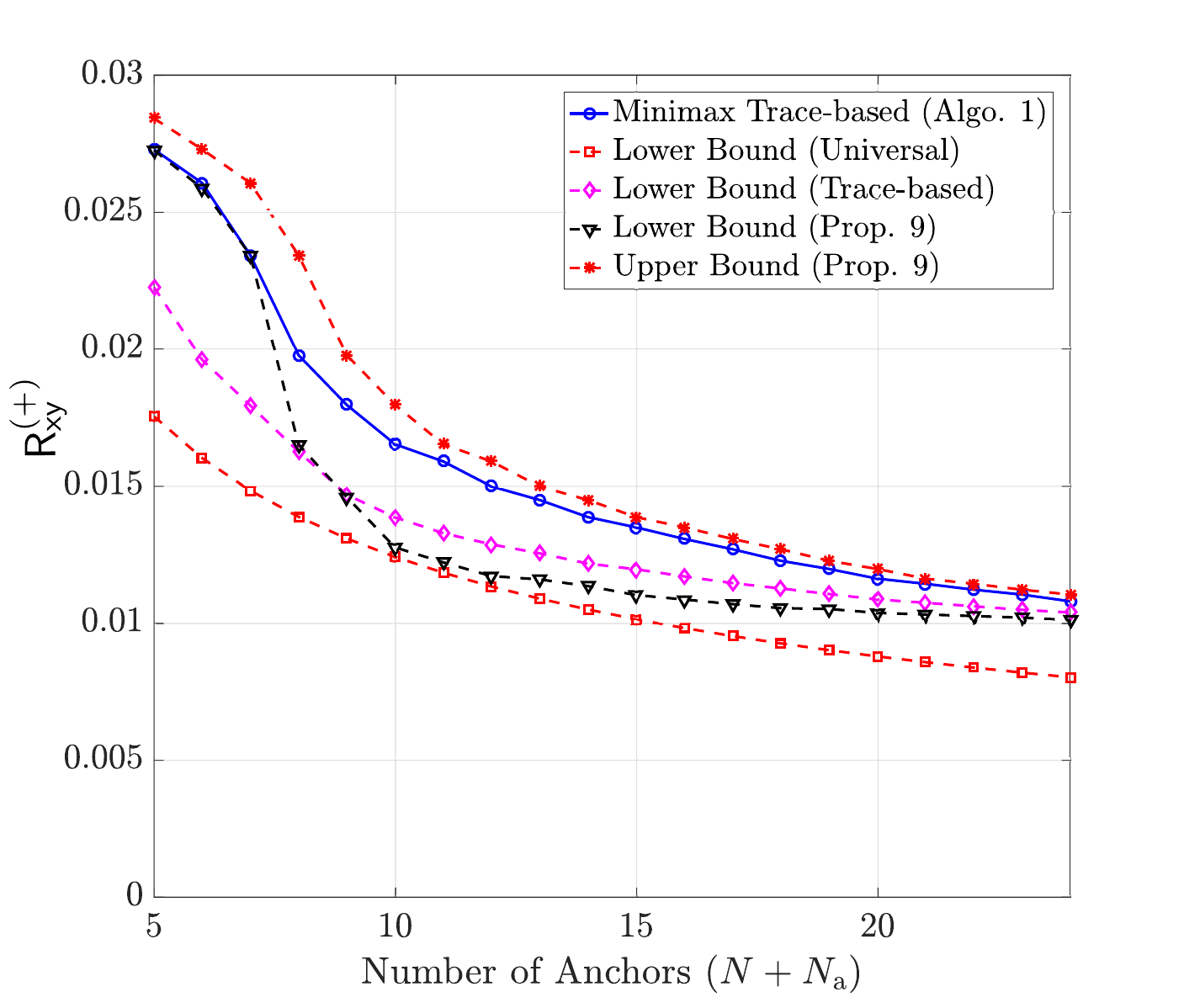}\\
			[-1ex]
			\caption{}
			\label{Fig4b_Trace_based_single_2D}		
		\end{subfigure}
		~
		\begin{subfigure}[t]{0.32\textwidth}
			\centering
			\includegraphics[width=2.4in]{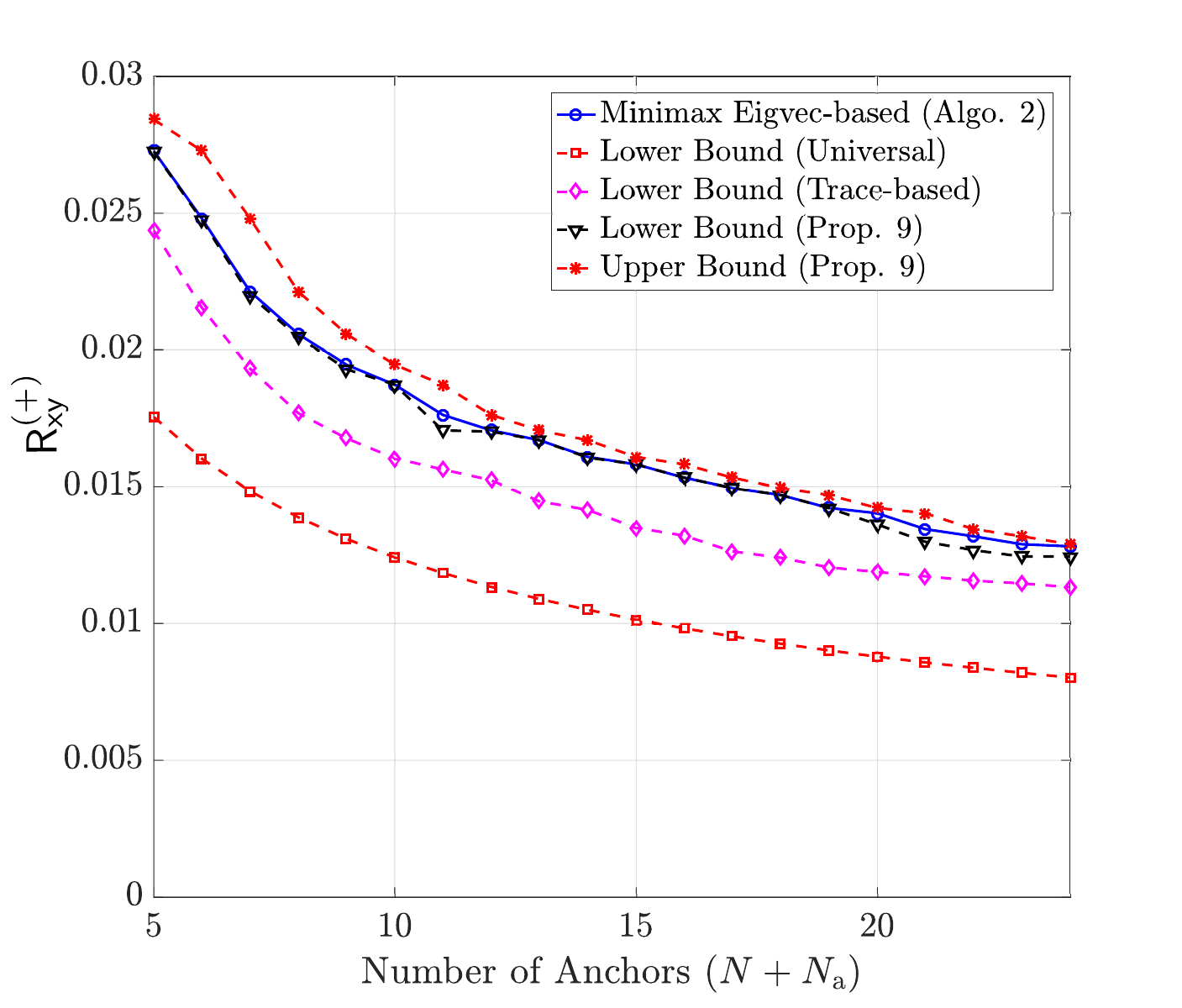}\\
			[-1ex]
			\caption{}
			\label{Fig4c_Trace_based_single_2D}		
		\end{subfigure}
		\caption{Variation of ${\mathsf{R}^{(+)}_\mathsf{xy}}$ and its bounds as a function of the number of anchors ($M=N + N_{\rm a}$) under the (a) Minimax RNDOP-optimal scheme (Lemma \ref{lemma:iterative_rndop_2d}), (b) Minimax Trace-based scheme (Lemma \ref{lemma:iterative_trace_based_2d}), and (c) Minimax Eigenvector-based (Lemma \ref{lemma:iterative_eigvec_based_2d}) scheme, for 2D (XY) beyond convex hull localization.}
		\label{Fig4_2d_schemes_RNDOP_single}
	\end{figure*}
	\vspace{-10pt}
	\subsection{Impact of Position Uncertainty on RNDOP}
	So far, we have assumed perfect knowledge of the anchor locations. However, this is often not the case in UAV-based positioning applications, due to (a) motion jitter \cite{Banagar_UAV_Wobble_TVT_2020} and (b) tracking error \cite{Zamurad_UAV_TrackError_TCS_2015}. The following result derives the change in $\mathsf{R}^{(+)}_\mathsf{xyz}$ and $\mathsf{R}^{(+)}_\mathsf{xy}$ due to anchor position uncertainty.  
	
	\begin{lemma} \label{lemma:delta_RNDOP_pert_3D_2D}
	Let $\mathsf{R}^{(+)}_\mathsf{xyz}(\mathcal{R}_a)$/$\mathsf{R}^{(+)}_\mathsf{xy}(\mathcal{R}_a)$ denote the maximum 3D/2D RNDOP respectively, for anchor configuration $\mathcal{R}_a = \{\bm{r_1},\bm{r_2},\cdots,\bm{r_N} \}$ which satisfy (\ref{eq:box_constraints})-(\ref{eq:anch_centroid_constraints}). Let the anchor position uncertainty be $\bm{\Delta r_i},i=1,2,\cdots,N$, resulting in an anchor configuration $\mathcal{R}'_a = \{\bm{\tilde{r}_1}, \bm{\tilde{r}_2}, \cdots, \bm{\tilde{r}_{N}} \}$ where $\bm{\tilde{r}_i} = \bm{r_i} + \bm{\Delta r_i}$. Then, the minimax RNDOP $\mathsf{R}^{(+)}(\mathcal{R}'_a)$ is given by 
	\begin{align}
		\label{eq:RNDOP_max_3D_pert}
		& \mathsf{R}^{(+)}_\mathsf{xyz}(\mathcal{R}'_a) = \Big(\sum\nolimits_{i=1}^3 \big[\lambda_i(\bm{D}) + \bm{v}^T_{\bm i}(\bm{D}) \bm{\Delta D} \bm{v_i}(\bm{D})\big] \nonumber \\
		& - \underset{i=\{1,2,3\}}{\min}\big(\lambda_i(\bm{D}) + \bm{v}^T_{\bm i}(\bm{D}) \bm{\Delta D} \bm{v_i}(\bm{D}) \big) \Big)^{1/2}, \\
		\label{eq:RNDOP_max_2D_pert}
		& \mathsf{R}^{(+)}_\mathsf{xy}(\mathcal{R}'_a) = \Big(\sum\nolimits_{j=1}^2 \big[\lambda_j(\bm{E}) + \bm{v}^T_{\bm j}(\bm{E}) \bm{\Delta E} \bm{v_j}(\bm{E})\big] \nonumber \\
			& -  \underset{j=\{1,2\}}{\min}\big(\lambda_j(\bm{E}) + \bm{v}^T_{\bm j}(\bm{E}) \bm{\Delta E} \bm{v_j}(\bm{E}) \big)\Big)^{1/2},
	\end{align}
	\begin{align}
		\label{eq:perturbed_matrices}
		& \text{ where }\bm{\Delta D}= -\bm{D}(\bm{D} + \bm{\Delta C}^{-1})^{-1} \bm{D}, \bm{\Delta E} = [\bm{\Delta D}]_{1:2,1:2}, \nonumber \\
		& \bm{\Delta C} = \sum\nolimits_{i=1}^N \Big(\bm{\Delta r_i \Delta r}^T_{\bm{i}} + \bm{\Delta r_i r}^T_{\bm{i}} + \bm{r_i \Delta r}^T_{\bm{i}}  \Big) - N \bm{r_{\mu} r}^T_{\bm{\mu}}, \nonumber \\
		& \text{ and } \bm{r_{\mu}} = \tfrac{1}{N} \sum\nolimits_{i=1}^N \bm{\Delta r_i},
	\end{align}
	if the eigenvalues $\lambda_i$ are non-degenerate.
	\begin{proof}
		We start with the 3D positioning case. Using Lemma (\ref{lemma:asympt_bounds_DOP_xyz_xy}) and (\ref{eq:RNDOP_def_3D_2D}), we have $\mathsf{R}^{(+)}_\mathsf{xyz}(\mathcal{R}'_a) = \sqrt{{\rm tr}(\bm{D_\epsilon}) - \lambda_{-}(\bm{D_\epsilon})}$, 
		\begin{align*} 
		& \text{where } \bm{D_\epsilon} = \bm{C_\epsilon}^{-1}, \bm{C_\epsilon} = \sum\nolimits_{i=1}^{N} (\bm{\tilde{r}_i} - \bm{\tilde{r}_\mu})(\bm{\tilde{r}_i} - \bm{\tilde{r}_\mu})^T, \text{ and } \\
		& \bm{\tilde{r}_\mu} = \tfrac{1}{N} \sum\nolimits_{i=1}^{N}\bm{\tilde{r}_i}.
		\end{align*} 
		We can then represent $\bm{C_\epsilon} = \bm{C} + \bm{\Delta C}$, where $\bm{C} = \sum\nolimits_{i=1}^N \bm{r_i r}^T_{\bm{i}}$. Expand $\bm{C_\epsilon}$ and simplifying using the fact that $\frac{1}{N}\sum_{i=1}^N \bm{r_i} = \bm{0}$, we get $\bm{\Delta C} = \sum\nolimits_{i=1}^N \Big(\bm{\Delta r_i \Delta r}^T_{\bm{i}} + \bm{\Delta r_i r}^T_{\bm{i}} + \bm{r_i \Delta r}^T_{\bm{i}}  \Big) - N \bm{r_{\mu} r}^T_{\bm{\mu}}$.
		Applying the matrix inversion lemma \cite{horn2012matrix} on $\bm{C_\epsilon}^{-1}$, we get\footnote{Clearly, this result is valid if and only if $\bm{\Delta C}$ is full-rank. If it is rank-deficient, matrix inversion can be carried out using the eigendecomposition of $\bm{\Delta C}$ and using the Sherman-Morrison-Woodbury Lemma \cite{horn2012matrix} sequentially. A rigorous analysis of ``rank-deficient $\bm{\Delta C}$'' scenarios is beyond the scope of this paper, and will be deferred to a future work. } $\bm{D_\epsilon} \triangleq (\bm{C} + \bm{\Delta C})^{-1} = \bm{C}^{-1} - \bm{C}^{-1}(\bm{\Delta C}^{-1} + \bm{C}^{-1})^{-1} \bm{C}^{-1}$. Using $\bm{D} \triangleq \bm{C}^{-1}$ (Lemma \ref{lemma:equi_DOP_xyz_and_xy}), defining $\bm{D}_\epsilon \triangleq \bm{D} + \bm{\Delta D}$ and simplifying, we get (\ref{eq:perturbed_matrices}). If the eigenvalues of $\bm{D}$ are  non-degenerate, we can apply the first-order approximation of results from eigenvalue perturbation theory \cite{nakatsukasa2017off} to get 
		\begin{align}
			\lambda_i(\bm{D_\epsilon}) \approx \lambda_i(\bm{D}) + \bm{v}^T_{\bm i}(\bm{D}) \bm{\Delta D} \bm{v_i}(\bm{D}).
		\end{align}
	Substituting this in the expression for $\mathsf{R}^{(+)}_\mathsf{xyz}(\mathcal{R}'_a)$ and simplifying, we get (\ref{eq:RNDOP_max_3D_pert}). For the 2D positioning case, note that $\bm{E_\epsilon} \triangleq [\bm{C_\epsilon}^{-1}]_{1:2,1:2}$. Using similar steps as above, and observing that $\bm{E_\epsilon} \in \mathbb{S}^{2 \times 2}_{++}$, we obtain (\ref{eq:RNDOP_max_2D_pert}) if the eigenvalues of $\bm{E} = [\bm{C}^{-1}]_{1:2,1:2}$ are non-degenerate. 
	\end{proof}
	\end{lemma}
	\vspace{-10pt}
	\section{Numerical Results} \label{sec:num_results}
	We consider a beyond convex hull localization scenario, as shown in Fig. \ref{Fig1_Anchors_FarAway_Regime}. In the following subsections, we present numerical results to (a) verify the iterative upper and lower bounds, and (b) demonstrate the error and computational performance of our proposed iterative anchor placement algorithms in 2D and 3D localization scenarios.
	\vspace{-8pt}
	\subsection{Comparison of Maximum RNDOP Performance}
	Recall that for far-away targets outside the convex hull, a target with the maximum RNDOP \emph{tends to have the worst localization error on average} when the range error statistics is similar for all targets. Here, we consider a UAV-based localization scenario, where the UAVs are constrained to be located in the box defined by $\bm{r_l} = [-30 \text{ m}, -20 \text{ m}, -10 \text{ m}]$, and $\bm{r_u} = [30 \text{ m}, 20 \text{ m}, 10 \text{ m}]$, and the minimum anchor separation distance $d_{\rm th} = 4.472$ m.  Each scheme is initialized with $N=4$ anchors, whose locations are obtained using a Monte-Carlo method. Specifically, the initial anchor configuration is set to one that has the minimum $\mathsf{R}^{(+)}_\mathsf{xyz}$ (for 3D)/$\mathsf{R}^{(+)}_\mathsf{xy}$ (for 2D) after $10^5$ Monte-Carlo trials. 

	From Fig. \ref{Fig3_3d_schemes_RNDOP_single}, we observe that $\mathsf{R}^{(+)}_\mathsf{xyz}$ decreases sublinearly as a function of the total number of anchors ($M$), for all the proposed anchor placement schemes. While the performance of the minimax RNDOP-optimal (Fig. \ref{Fig3a_RNDOP_based_single_3D}) and trace-based (Fig. \ref{Fig3b_Trace_based_single_3D}) schemes are close, the eigenvector-based scheme (Fig. \ref{Fig3c_Trace_based_single_3D}) has noticeably higher $\mathsf{R}^{(+)}_\mathsf{xyz}$. For the considered scenario, the minimax RNDOP and trace-based iterative schemes yield similar $\mathsf{R}^{(+)}_\mathsf{xyz}$ performance since $\lambda_1(\bm{D_k}) \ll \sum_{i=\{2,3\}} \lambda_i(\bm{D_k})$ at initialization, and holds true as the anchor index $k$ increases. 
	It is worthwhile to point out that the trace-based scheme slightly outperforms its minimax RNDOP counterpart. This is because it \emph{minimizes the sum of all the eigenvalues}, unlike the minimax RNDOP scheme which \emph{minimizes the sum of the two largest eigenvalues}. For instance, since the RNDOP-based scheme incentivizes minimization of the top two eigenvalues at each iteration, the minimum eigenvalue risks remaining stagnant over multiple iterations, eventually ending up becoming the largest, or second-largest eigenvalue. In other words, accounting for all eigenvalues during iterative minimization makes the scheme more robust against eigenvalue ordering changes that occur over multiple iterations.  

	On the other hand, the resultant $\mathsf{R}^{(+)}_\mathsf{xyz}$ is relatively large for the eigenvector-based heuristic scheme since it sequentially minimizes the numerator and maximizes the denominator of the trace-based heuristic, rather than minimizing the fractional term jointly. However, the random perturbation step introduces stochasticity into the RNDOP performance which is non-trivial to analyze from a theoretical standpoint. Nevertheless, the minimax RNDOP is $20-30\%$ higher when compared to the other two schemes. Next, we observe that the iterative upper and lower bounds in (\ref{eq:up_low_bounds_minimax_dop}) are fairly tight for all the schemes. Finally, the universal lower bound (\ref{eq:RNDOP_bound_universal}) is quite loose, whereas the tightness of its trace-based variant (\ref{eq:RNDOP_bound_univ_trace}) is similar to that of (\ref{eq:up_low_bounds_minimax_dop}). This is consistent due to the fact that (\ref{eq:RNDOP_bound_universal}) is itself a lower bound on (\ref{eq:RNDOP_bound_univ_trace}).

	From Fig. \ref{Fig4_2d_schemes_RNDOP_single}, we observe that the trace-based scheme outperforms the other two methods in terms of $\mathsf{R}^{(+)}_\mathsf{xy}$. The $\mathsf{R}^{(+)}_\mathsf{xy}$ performance of the trace-based iterative scheme is superior when compared to the other. This can also be attributed to the \textit{`long-term RNDOP minimization behavior'} of trace-based cost functions outlined in the previous few paragraphs. In addition, we observe that the universal bound in Corollary \ref{corollary:dop_universal_bounds_2d} is fairly loose, and the trace-based lower bound in (\ref{eq:RNDOP_xy_bound_univ_and_trace}) is relatively tight. Finally, we observe that the iterative upper and lower bounds (Proposition \ref{prop:lb_ub_rndop_xy_iterative}) is tight for all the proposed 2D-optimal anchor placement schemes. 
	
	\begin{figure*}[t]
		\centering
		\begin{subfigure}[t]{0.48\textwidth}
			\raggedleft
			\includegraphics[width=3.3in]{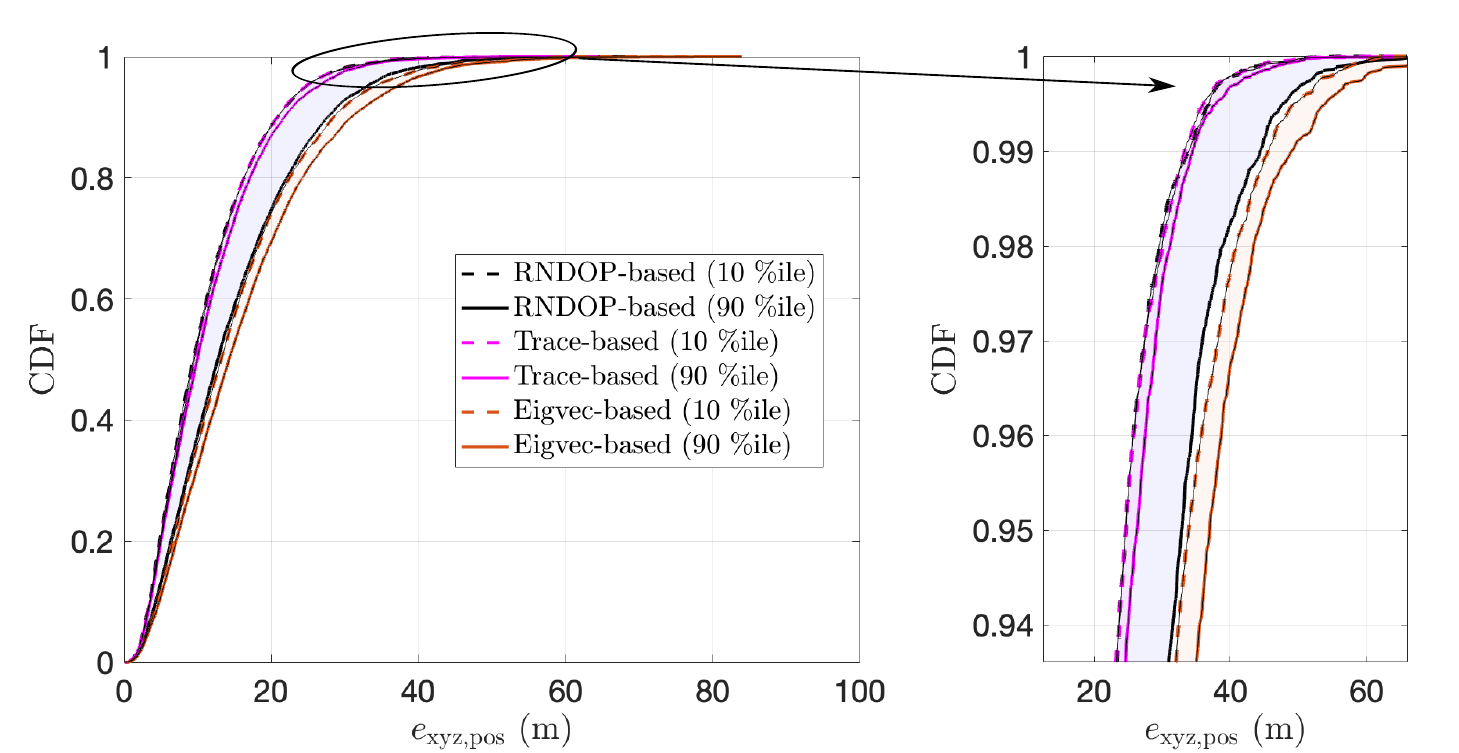}\\
			[-2ex]
			\caption{}
			\label{Fig5a_pos_error_3d_NLS}
		\end{subfigure}
		~
		\begin{subfigure}[t]{0.48\textwidth}
			\centering
			\includegraphics[width=3.3in]{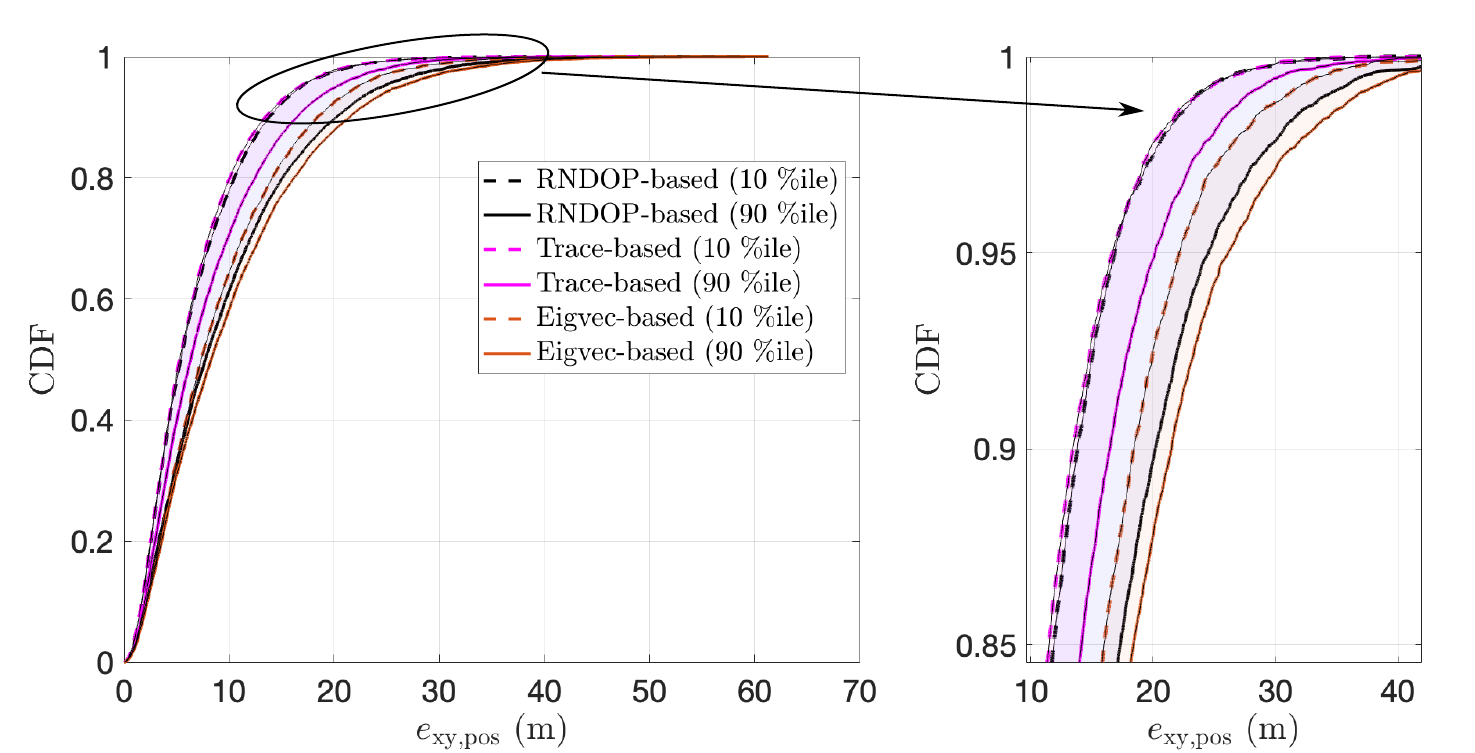}\\
			[-2ex]
			\caption{}
			\label{Fig5b_pos_error_2d_NLS}		
		\end{subfigure}
		\caption{(a) CDF of 3D positioning error for the $\mathcal{R}_{i_{10}}$ (good, denoted by dashed lines) and $\mathcal{R}_{i_{90}}$ (bad, denoted by solid lines) anchor configurations obtained by each iterative method (3D optimal schemes), for targets distributed uniformly at random in a sphere of radius 200m. (b) CDF of 2D positioning error for the $\mathcal{R}_{i_{10}}$ (good, denoted by dashed lines) and $\mathcal{R}_{i_{90}}$ (bad, denoted by solid lines) anchor configurations obtained by each iterative method (2D optimal schemes), for targets distributed uniformly on the XY-plane at random in a circle of radius 200m.}
		\label{Fig5_NLS_pos_error_2D_and_3D}
	\end{figure*}
    \vspace{-10pt}
    
	\subsection{Comparison of 3D and 2D Localization Error}
	To compare the 3D/2D localization error performance of the proposed methods, we performed $N_{\rm mc, algo} = 500$ Monte-Carlo (MC) simulations for each iterative anchor placement scheme. The following steps are undertaken in each MC trial. 
	
	\subsubsection{Initialization} The initial $N-\text{anchor}$ configuration is obtained using $N_{\rm mc, init} = 10^5$ trials. In each trial, $N=4$ anchors are placed uniformly at random in the box specified by $\bm{r_l}$ and $\bm{r_u}$. The initial anchor configuration corresponds to the one that (a) satisfies the minimum anchor separation ($d_{\rm th}$) constraint and (b) has the lowest $\mathsf{R}^{(+)}_\mathsf{xyz}$ (for 3D localization) or $\mathsf{R}^{(+)}_\mathsf{xy}$ (for 2D localization).
	\subsubsection{Iterative Anchor Placement} For each MC trial, the proposed iterative anchor placement methods are executed using Algorithm \ref{algo:sequential_minimax_gdop_anch_plcement} (for the Minimax RNDOP and the Trace-based heuristic methods) and Algorithm \ref{algo:random_prtrb_min_anch_sep} (for the Eigenvector heuristic-based scheme) to add $N_{\rm a} = 20$ additional anchors to the initial $N=4$ anchors in the network.
	
	\subsubsection*{Configuration Selection}	
	For characterizing the error performance of each method, we choose two anchor configurations per method: one configuration that corresponds to low max-RNDOP (low localization error), and another that corresponds to high max-RNDOP (high localization error). Let the max-RNDOP of the anchor configuraton obtained in all the $N_{\rm mc, algo}$ MC trials be ordered in descending order as $\mathsf{R}^{(+)}_1 \leq \mathsf{R}^{(+)}_2 \leq \cdots \leq \mathsf{R}^{(+)}_{N_{\rm mc,algo}}$, such that the corresponding anchor configuration is $\mathcal{R}_{i},i=1,2,\cdots, N_{\rm mc, algo}$. The \emph{``good''} anchor configuration is chosen to correspond to the 10-percentile max-RNDOP (i.e. $\mathsf{R}^{(+)}_{i_{10}}$, where $i_{10} = \lfloor 0.1\times N_{\rm mc,algo} \rfloor$), and the \emph{``bad''} anchor configuration corresponds to the 90-percentile max-RNDOP (i.e. $\mathsf{R}^{(+)}_{i_{90}}$, where $i_{90} = \lceil 0.9\times N_{\rm mc,algo} \rceil$).
	
	\subsubsection*{Random Target Deployment} To evaluate the efficacy of the good and bad anchor configurations under investigation, $N_{\rm targ} = 10^4$ targets are placed in the network. 
	\begin{enumerate}
		\item 3D Localization: For evaluating the 3D localization performance, each target is placed uniformly at random in a sphere of radius $r_{\rm cov} = 200$ m centered at the origin. 
		\item 2D Localization: For evaluating the 2D localization performance, each target is placed uniformly at random in a circle of radius $r_{\rm cov} = 200$ m on the XY plane centered at the origin. 
	\end{enumerate}
	Note that even though we are concerned with beyond convex-hull localization scenarios, we relax this constraint for evaluation purposes and consider the localization error performance across all possible target locations within the spherical/circular coverage region for 3D/2D localization respectively. 
	\subsubsection*{Localization using NLS} The target performs $N + N_{\rm a} = 24$ ToA-based range measurements using the relation $\hat{r}_i = \max(0, r_i + w_i)$, where $w_i \sim \mathcal{N}(b, \sigma^2_w)$, $b=1$ m, and $\sigma_w = 6$ m for $i=1,2,\cdots,N + N_{\rm a}$, where $r_i$ is the actual range, and $w_i$ the i.i.d range error. The location is then estimated using the Nonlinear Least Squares (NLS) algorithm \cite{buehrer2019handbook}.
	
	\subsubsection*{Localization Performance Comparison} 
	To ensure a fair comparison of the proposed schemes, the sources of randomness is controlled. In other words, the initial anchor locations, target locations, and the range errors are the same for the $i^\text{th}$ MC trial ($i=1,2,\cdots,N_{\rm mc, algo}$) of each method. Fig. \ref{Fig5_NLS_pos_error_2D_and_3D} shows the CDF of 3D and 2D positioning error for the good and bad anchor configurations obtained using all the proposed anchor placement schemes. From Fig. \ref{Fig5a_pos_error_3d_NLS}, we observe that the trace-based heuristic has the lowest 3D positioning error ($e_{\rm xyz, pos}$) when compared to the other schemes. Furthermore, the error performance achieved by the good and bad anchor condigurations are very close to each other. Since the initial conditions is the root cause of different max-RNDOP, the small performance difference signifies that the trace-based method is robust to initial conditions. Meanwhile, the error performance of the eigenvector-based heuristic is relatively worse when compared to the other two. Similar to its trace-based counterpart, we observe that it is robust to initial conditions as well. On the other hand, the $\mathcal{R}_{i_{10}}$ and $\mathcal{R}_{{i_{90}}}$ minimax RNDOP anchor configurations have a large difference in error performance, signifying sensitivity to initial conditions. One one end, its $\mathcal{R}_{i_{10}}$ configuration has similar performance as that of the trace-based method, whereas on the other end, its $\mathcal{R}_{i_{90}}$ configuration has simiilar performance as that of the eigenvector-based method.
	
	We observe a similar trend in 2D positioning error ($e_{\rm xy, pos}$) from Fig. \ref{Fig5b_pos_error_2d_NLS}. However, the minimax-RNDOP and eigenvector-based methods have significant overlap in their 2D positioning error distributions, unlike the 3D case.  
			
	\begin{figure}[!t]
		\centering
		\includegraphics[width=2.9in]{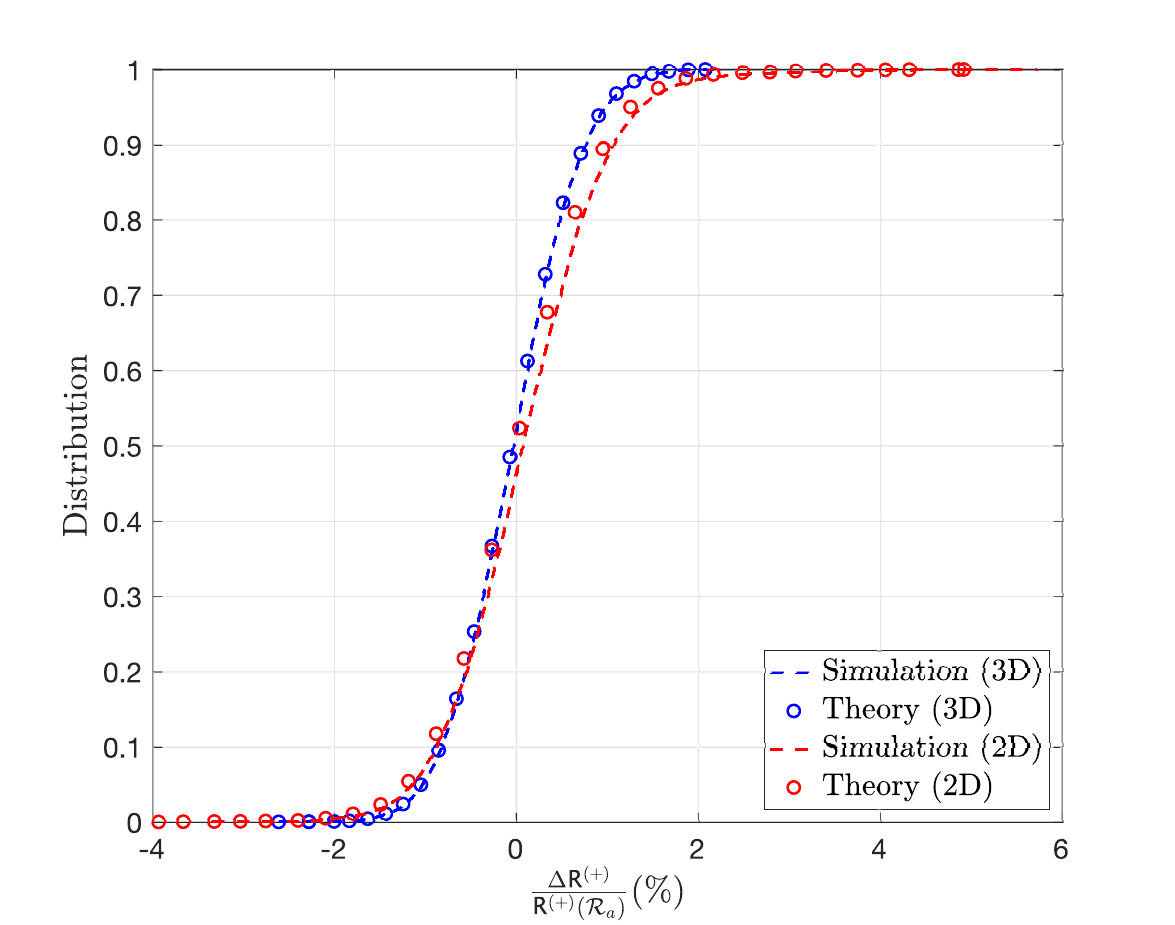}\\
		[-2ex]
		\caption{Distribution of the fractional change in minimax RNDOP $(\Delta \mathsf{R}^{(+)}/\mathsf{R}^{(+)}(\mathcal{R}_a) )$ due anchor position uncertainty of $\bm{\Delta r_i} = [\Delta x_i, \Delta y_i, \Delta z_i]^T$, for (a) 3D (blue) , and (b) 2D (red) beyond convex hull localization scenarios. The anchor position uncertainty is distributed as $\Delta x_i \sim {\rm U}[-1 \text{m},1\text{m}], \Delta y_i \sim {\rm U}[-1 \text{m},1\text{m}]$, and $\Delta z_i \sim {\rm U}[-1 \text{m},1\text{m}]$ for $i=1,2,\cdots,N$.}
		\label{fig:frac_change_rndop_2d_3d}
	\end{figure}
	
	\begin{table*}[t]
		\caption{Benchmarking results: Median computation time (in seconds) for $N=20$ on different computing platforms}
		\label{Tab1_median_compute_time}
		\centering 
		\begin{tabular}{|c | c | c | c | c | c | c|}
			\hline
			\multirow{2}{*}{\textbf{Computing Platform}} & \multicolumn{2}{c |}{\textbf{Minimax RNDOP}} & \multicolumn{2}{c |} {\textbf{Trace-based}} &  \multicolumn{2}{c|}{\textbf{Eigenvector-based}} \\
			\cline{2-7}
			& \textbf{3D} & \textbf{2D} & \textbf{3D} & \textbf{2D} & {\color{blue}\textbf{3D}} & {\color{magenta} \textbf{2D}} \\
			\hline
			Windows 10, Intel i7-8700T, 32 GB Memory, MATLAB 9.0 & 4.43 & 4.90 & 6.54 & 7.90 & {\color{blue} \textbf{0.099}} & {\color{magenta} \textbf{0.0082}} \\
			\hline
			Ubuntu 20.04, Intel i5-7500T, 16 GB Memory, MATLAB 9.0 & 4.40 & 4.85 & 6.49 & 7.57 & {\color{blue} \textbf{0.091}} & {\color{magenta} \textbf{0.0101}} \\
			\hline
			MacOS Monterey, Apple M1 Pro, 16 GB Memory, MATLAB 9.12 & 2.15 & 2.77 & 3.47 & 4.26 & {\color{blue} \textbf{0.016}} & {\color{magenta} \textbf{0.0027}} \\		
			\hline
		\end{tabular} \\
	[-3ex]
	\end{table*}
	
	\begin{figure*}[t]
		\centering
		\begin{subfigure}[t]{0.48\textwidth}
			\raggedleft
			\includegraphics[width=2.8in]{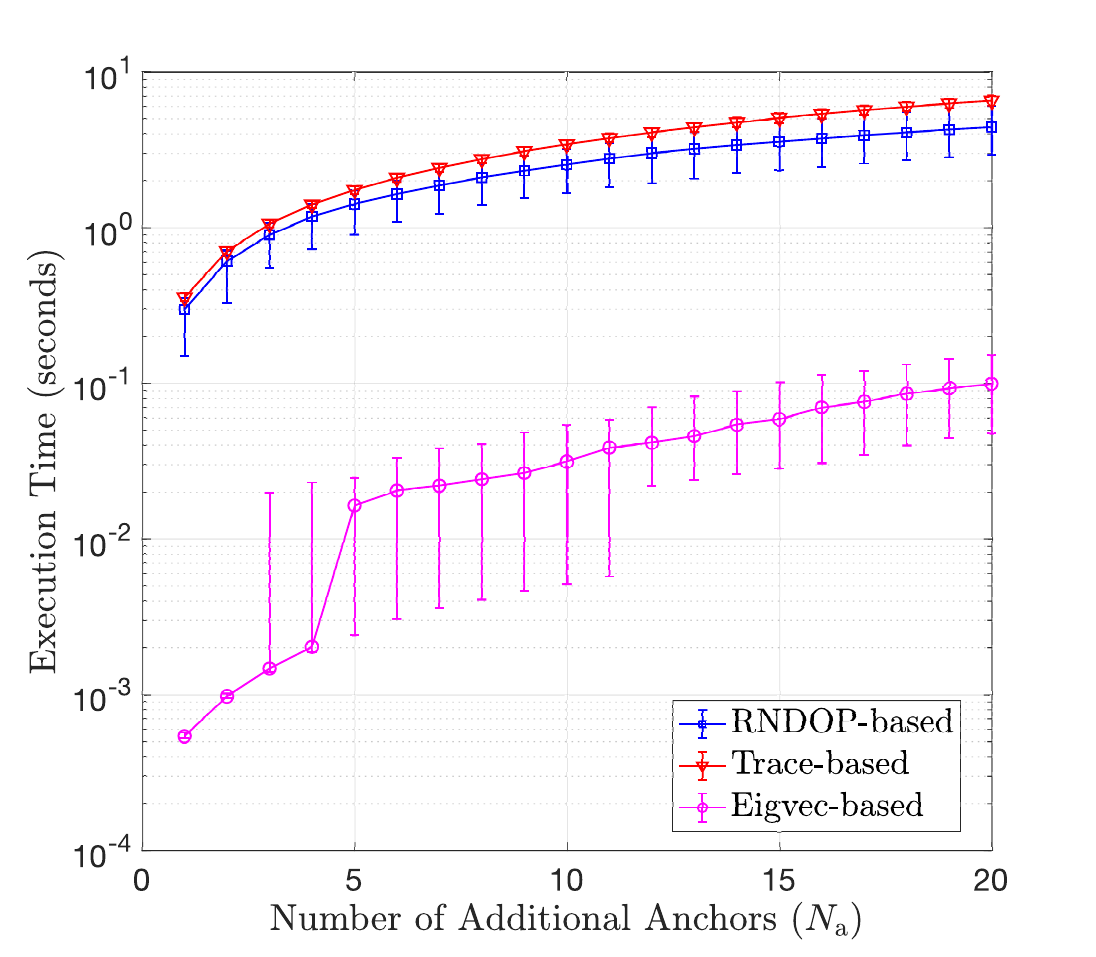}\\
			[-2ex]
			\caption{}
			\label{Fig6a_runtimes_errorbars_3d}
		\end{subfigure}
		~
		\begin{subfigure}[t]{0.48\textwidth}
			\centering
			\includegraphics[width=2.8in]{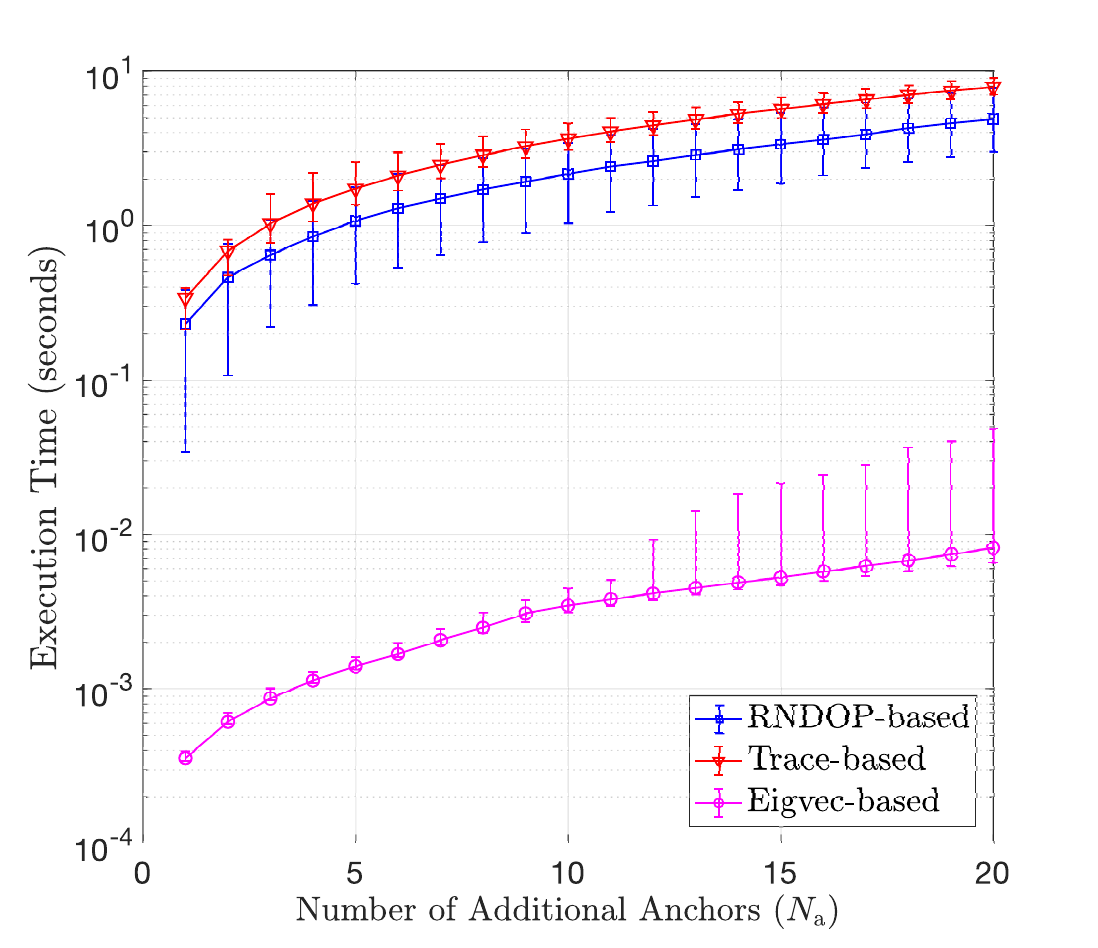}\\
			[-2ex]
			\caption{}
			\label{Fig6b_runtimes_errorbars_2d}		
		\end{subfigure}
		\caption{Comparison of $t_{\rm exec}$ as a function of $N_{\rm a}$, for the proposed iterative anchor placement schemes, for (a) 3D localization scenarios, and (b) 2D localization scenarios. The error bars correspond to the 10 percentile (lower) and 90 percentile (upper) values of the algorithm execution time.}
		\label{Fig6_runtimes_prctile_plots}
	\end{figure*}

	\begin{table}[t]
		\caption{Goodness of Linear Fit for Different Anchor Addition Schemes}		
		\label{Tab2_R_squared}
		\centering 
		\begin{tabular}{|c | c | c | c | c | c | c|}
			\hline
			\multirow{2}{*}{$\bm{R^2}$} & \multicolumn{2}{c |}{\textbf{Minimax RNDOP}} & \multicolumn{2}{c |} {\textbf{Trace-based}} &  \multicolumn{2}{c|}{\textbf{Eigvec-based}} \\
			\cline{2-7}
			& \textbf{3D} & \textbf{2D} & \textbf{3D} & \textbf{2D} & \textbf{3D} & \textbf{2D} \\
			\hline
			$R^2$ of $t^{(10)}_{\rm exec}$ & 0.993 & 0.978 & 0.999 & 0.997 & 0.884 & 0.997 \\
			\hline
			$R^2$ of $t^{(50)}_{\rm exec}$ & 0.994 & 0.994 & 0.999 & 0.999 & 0.967 & 0.988 \\
			\hline
			$R^2$ of $t^{(90)}_{\rm exec}$ & 0.999 & 0.995 & 1 & 0.998 & 0.977 & 0.828 \\
			\hline 
		\end{tabular} 
	\end{table}
	\vspace{-10pt}

	\subsection{Impact of Anchor Position Uncertainty}
	To validate our theoretical results in Lemma \ref{lemma:delta_RNDOP_pert_3D_2D}, we consider beyond convex hull 3D and 2D localization scenarios. We perform $10^4$ Monte-Carlo simulation trials, where the anchor position uncertainty $\bm{\Delta r_i} = [\Delta x_i, \Delta y_i, \Delta z_i]^T$ is distributed as $\Delta x_i \sim {\rm U}[-1 \text{m},1\text{m}], \Delta y_i \sim {\rm U}[-1 \text{m},1\text{m}]$, and $\Delta z_i \sim {\rm U}[-1 \text{m},1\text{m}]$ for $i=1,2,\cdots,N$. 
	
	Firstly, we observe that the distribution obtained from theory (circular markers) match that obtained through simulations (dashed curves). We also observe that this uncertainty results in a minimax RNDOP variation within $\pm 1\%$ for $>90\%$ of the scenarios for both 3D as well as 2D positioning scenarios. It is worthwhile to note that the distribution of the minimax RNDOP variation is symmetric around $0$, signifying that the probability of a slight reduction/increase in $\mathsf{R}^{(+)}$ due to anchor position uncertainty are equally likely. 

    \vspace{-10pt}
	\subsection{Comparison of Algorithm Execution Time}
	We characterize the statistics of the \emph{execution time} ($t_{\rm exec}$) for each iterative scheme, across $N_{\rm mc,algo} = 500$ MC trials for each scheme. Let $\bm{t_{\rm exec}} = [t_{\rm exec, 1}, t_{\rm exec, 2}, \cdots, t_{\rm exec, N_{{\rm mc}, algo}}]$ be the vector containing the execution times of a scheme and $t^{(\rho)}_{{\rm exec}}$ its $\rho^{\rm th}$ percentile value ($0 \leq \rho \leq 100$). Then, we compute $t^{(10)}_{{\rm exec}}, t^{(50)}_{{\rm exec}}$, and $t^{(90)}_{{\rm exec}}$ as a function of the number of additional anchors ($N_{\rm a}$) for each scheme. 

	Table \ref{Tab1_median_compute_time} shows the median computation time ($t^{(50)}_{{\rm exec}}$) for each scheme across different hardware/software configurations. For both 3D and 2D localization scenarios, we observe that in terms of execution time performance (lower value desired), the ranking is Eigenvector-based > Minimax RNDOP > Trace-based. Hence, the minimax RNDOP scheme offers a compromise between localization error and computation performance for both 2D and 3D-optimal anchor placement problems. However, the Eigenvector-based scheme executes $\sim 10^3$ times faster than the other schemes at the cost of marginally higher positioning error, making it attractive for scenarios where timely deployment supersedes positioning error performance. 

	Fig. \ref{Fig6_runtimes_prctile_plots} shows the variation of $t^{(10)}_{{\rm exec}}, t^{(50)}_{{\rm exec}}$, and $t^{(90)}_{{\rm exec}}$ as a function of $N_{\rm a}$ for each iterative scheme. From Fig. \ref{Fig6a_runtimes_errorbars_3d} and Fig \ref{Fig6b_runtimes_errorbars_2d}, we observe that $t^{(10)}_{{\rm exec}}, t^{(50)}_{{\rm exec}}$, and $t^{(90)}_{{\rm exec}}$ seems to scale almost linearly with $N_{\rm a}$ for the Minimax RNDOP and Trace-based schemes. The goodness of linear fit ($R^2$) values shown in Table \ref{Tab2_R_squared} confirms our observations regarding linear scaling with $N_{\rm a}$, since $R^2 > 0.975$ for the Minimax RNDOP and Trace-based schemes. This is in agreement with the computational complexity results from Section \ref{subsec:compute_complexity}. On the other hand, $t_{\rm exec}$ for the Eigenvector-based scheme has a \emph{slightly weaker linear trend} ($0.8 < R^2 < 0.997$), which can be attributed to the perturbation step in Algorithm \ref{algo:random_prtrb_min_anch_sep}. 

    \vspace{-10pt}
	\section{Conclusion} \label{sec:conclusion}
	In this paper, we proposed a novel metric termed \textit{'Range-Normalized Dilution of Precision (RNDOP)'} to characterize the positioning error in ToA-based localization scenarios where the target lies outside the anchors' convex hull. We derived the upper and lower bounds of RNDOP in 2D/3D localization scenarios \emph{in 3D space}, and used them to formulate a robust anchor placement problem to \emph{minimize the worst-case positioning error}. Since the imposition of the anchor separation and volumetric constraints make the resulting problem intractable, we proposed an iterative anchor addition algorithm based on three heuristic approaches which seek to minimize the worst-case 2D/3D positioning error: the RNDOP-based, trace-based, and eigenvector-based methods, each operating on a matrix function of the anchor locations. We also characterized iterative upper and lower bounds on the achievable RNDOP at each iteration, as well as the universal lower bound. 
	
	Then, we designed an algorithmic framework that seeks to find the \emph{worst case error-optimal} locations of $N_{\rm a}$ additional anchors, given an initial set of anchors, and characterized the computational complexity of these methods. We also derived theoretical results to characterize the impact of anchor position uncertainty on the maximum RNDOP. Finally, we validated our theoretical results and used Monte-Carlo simulations to compare the positioning error and execution time performance of the proposed schemes across different hardware and software configurations. We also discussed the performance tradeoffs, and provided useful insights regarding the trace-based heuristic's superior localization error performance, and scaling of execution time with $N_{\rm a}$. 
	
	One of the main assumptions of this paper is that the range error statistics is the same on all anchors, which is often accurate for LoS and NLoS (with homogeneous multipath) environments. The two aspects which this work did not consider is (a) non-homogeneous multipath that often results in different range error statistics at each anchor, and (b) outage of anchors due to blockage etc. The modification of the RNDOP-based formulation to account for the aforementioned factors is a compelling extension of this work that can lead to fundamental insights regarding the large-scale trends of positioning error for beyond convex hull localization scenarios. Also, as highlighted in Remark \ref{rem:minimize_anchs_prblm}, another worthwhile direction to pursue is the design of anchor placement schemes where the objective is to minimize the number of participating anchors under robust localization performance constraints, to reduce the deployment cost and system design complexity.
 
	\vspace{-5pt}
	\appendix 

	\subsection{Proof of Corollary \ref{corollary:dop_universal_bounds_2d}}\label{appendix:a}
	Using the definition of $\bm{H}^T \bm{H}$ using (\ref{eq:matrix_H}), we can write it as the block matrix
	\begin{align*}
		\bm{H}^T \bm{H} & = \begin{pmatrix}
			\bm{S} & \bm{\beta}^T \\
			\bm{\beta} & \alpha 
		\end{pmatrix}, \text{ where } \bm{\beta} = \sum_{i=1}^N \begin{pmatrix}
			\frac{(x_t - x_i)(z_t - z_i)}{\| \bm{r_t} - \bm{r_i} \|^2_2} \\
			\frac{(y_t - y_i)(z_t - z_i)}{\| \bm{r_t} - \bm{r_i} \|^2_2} \\
		\end{pmatrix},	\\
		\bm{S} & = \sum_{i=1}^N \begin{pmatrix}
			\tfrac{(x_t - x_i)^2}{\| \bm{r_t} - \bm{r_i} \|^2_2} & \frac{(x_t - x_i)(y_t - y_i)}{\| \bm{r_t} - \bm{r_i} \|^2_2} \\
			\frac{(y_t - y_i)(x_t - x_i)}{\| \bm{r_t} - \bm{r_i} \|^2_2} & \frac{(y_t - y_i)^2}{\| \bm{r_t} - \bm{r_i} \|^2_2} 
		\end{pmatrix}, \text{ and } \\
		\alpha & = \sum_{i=1}^N \tfrac{(z_t - z_i)^2}{\|\bm{r_t} - \bm{r_i} \|^2_2}.			
	\end{align*}
	To ensure a low DOP, the anchors are placed in a non-coplanar configuration. Hence, $\exists z_i \neq z_j$ for $i,j=1,2,\cdots,N$ and as a result, $\alpha > 0$. Hence, we can obtain $\bm{Q} = [\bm{H}^T \bm{H}]^{-1}$ in terms of its submatrices using \cite{horn2012matrix}
	\begin{align*}
		\bm{Q} = \begin{pmatrix}
			\bm{T} & -\tfrac{\bm{T \beta}}{\alpha} \\
			-\tfrac{\bm{\beta}^T}{\alpha} & \tfrac{1 + \bm{\beta}^T \bm{T \beta}}{\alpha}
		\end{pmatrix}, \text{ where } \bm{T} = \big(S - \tfrac{\bm{\beta \beta}^T}{\alpha} \big)^{-1}.
	\end{align*}
	Applying Proposition \ref{prop:Sherman_Morrison_formula} and taking the trace, we get 
	\begin{align*}
		{\rm tr}(\bm{T}) = {\rm tr}\big( \bm{S}^{-1} \big) + \tfrac{\| \bm{S}^{-1} \bm{\beta} \|^2_2}{\alpha + \bm{\beta}^T \bm{S}^{-1} \bm{\beta}}.
	\end{align*}
	It is clear that the second term is non-negative since (a) the numerator is the Euclidean norm of a vector, and (b) $\bm{S}^{-1} \succcurlyeq \bm{0}$. Hence, we get $\mathsf{D}^2_\mathsf{xy} \triangleq {\rm tr}(\bm{T}) \geq {\rm tr}(\bm{S}^{-1})$. Using Lemma \ref{lemma:asympt_bounds_DOP_xyz_xy}, we get $\mathsf{D}^{(+)}_\mathsf{xy} \geq r_t \sqrt{\lambda_{+}(\bm{U}^{-1})}$, where
	\begin{align*}
		\bm{U} = \sum_{i=1}^N \begin{pmatrix}
			x^2_i & x_i y_i \\
			y_i x_i & y^2_i
		\end{pmatrix}.
	\end{align*}
	The desired result is then obtained using a similar sequence of steps as in Lemma \ref{lemma:dop_universal_bounds_3d}, using (\ref{eq:sum_of_top_2_eigvals_optim_X2}).
	\bibliographystyle{IEEEtran}
	\bibliography{References_Vehic_Optim_Anch_2021}
	\ifCLASSOPTIONcaptionsoff
	\newpage
	\fi
	\begin{IEEEbiography}[{\includegraphics[width=1in,height=1.25in,clip,keepaspectratio]{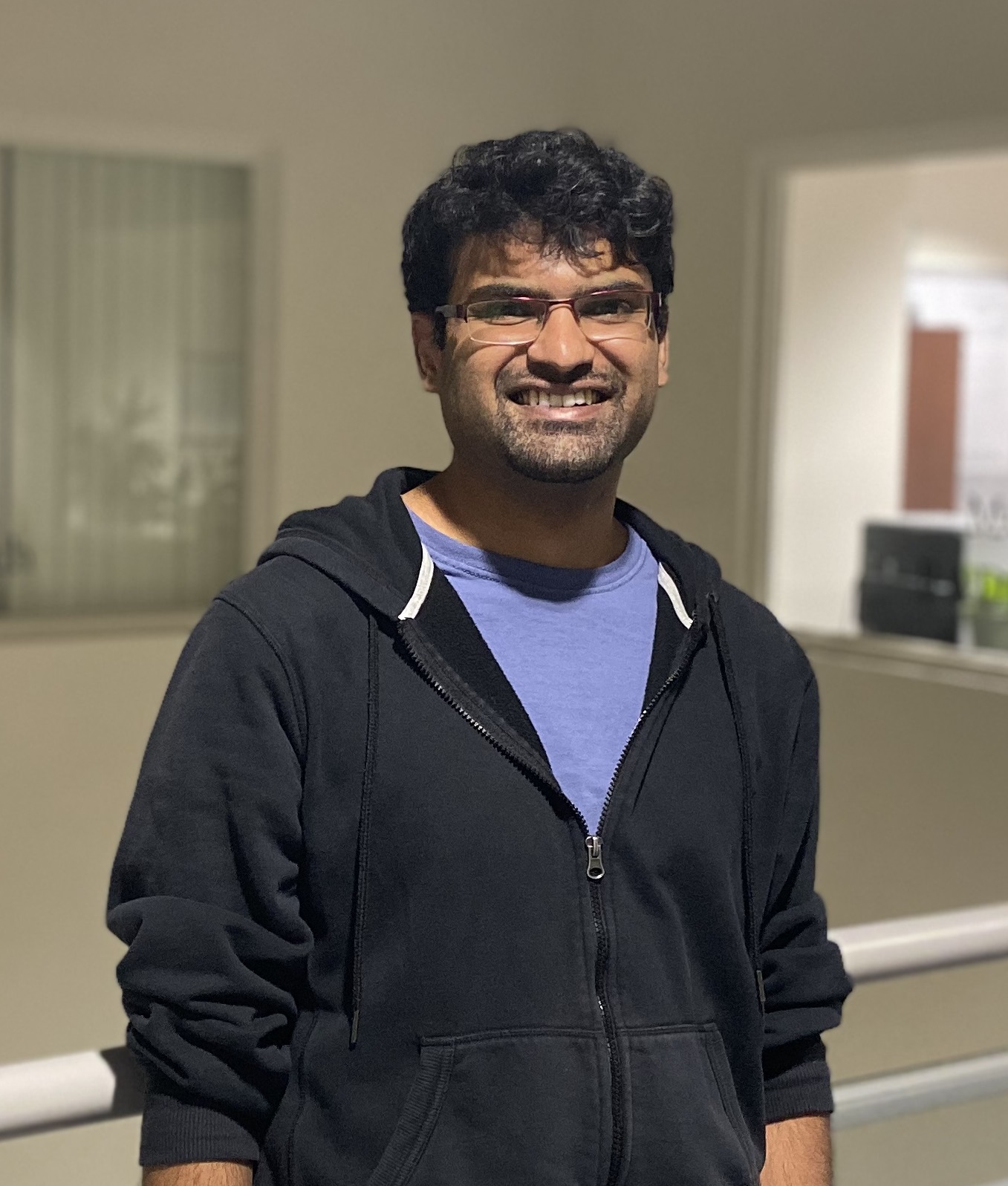}}]{Raghunandan M. Rao}(S'16-M'21) received his B.Eng degree in Telecommunication Engineering from R. V. College of Engineering in 2011, the M.Tech degree in Laser Technology from IIT Kanpur in 2013, the M.S degree in Electrical Engineering from Virginia Tech in 2016, and the Ph.D. degree in Electrical Engineering from Virginia Tech in 2020. He is currently a System Development Engineer at Amazon Lab126 in Sunnyvale, California, USA. His current research interests are in the areas of localization, radar sensing, and beyond 5G cellular systems. He served as an Editor for the {\sc IEEE Transactions on Vehicular Technology} from 2021-2023.
	\end{IEEEbiography}
	\begin{IEEEbiography}[{\includegraphics[width=1in,height=1.25in,clip,keepaspectratio]{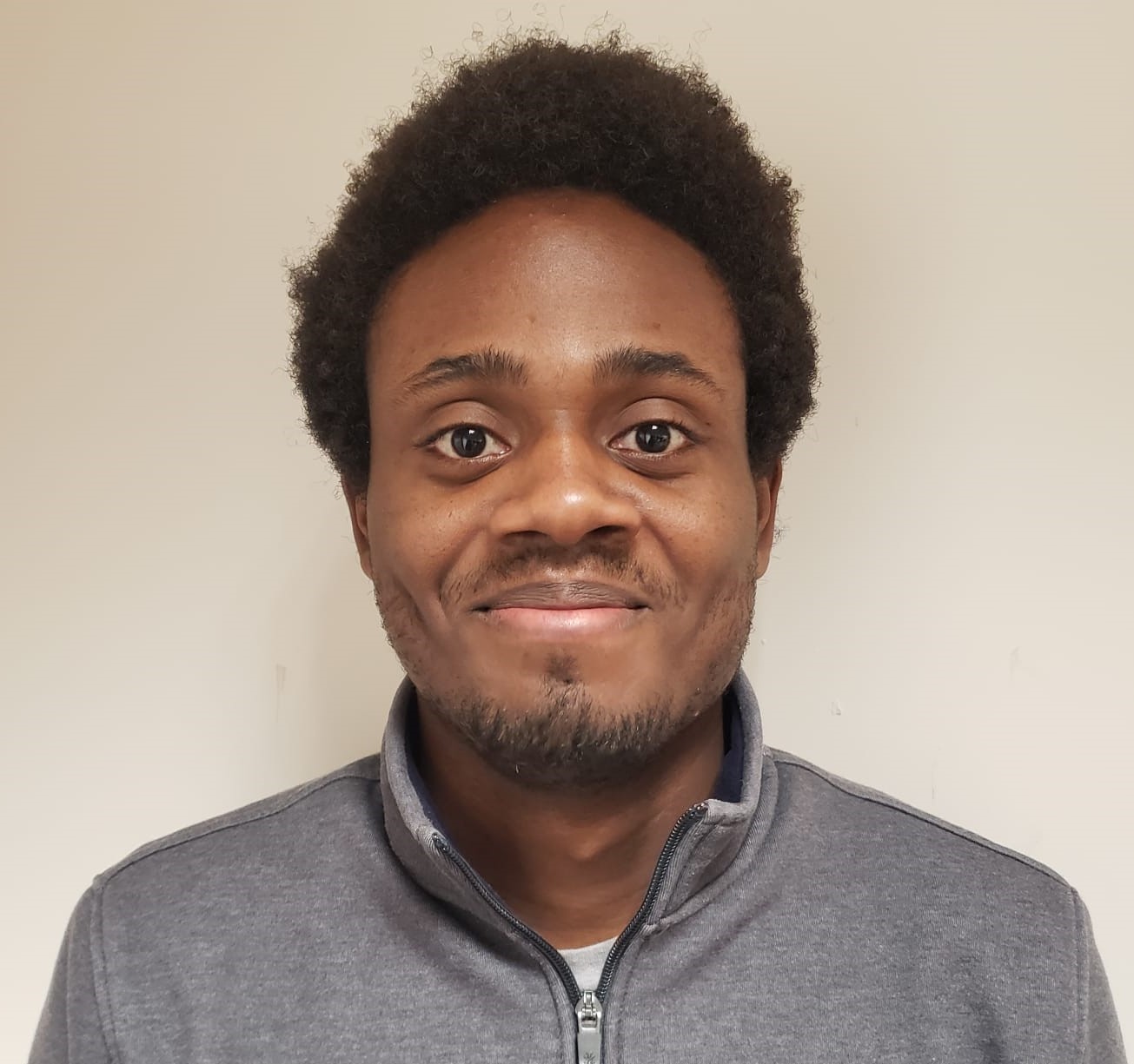}}]{Don-Roberts Emenonye}(Graduate Student Member, IEEE) received the B.Sc. degree in electrical and electronics engineering from the University of Lagos, Nigeria, in 2016, and the M.S. degree in electrical engineering from Virginia Tech, USA, in 2020, where he is currently pursuing the Ph.D. degree with the Bradley Department of Electrical and Computer Engineering. His research interests span the area of communication theory and wireless positioning.
\end{IEEEbiography}
\end{document}


	\title{Supplementary Material}
	\author{Raghunandan M. Rao,~\IEEEmembership{Member,~IEEE}, Don-Roberts Emenonye,~\IEEEmembership{Student Member,~IEEE} \\ 
		\thanks{R. M. Rao is currently at Amazon Lab 126, Sunnyvale, CA, USA, 94086. This work was undertaken before joining Amazon in his capacity as an independent scholar and does not relate to his current role at Amazon. All findings and opinions in the paper are the author's own and do not reflect the opinions of his employer (e-mail: raghumr@vt.edu).}
		\thanks{D. R. Emenonye is affiliated with the Wireless@VT research group, Bradley Department of ECE, Virginia Tech, Blacksburg, VA, USA 24061 (email: donroberts@vt.edu).}
	}
	\maketitle
	\thispagestyle{empty}
	\section{Proof of Proposition 4}
	\subsection{Proof of Eq. (5)}
	Consider the optimization problem 
	\begin{align}
		\label{eq:sum_of_top_2_eigvals_optim}
		\Lambda^*_3 = \underset{\bm{X_{3}}}{\min} & \ {\rm tr}(\bm{X}^{-1}_{\bm{3}}) - \lambda_{-}(\bm{X}^{-1}_{\bm{3}}) \\
		{\rm s.t.\ } & {\rm tr}(\bm{X_{3}}) = K, \nonumber
	\end{align}
	where $\bm{X_3} \in \mathbb{S}^{3\times 3}_{++}$. Let $\lambda_1 \leq \lambda_2 \leq \lambda_3$ be its ordered eigenvalues, where $\lambda_{+}(\bm{X_3}) = \lambda_3$ and $\lambda_{-}(\bm{X_3}) = \lambda_1$. Since $\bm{X_3}$ is positive semidefinite, we can rewrite (\ref{eq:sum_of_top_2_eigvals_optim}) in an equivalent form as 
	\begin{align}
		\label{eq:sum_of_top_2_eigvals_optim_equiv}
		\Lambda^*_3 & = \underset{\{\lambda_1, \lambda_2, \lambda_3\}}{\min} \ \sum_{i=1}^2 \lambda^{-1}_i \\
		{\rm s.t.\ } \sum_{i=1}^3 \lambda_i & = K, \nonumber \\
		 \lambda_i & \geq 0, i=1,2,3, \text{ and } \nonumber \\
		 \lambda_j & \leq \lambda_{j + 1}, j = 1,2. \nonumber
	\end{align}
	We begin by considering the Lagrangian dual problem of (\ref{eq:sum_of_top_2_eigvals_optim_equiv}), which is given by \cite{boyd2004convex}
	\begin{align}
		\label{eq:lagrnge_dual}
		g(\bm{\xi}, \nu, \bm{\lambda}) = \sum_{i=1}^2 \lambda^{-1}_i + \bm{\xi}^T \bm{g_e}(\bm{\lambda}) + \nu \big(\sum_{i=1}^3 \lambda_i - K \big),
	\end{align}
	where $\bm{g_e}(\bm{\lambda}) = [-\lambda_1,\ -\lambda_2,\ -\lambda_3,\ \lambda_1-\lambda_2,\ \lambda_2 - \lambda_3]^T$, and $\bm{\xi} = [\xi_1, \xi_2, \cdots, \xi_5]^T$. Noting that (\ref{eq:sum_of_top_2_eigvals_optim}) is a convex optimization problem, we use the KKT conditions for convex problems. In particular, the stationarity, dual feasibility, and complementary slackness conditions are given by \cite{boyd2004convex}
	\begin{align}
		\label{eq:kkt_condns_lagrnge_dual_partial}
		\bm{\nabla} f_3(\bm{\lambda}) + \bm{p_3}(\bm{\xi}) + \nu \bm{1_3} & = \bm{0}, \\
		\xi_m \geq 0, m = 1,2,\cdots, 5, & \\
		\label{eq:kkt_ineq_comp_slack}
		\xi_m [\bm{g_e}(\bm{\lambda})]_m = 0, m = 1,2,\cdots, 5, &
	\end{align}    
	where 
	\begin{align}
		\label{eq:gradient_ineq_constr}
		\bm{p_3}(\bm{\xi}) & = [-\xi_1 + \xi_4,\ -\xi_2 - \xi_4 + \xi_5 ,\ -\xi_3 - \xi_5]^T,  \nonumber \\
		\bm{\nabla} f_3(\bm{\lambda}) & = \big[-\lambda^{-2}_1,\ -\lambda^{-2}_2,\ 0 \big]^T, \text{ and} \nonumber \\
		\bm{1_3} & = [1\ 1\ 1]^T.
	\end{align}
	By inspection, we note that the strict inequality $\lambda_i > 0$ should hold for $i=1,2,3$, since $\lambda_i = 0$ for any $i$ would make the cost function in (\ref{eq:sum_of_top_2_eigvals_optim_equiv}) infinite. Hence, using this in (\ref{eq:kkt_ineq_comp_slack}), we get $\xi_i = 0$ for $i=1,2,3$. In addition, we expand (\ref{eq:kkt_condns_lagrnge_dual_partial}) to get
	\begin{align}
		\label{eq:expand_kkt_1}
		\begin{bmatrix}
			-\frac{1}{\lambda^2_1} \\
			-\frac{1}{\lambda^2_2} \\
			0
		\end{bmatrix} + 
		\begin{bmatrix}
			-\xi_1 + \xi_4 \\
			-\xi_2 - \xi_4 + \xi_5 \\
			-\xi_3 - \xi_5 
		\end{bmatrix} + 
		\begin{bmatrix}
			\nu \\
			\nu \\
			\nu 
		\end{bmatrix} = 
		\begin{bmatrix}
			0 \\
			0 \\
			0
		\end{bmatrix}.
	\end{align} 
	Using the above, we can solve for $\xi_4$ and $\xi_5$ to get
	\begin{align}
		\label{eq:kkt_gamma_4_5}
		\xi_4 & = \frac{1}{3 \lambda^2_1} + \frac{1}{3 \lambda^2_2}, \text{ and} \\
		\xi_5 & = \frac{2}{\lambda^2_1} - \frac{1}{\lambda^2_2}. 
	\end{align}	
	From the above, it is clear that $\xi_4 > 0$ since $\lambda_i > 0$ for all $i=1,2,3$. Therefore, the complementary slackness conditions in (\ref{eq:kkt_ineq_comp_slack}) yields $\lambda_1 = \lambda_2$. Furthermore, since $\lambda_1 \leq \lambda_2$ by design, we can prove that $\frac{1}{\lambda^2_1} \leq \xi_5 \leq \frac{1}{\lambda^2_2}$. In other words, $\xi_5 \neq 0$, which implies that $\lambda_2 = \lambda_3$ due to (\ref{eq:kkt_ineq_comp_slack}). Hence, equal eigenvalues yield the minimum cost function. 
	
	Finally, since $\sum_{i=1}^{3}\lambda_i = K$, we get $\lambda_i = \frac{K}{3}$ for $i = 1,2,3$. Substituting this in (\ref{eq:sum_of_top_2_eigvals_optim_equiv}) and solving for the optimal solution, we get $\Lambda^*_3 = \frac{6}{K}$. Hence, proved.
	
	\subsection{Proof of Eq. (6)}
	Consider the optimization problem
	\begin{align}
		\Lambda^*_2 = \underset{\bm{X_{2}}}{\min} & \ {\rm tr}(\bm{X}^{-1}_{\bm{2}}) - \lambda_{-}(\bm{X}^{-1}_{\bm{2}}) \\
		{\rm s.t.\ } & {\rm tr}(\bm{X_{2}}) = K, \nonumber
	\end{align}
	where $\bm{X_2} \in \mathbb{S}^{3\times 3}_{++}$. Let $\lambda_1 \leq \lambda_2$ be its ordered eigenvalues, where $\lambda_{+}(\bm{X_2}) = \lambda_2$ and $\lambda_{-}(\bm{X_2}) = \lambda_1$. Since $\bm{X_2}$ is positive semidefinite, we can rewrite (\ref{eq:sum_of_top_2_eigvals_optim}) in an equivalent form as 
	\begin{align}
		\label{eq:sum_of_max_eigval_optim_equiv_X_2}
		\Lambda^*_2 & = \underset{\{\lambda_1, \lambda_2\}}{\min} \ \lambda^{-1}_1 \\
		{\rm s.t.\ } \lambda_1 + \lambda_2 & = K, \nonumber \\
		\lambda_i & \geq 0, i=1,2, \text{ and } \nonumber \\
		\lambda_1 & \leq \lambda_2. \nonumber
	\end{align}
	The Lagrangian of the above problem is given by \cite{boyd2004convex}
	\begin{align} 
		\label{eq:lagr_X_2_optim}
		h(\bm{\xi}, \nu, \bm{\lambda}) = \lambda^{-1}_1 + \bm{\xi}^T \bm{h_e}(\bm{\lambda}) + \nu \bm{1},
	\end{align}
	where $\bm{h_e}(\bm{\lambda}) = [-\lambda_1, \ -\lambda_2, \ \lambda_1 - \lambda_2]^T$, and $\bm{\xi} = [\xi_1,\ \xi_2, \ \xi_3]^T$. Similar to (\ref{eq:lagrnge_dual}), the KKT conditions are given by \cite{boyd2004convex}
	\begin{align}
		\label{eq:kkt_conditions_X_2}
		\bm{\nabla} f_2(\bm{\lambda}) + \bm{p_2}(\bm{\xi}) + \nu \bm{1_2} & = \bm{0}, \\
		\xi_m \geq 0, m = 1,2,3 & \\
		\label{eq:kkt_ineq_comp_slack_X_2}
		\xi_m [\bm{h_e}(\bm{\lambda})]_m = 0, m = 1,2,3, &
	\end{align}
	where
	\begin{align}
		\label{eq:gradient_ineq_constr_X_2}
		\bm{p_2}(\bm{\xi}) & = [-\xi_1 + \xi_3,\ -\xi_2 - \xi_3]^T,  \nonumber \\
		\bm{\nabla} f_2(\bm{\lambda}) & = \big[-\lambda^{-2}_1,\ 0 \big]^T, \text{ and} \nonumber \\
		\bm{1_2} & = [1\ 1]^T.
	\end{align}
	Similar to (\ref{eq:gradient_ineq_constr}), the strict inequality $\lambda_i > 0$ should hold for $i=1,2,$ since $\lambda_i = 0$ will make the objective function in (\ref{eq:sum_of_max_eigval_optim_equiv_X_2}) infinite. Hence, we get $\xi_1 = \xi_2 = 0$ in (\ref{eq:kkt_ineq_comp_slack_X_2}). Expanding (\ref{eq:gradient_ineq_constr_X_2}, we can write 
	\begin{align}
		\label{eq:kkt_after_simplify_X_2}
		\begin{bmatrix}
			-\frac{1}{\lambda^2_1} \\
			0
		\end{bmatrix} + 
		\begin{bmatrix}
			-\xi_1 + \xi_3 \\
			-\xi_2 - \xi_3
		\end{bmatrix} + 
		\begin{bmatrix}
			\nu \\
			\nu  
		\end{bmatrix} = 
		\begin{bmatrix}
			0 \\
			0 
		\end{bmatrix}.
	\end{align}
	Solving this system of equations, we get $\nu = \xi_3 = \frac{1}{2 \lambda^2_1}$. Hence, $\xi_3 \neq 0$ since $\lambda_i > 0$ (this is because $\lambda_1 = 0 \Rightarrow \lambda^{-1}_1 = +\infty$, resulting in an absurd value for the cost function). Therefore, from (\ref{eq:kkt_ineq_comp_slack_X_2}), we get $\lambda_1 = \lambda_2$. Substituting this in the equality constraint in (\ref{eq:sum_of_max_eigval_optim_equiv_X_2}), we get $\lambda_1 = \lambda_2 = \frac{K}{2}$, resulting in the optimal value $\Lambda^*_2 = \frac{2}{K}$. Hence, proved.
	
	\bibliographystyle{IEEEtran}
	\bibliography{supp_mat_bib}
	\ifCLASSOPTIONcaptionsoff
	\newpage
	\fi